%% file: main.tex
\documentclass[12pt]{article}
\usepackage{amsmath,amsfonts,amsthm,amssymb,color}
\usepackage{fancybox}
\usepackage{float}
\usepackage{hyperref}
\usepackage{mathabx}

\usepackage{tikz}
\usepgflibrary{arrows}
\usetikzlibrary{arrows.meta,automata,decorations.pathreplacing}

\usepackage{thm-restate}

\newtheorem{theorem}{Theorem}[section]
\newtheorem{corollary}[theorem]{Corollary}
\newtheorem{lemma}[theorem]{Lemma}
\newtheorem{observation}[theorem]{Observation}

\newtheorem{claim}[theorem]{Claim}
\newtheorem{fact}[theorem]{Fact}

\theoremstyle{definition}
\newtheorem{definition}[theorem]{Definition}

\newenvironment{fminipage}
{\begin{Sbox}\begin{minipage}}
		{\end{minipage}\end{Sbox}\fbox{\TheSbox}}

\newenvironment{algbox}[0]{\vskip 0.2in
	\noindent 
	\begin{fminipage}{6.3in}
	}{
	\end{fminipage}
	\vskip 0.2in
}

\def\prob#1#2{\mbox{Pr}_{#1}\left[ #2 \right]}

\def\defeq{\stackrel{\mathrm{def}}{=}}

\def\abs#1{\left|#1  \right|}

\def\norm#1{\left\| #1 \right\|}
\def\normi#1{\left\vvvert  #1 \right\vvvert }

\newcommand\Ahat{\widehat{\mathit{A}}}
\newcommand\Atil{\widetilde{\mathit{A}}}
\newcommand\Cbar{\overline{\mathit{C}}}

\newcommand\Khat{\widehat{\mathit{K}}}

\newcommand\Mhat{\widehat{\mathit{M}}}
\newcommand\Mtil{\widetilde{\mathit{M}}}
\newcommand\Otil{\widetilde{O}}

\newcommand\Shat{\widehat{\mathit{S}}}
\newcommand\Xhat{\widehat{\mathit{X}}}
\newcommand\Yhat{\widehat{\mathit{Y}}}
\newcommand\Zbar{\overline{\mathit{Z}}}
\newcommand\Ztil{\widetilde{\mathit{Z}}}

\newcommand\bhat{\widehat{\mathit{b}}}

\newcommand\qhat{\widehat{\mathit{q}}}
\newcommand\rhat{\widehat{\mathit{r}}}
\newcommand\that{\widehat{\mathit{t}}}

\newcommand\what{\widehat{\mathit{w}}}
\newcommand\xhat{\widehat{\mathit{x}}}
\newcommand\yhat{\widehat{\mathit{y}}}
\newcommand\zhat{\widehat{\mathit{z}}}

\newcommand\epsilonhat{\widehat{\mathit{\epsilon}}}
\newcommand\deltahat{\widehat{\mathit{\delta}}}

\newcommand\normal{\mathcal{N}}

\newcommand\Qcal{\mathcal{Q}}
\newcommand\Wcal{\mathcal{W}}
\newcommand\Ycal{\mathcal{Y}}

\newcommand{\diag}[1]{\mathrm{DIAG} \left(#1\right)}

\begin{document}
	
	\title{Solving Sparse Linear Systems Faster\\ 
		than Matrix Multiplication}
	
	\author{Richard Peng\\
		Georgia Tech\\
		\texttt{rpeng@cc.gatech.edu}
		\and
		Santosh Vempala\\
		Georgia Tech\\
		\texttt{vempala@gatech.edu}
	}
	
	\maketitle
	
	\begin{abstract}
		Can linear systems be solved faster than matrix multiplication?
		While there has been remarkable progress for the special cases of graph structured linear systems, in the general setting, the bit complexity of solving an $n \times n$ linear system $Ax=b$ is $\tilde{O}(n^\omega)$,
		where $\omega < 2.372864$ is the matrix multiplication exponent.
		Improving on this has been an open problem even for sparse linear systems with poly$(n)$ condition number.
		
		In this paper, we present an algorithm that solves linear systems in sparse matrices asymptotically faster than
		matrix multiplication for any $\omega > 2$.
		This speedup holds for any input matrix $A$ with $o(n^{\omega -1}/\log(\kappa(A)))$ non-zeros, where $\kappa(A)$ is the condition number of $A$.
		For poly$(n)$-conditioned matrices with $\tilde{O}(n)$ nonzeros,
		and the current value of $\omega$, the bit complexity of our algorithm to solve to within any $1/\text{poly}(n)$ error is $O(n^{2.331645})$.
		
		Our algorithm can be viewed as an efficient, randomized implementation of the block Krylov method via recursive low displacement rank factorizations. It is inspired by the algorithm of [Eberly et al. ISSAC `06 `07]
		for inverting matrices over finite fields.
		In our analysis of numerical stability, we develop matrix anti-concentration techniques to bound the smallest eigenvalue and the smallest gap in eigenvalues of semi-random matrices.
	\end{abstract}
	
	\newpage

	\tableofcontents

	\pagebreak

	\section{Introduction}
	Solving a linear system $Ax=b$ is a basic algorithmic problem with direct applications to scientific computing, engineering, and physics, and is at the core of algorithms for many other problems, including optimization~\cite{Ye11:book},
	data science~\cite{BlumHK20:book},
	and computational geometry~\cite{EdelsbrunnerH10:book}. 
	It has enjoyed an array of elegant approaches, from Cramer's rule and Gaussian elimination to numerically stable iterative methods to more modern randomized variants based on random sampling~\cite{SpielmanT11, KoutisMP12} and sketching~\cite{DrineasMM08,Woodruff14:book}.
	Despite much recent progress on faster solvers for graph-structured linear systems~\cite{Vaidya89,Gremban96:thesis,SpielmanTengSolver:journal,KoutisMP12,Kyng17:thesis}, progress on the general case has been elusive.
	
	Most of the work in obtaining better running time bounds
	for linear systems solvers has focused on efficiently computing
	the inverse of $A$, or some factorization of it.
	Such operations are in turn closely related to the cost of
	matrix multiplication.
	Matrix inversion can be reduced to matrix multiplication via
	divide-and-conquer, and this reduction
	was shown to be stable when the word size for representing numbers\footnote{
		We will be measuring bit-complexity under fixed-point arithmetic.
		Here the machine word size is on the order of the
		maximum number of digits of precision in $A$,
		and the total cost is measured by the number of word operations.
		The need to account for bit-complexity of the numbers
		naturally led to the notion of condition number~\cite{Turing48,Blum04}.
		The logarithm of the condition number measures the additional number of
		words needed to store $A^{-1}$ (and thus $A^{-1} b$) compared to $A$.
		In particular, matrices with $poly(n)$ condition number can be stored
		with a constant factor overhead in precision, and are numerically
		stable under standard floating point number representations.} is increased by
	a factor of $O(\log{n})$~\cite{DemmelDHK07}.
	The current best runtime of $O(n^{\omega})$ with
	$\omega < 2.372864 $~\cite{Legall14} follows a long line of work on faster matrix multiplication algorithms~\cite{Strassen69,Pan84:book,CoppersmithW87,Williams12,Legall14}
	and is also the current best running time for solving $Ax = b$:
	when the input matrix/vector are integers, matrix multiplication based algorithms can obtain the exact rational value solution using $O(n^{\omega})$ word operations~\cite{Dixon82,Storjohann05}.
	
	Methods for matrix inversion or factorization are often referred
	to as direct methods in the linear systems literature~\cite{DavisRS16:survey}.
	This is in contrast to iterative methods, which gradually converge to
	the solution.
	Iterative methods have little space overhead, and therefore
	are widely used for solving large, sparse,
	linear systems that arise in scientific computing. Another reason for their popularity is that they are naturally suited to producing approximate solutions of desired accuracy in floating point arithmetic, the de facto method for representing real numbers.
	Perhaps the most famous iterative method is the Conjugate Gradient (CG)
	/ Lanczos algorithm~\cite{HestenesS52,Lanczos50}. 
	It was introduced as an $O(n \cdot nnz)$ time algorithm under exact
	arithmetic, where $nnz$ is the number of non-zeros in the matrix.
	However, this bound only holds under the Real RAM model
	where the words have with unbounded precision~\cite{PreparataS85:book,BlumSS89}. When taking bit sizes into account, it incurs an additional factor of $n$. Despite much progress in iterative techniques in the intervening decades,
	obtaining analogous gains over matrix multiplication in
	the presence of round-off errors has remained an open question.
	
	The convergence and stability of iterative methods typically depends on
	some {\em condition number} of the input. 
	When all intermediate steps are carried out to precision close
	to the condition number of $A$, the running time bounds
	of the conjugate gradient algorithm, as well as other currently
	known iterative methods, depend polynomially
	on the condition number of the input matrix $A$.
	Formally, the condition number of a symmetric matrix $A$,
	$\kappa(A)$, is the ratio between
	the maximum and minimum eigenvalues of $A$.
	Here the best known rate of convergence when all intermediate operations
	are restricted to bit-complexity $O(\log(\kappa(A)))$ is to an error of $\epsilon$
	in $O(\sqrt{\kappa(A)} \log(1 / \epsilon))$ iterations.
	This is known to be tight if one restricts to matrix-vector multiplications
	in the intermediate steps~\cite{SachdevaN13,MuscoMS18}.
	This means for moderately conditioned (e.g. with $\kappa = poly(n)$), sparse, systems,
	the best runtime bounds are still via direct methods,
	which are stable when $O(\log(1  / \kappa))$ words of precision are
	maintained in intermediate steps~\cite{DemmelDHK07}.
	
	Many of the algorithms used in scientific computing for solving linear systems
	involving large, space, matrices are based on combining direct and iterative
	methods: we will briefly discuss this perspectives in Section~\ref{subsec:Related}.
	From the asymptotic complexity perspective, the practical successes
	of many such methods naturally leads to the question of whether one can
	provably do better than the $O(\min\{n^{\omega}, nnz \cdot \sqrt{\kappa(A)}\})$
	time corresponding to the faster of direct or iterative methods.
	Somewhat surprisingly, despite the central role of this question in scientific
	computing and numerical analysis, as well as extensive studies of linear systems solvers,
	progress on this question has been elusive.
	The continued lack of progress on this question has led to its
	use as a hardness assumption for showing conditional lower bounds
	for numerical primitives such as linear elasticity problems~\cite{KyngZ17}
	and positive linear programs~\cite{KyngWZ20}.
	One formalization of such hardness is the
	\emph{Sparse Linear Equation Time Hypothesis} (\textsc{SLTH}) from~\cite{KyngWZ20}: $\textsc{SLTH}_{k}^{\gamma}$ denotes the
	assumption that a sparse linear system with $\kappa \leq nnz(A)^{k}$
	cannot be solved in time faster than $nnz(A)^{\gamma}$ to within relative error $\epsilon = n^{-10 k}$.
	Here improving over the smaller running time of both direct
	and iterative methods can be succinctly encapsulated
	as refuting $\textsc{SLTH}_{k}^{\min\{1 + k/2, \omega\}}$.
	\footnote{
		The hardness results in~\cite{KyngWZ20} were based
		on $\textsc{SLTH}_{1.5}^{1.99}$ under the Real RAM model
		in part due to the uncertain status of conjugate gradient
		in different models of computation.}
	
	In this paper, we provide a faster algorithm for solving sparse linear systems.
	Our formal result is the following (we use the form defined in \cite{KyngWZ20} 
	[Linear Equation Approximation Problem, \textsc{LEA}]).
	\begin{theorem}
		\label{thm:Main}
		Given a matrix ${A}$ with max dimension $n$,
		$nnz({A})$ non-zeros (whose values fit into a single word),
		along with a parameter $\kappa({A})$ such that
		$\kappa({A}) \ge \sigma_{\max}({A}) / \sigma_{\min}({A})$,
		along with a vector $b$ and error requirement $\epsilon$,
		we can compute, under fixed point arithmetic, in time
		\[
		O\left(\max \left\{ nnz(A)^{\frac{\omega-2}{\omega-1}}n^2, n^{\frac{5\omega-4}{\omega+1}} \right\}
		\log^2\left( \kappa / \epsilon \right) \right)
		\]
		a vector $x$ such that
		\[
		\norm{Ax - \Pi_{A} b}_{2}^2
		\leq 
		\epsilon \norm{\Pi_{A} b}_2^2,
		\]
		where $c$ is a fixed constant and $\Pi_{A}$ is the projection operator onto the
		column space of $A$.
	\end{theorem}
	Note that $\norm{\Pi_{A} b}_2
	= \norm{A^Tb}_{(A^TA)^{-1}}$,
	and when $A$ is square and full rank,
	it is just $\norm{b}_2$.
	
	The cross-over point for the two bounds is at
	$nnz(A) = n^{\frac{3 (\omega - 1)}{\omega + 1}}$.
	In particular, for the sparse case with $nnz(A) = O(n)$,
	and the current best $\omega \leq 2.372864$~\cite{Legall14},
	we get an exponent of
	\[
	\max\left\{ 
	2 + \frac{\omega - 2}{\omega - 1},
	\frac{5 \omega - 4}{\omega + 1}
	\right\}
	<
	\max\{
	2.271595,
	2.331645
	\}
	=
	2.331645.
	\]
	As $n \leq nnz$, this also translates to a running time of
	$O(nnz^{\frac{5\omega - 4}{\omega + 1} })$,
	which as $\frac{5\omega - 4}{\omega + 1}
	= \omega - \frac{(\omega - 2)^2}{\omega + 1}$,
	refutes $\textsc{SLTH}_k^{\omega}$ for constant
	values of $k$ and any value of $\omega > 2$.
	
	We can parameterize the asymptotic gains over
	matrix multiplication for moderately sparse instances.
	Here we use the $\Otil(\cdot)$ notation to hide lower-order
	terms, specifically $\Otil(f(n))$ denotes $O(f(n)
	\cdot \log^{c}(f(n)))$ for some absolute constant $c$.
	
	\begin{corollary}
		For any matrix $A$ with dimension at most $n$,
		$O(n^{\omega - 1 - \theta})$ non-zeros,
		and condition number $n^{O(1)}$,
		a linear system in $A$ can be solved 
		to accuracy $n^{-O(1)}$ in time $\Otil{(
			\max\{n^{\frac{5 \omega - 4}{\omega + 1}}, n^{\omega - \frac{\theta(\omega - 2)}{\omega - 1}}\} ) }$.
	\end{corollary}
	
	Here the cross-over point happens at
	$\theta = \frac{(\omega - 1)(\omega - 2)}{\omega + 1}$.
	Also, because $\frac{5 \omega - 4}{\omega + 1}
	= \omega - \frac{(\omega - 2)^2}{\omega + 1}$,
	we can also infer that for any
	$0 < \theta \leq \omega - 2$ and any $\omega > 2$,
	the runtime is $o(n^{\omega})$,
	or asymptotically faster than matrix multiplication.
	
	\subsection{Idea}
	
	At a high level, our algorithm follows the 
	block Krylov space method
	(see e.g. Chapter 6.12 of Saad~\cite{Saad03:book}).
	This method is a multi-vector extension of the conjugate gradient / Lanczos method,
	which in the single-vector setting is  known to be problematic under
	round-off errors both in theory~\cite{MuscoMS18} and in practice~\cite{GolubO89}.
	Our algorithm starts with a set of $s$ initial vectors, $B \in \Re^{n \times s}$,
	and forms a column space by multiplying these vectors by $A$ repeatedly, $m$ times.
	Formally, the block Krylov space matrix is
	\[
	K=
	\left[
	\begin{array}{c|c|c|c|c}
		B & AB & A^2B & \ldots & A^{m - 1}B
	\end{array}
	\right].
	\]
	The core idea of Krylov space methods is to efficiently
	orthogonalize this column space.
	For this space to be spanning, block Krylov space methods
	typically choose $s$ and $m$ so that $sm = n$.
	
	The conjugate gradient algorithm can be viewed as an
	efficient implementation of the case $s = 1$, $m = n$, and $B$ is set to $b$, the RHS of the input linear system.
	The block case with larger values of $s$ was studied by
	Eberly, Giesbrecht, Giorgi, Storjohann, and Villard~\cite{EberlyGGSV06,EberlyGGSV07} over finite fields, 
	and they gave an $O(n^{2.28})$
	time algorithm for computing the inverse
	of a sparse matrix over a finite field.
	
	Our algorithm also leverages the top-level insight of
	the Eberly et al. results: the Gram matrix 
	of the Krylov space matrix (which can be used inter-changeably 
	for solving linear systems) is a block Hankel matrix.
	That is, if we view the Gram matrix $(AK)^T(AK)$ as an $m$-by-$m$
	matrix containing $s$-by-$s$ sized blocks, then all the blocks
	along each anti-diagonal are the same:
	\[
	\left(AK\right)^T \left(AK\right)
	=
	\left[
	\begin{array}{c|c|c c c}
		B^TA^2B & B^TA^3B & B^TA^4B & \ldots & B^{T} A^{m + 1}B\\
		\hline
		B^TA^3B & B^TA^4B & B^TA^5B &\ldots & B^{T} A^{m + 2}B\\
		\hline
		B^TA^4B & B^TA^5B & B^TA^6B& \ldots & B^{T} A^{m + 3}B\\
		\ldots & \ldots & \ldots & \ldots & \ldots\\
		B^TA^{m + 1}B & B^TA^{m + 2}B & B^TA^{m + 3}B & \ldots & B^{T} A^{2m}B\\
	\end{array}
	\right]
	\]
	Formally, the $s$-by-$s$ inner product matrix formed from
	$A^{i}B$ and $A^{j}B$ is $B^{T}A^{i + j}B$, and only depends
	on $i + j$.
	So instead of $m^2$ blocks each of size $s \times s$,
	we are able to represent a $n$-by-$n$ matrix with about $m$ blocks.
	
	Operations involving these $m$ blocks of the Hankel matrix
	can be handled using $\Otil(m)$ block operations.
	This is perhaps easiest seen for computing matrix-vector
	products using $K$.
	If we use $\{i\}$ to denote the $i$th block of the Hankel matrix,
	that is
	\[
	H_{\left\{i, j\right\}} = M\left(i + j \right)
	\]
	for a sequence of matrices $M$, we get that the $i\textsuperscript{th}$ block
	of the product $Hx$ can be written in block-form as
	\[
	\left(Hx\right)_{\left\{i\right\}}
	=
	\sum_{j}H_{\left\{i, j\right\}} x_{\left\{j\right\}}
	=
	\sum_{j} M\left(i + j \right) x_{\left\{j\right\}}.
	\]
	Note this is precisely the convolution of (a sub-interval) of $M$
	and $x$, with shifts indicated by $i$.
	Therefore, in the forward matrix-vector multiplication direction,
	a speedup by a factor of about $m$ is possible with fast convolution algorithms.
	The performance gains of the Eberly et al. algorithms~\cite{EberlyGGSV06,EberlyGGSV07}
	can be viewed as of similar nature, albeit in the more difficult direction
	of solving linear systems.
	Specifically, they utilize algorithms for
	the Pad\'{e} problem of computing a polynomial from the result of its convolution~\cite{XuB90,BeckermannL94}.
	Over finite fields, or under exact arithmetic, such algorithms for matrix
	Pad\'{e} problems take $O(m \log{m})$ block operations~\cite{BeckermannL94},
	for a total of $\Otil(s^{\omega} m)$ operations..
	
	The overall time complexity follows from two opposing goals:
	\begin{enumerate}
		\item Quickly generate the Krylov space:
		repeated multiplication by $A$
		allows us to generate $A^{i}B$ using $O(ms \cdot nnz) = O(n \cdot nnz)$
		arithmetic operations.
		Choosing a sparse $B$ then allows us to compute $B^T A^{i} B$ in
		$O(n \cdot s)$ arithmetic operations, for a total overhead of $O(n^2) =
		O(n \cdot nnz)$.
		\item Quickly invert the Hankel matrix.
		Each operation on an $s$-by-$s$ block takes $O(s^{\omega})$ time.
		Under the optimistic assumption of $\Otil(m)$ block operations,
		the total is $\Otil(m \cdot s^{\omega})$.
	\end{enumerate}
	Under these assumptions, and the requirement of $n \approx ms$, the total
	cost becomes about $O(n \cdot nnz + m \cdot s^{\omega})$, which is
	at most $O(n \cdot nnz)$ as long as $m > n^{\frac{\omega - 2}{\omega - 1}}$.
	However, this runtime complexity is over finite fields, where numerical stability is not an issue, 
	instead of over reals under round-off errors, where one must contend with numerical errors without blowing up the bit complexity.
	This is a formidable challenge; indeed, with exact arithmetic, the CG method takes time $O(n\cdot nnz)$, but this is misleading since the computation is effective only the word sizes increase by a factor of $n$ (to about $n \log\kappa$ words), which leads to an overall complexity of $O(n^2\cdot nnz \cdot \log \kappa)$. 
	
	\subsection{Our Contributions} 
	\label{subsec:Contributions}

	Our algorithm can be viewed as the numerical generalization of the algorithms from~\cite{EberlyGGSV06,EberlyGGSV07}.
	We work with real numbers of bounded precision,
	instead of entries over a finite field.
	The core of our approach can be summarized as:
	
	\begin{center}
		\framebox{
			\parbox{13cm}{
				The block Krylov space method together with fast Hankel solvers can
				be made numerically stable using $\Otil(m \log(\kappa))$ words of precision.
			}
		}
	\end{center}
	
	Doing so, on the other hand, requires developing tools
	for two topics that have been extensively studied
	in mathematics, but separately.
	\begin{enumerate}
		\item Obtain low numerical cost solvers for block
		Hankel/Toeplitz matrices.
		Many of the prior algorithms rely on algebraic identities
		that do not generalize to the block setting, and are often
		(experimentally) numerically unstable~\cite{GallivanTVV96,Gray06:book}.
		\item Develop matrix anti-concentration bounds for analyzing the
		word lengths of inverses of random Krylov spaces.
		Such bounds upper bound the probability of random matrices being in some
		set of small measure, which in our case is the set of nearly singular matrices.
		Previously, they were known assuming
		the matrix entries are independent~\cite{SankarST03:journal,TaoV10},
		while Krylov spaces have correlated columns.
	\end{enumerate}
	Furthermore, due to the shortcomings of the matrix anti-concentration
	bounds, we modify the solver algorithm so that it uses a more limited
	version of the block-Krylov space that fall under the cases that could
	be analyzed.
	
	Before we describe the difficulties and new tools needed,
	we first provide some intuition on why a factor $m$ increase
	in word lengths may be the right answer by
	upper-bounding the magnitudes of entries in a $m$-step Krylov space.
	The maximum magnitude of $A^{m}b$ is bounded by the max magnitude
	of $A$ to the power of $m$, times a factor corresponding to the number
	of summands in the matrix product:
	\[
	\norm{A^{m}b}_{\infty}
	\leq
	\left( n \normi{A}_{\infty} \right)^{m} \norm{b}_{\infty}.
	\]
	So the largest numbers in $K$ (as well as $AK$) can be bounded
	by $(n\kappa)^{O(m)}$,
	or $O(m \log{\kappa})$ words in front of the decimal point
	under the assumption of $\kappa > n$.
	
	Should such a bound of $O(m \log{\kappa})$ hold for all numbers
	that arise, including the matrix inversions, and the matrix $B$
	is sparse with $O(n)$ entries, the
	cost of computing the block-Krylov matrices becomes
	$O( m \log{\kappa} \cdot ms \cdot nnz)$, while the cost of the
	matrix inversion portion encounters an overhead of $O(m \log{\kappa})$,
	for a total of $\Otil(m^2 s^{\omega} \log{\kappa})$.
	In the sparse case of $nnz = O(n)$, and $n \approx ms$, this becomes:
	\begin{equation}
		O\left( n^2 m \log{\kappa} + m^2 s^{\omega} \log{\kappa} \right)
		=
		O\left( n^2 m \log{\kappa} + \frac{n^{\omega}}{m^{\omega - 2}} \log{\kappa} \right).
		\label{eq:BlockKrylovGoal}
	\end{equation}
	Due to the gap between $n^2$ and $n^{\omega}$,
	setting $m$ appropriately gives improvement over $n^{\omega}$
	when $\log\kappa < n^{o(1)}$.
	
	However, the magnitude of an entry in the inverse depends on
	the smallest magnitude, or in the matrix case, its minimum
	singular value.
	Bounding and propagating the min singular value,
	which intuitively corresponds to how close a matrix
	is to being degenerate,
	represents our main challenge.
	In exact/finite fields settings,
	non-degeneracies are certified via the Schwartz-Zippel Lemma about polynomial roots.
	The numerical analog of this is more difficult:
	the Krylov space matrix $K$ is asymmetric, even for a symmetric matrix $A$.
	It is much easier for an asymmetric matrix with correlated entries
	to be close to singular.
	
	Consider for example a two-banded, two-block matrix with
	all diagonal entries set to the same random variable $\alpha$ (see Figure~\ref{fig:BadMatrix}):
	\[
	A_{ij}
	=
	\begin{cases}
		1 & \text{if $i = j$ and $j \leq n / 2$},\\
		\alpha& \text{if $i = j + 1$ and $j \leq n/2$},\\
		\alpha & \text{if $i = j + 1$ and $n / 2 < j$},\\
		2 & \text{if $i = j + 1$ and $n/2 < j$},\\
		0 & \text{otherwise}.
	\end{cases}
	\]
	
	\begin{figure}
		\begin{center}
			\begin{tikzpicture}
				\draw (0,0) -- (6,0) -- (6,6) -- (0,6) -- (0,0);
				\draw (0,5.5) -- (0,6) -- (3,3) -- (2.5,3) -- (0,5.5);
				\draw (0,5) -- (0,5.5) -- (2.5,3) -- (2,3) -- (0,5);
				\node at (0.2, 5.55) (a) {$1$};
				\node at (0.7, 5.05) (a) {$1$};
				\node at (2.05, 3.7) (a) {$1$};
				\node at (2.55, 3.2) (a) {$1$};
				\node[rotate = -45] at (1.5, 4.3) (blah) {$\ldots$};
				\node[rotate = -45] at (1.2, 4) (blah) {$\ldots$};
				\node at (0.2, 5.05) (a) {$\alpha$};
				\node at (0.7, 4.55) (a) {$\alpha$};
				\node at (2.05, 3.2) (a) {$\alpha$};
				\draw (3,2.5) -- (3,3) -- (6,0) -- (5.5,0) -- (3,2.5);
				\draw (3,2) -- (3,2.5) -- (5.5,0) -- (5,0) -- (3,2);
				\node at (3.2, 2.55) (a) {$\alpha$};
				\node at (3.7, 2.05) (a) {$\alpha$};
				\node at (5.05, 0.7) (a) {$\alpha$};
				\node at (5.55, 0.2) (a) {$\alpha$};
				\node[rotate = -45] at (4.5, 1.3) (blah) {$\ldots$};
				\node[rotate = -45] at (4.2, 1) (blah) {$\ldots$};
				\node at (3.2, 2.05) (a) {$2$};
				\node at (3.7, 1.55) (a) {$2$};
				\node at (5.05, 0.2) (a) {$2$};
				\draw [decorate,decoration={brace,amplitude=10pt}]
				(0, 6) -- (6, 6) node (curly_bracket)[black,midway, yshift =- 0.3 cm] 
				{};
				\node at (3, 6.6) {$n$};
				\draw [decorate,decoration={brace,amplitude=10pt, mirror}]
				(6, 0) -- (6, 6) node (curly_bracket)[black,midway, yshift =- 0.3 cm] 
				{};
				\node at (6.6, 2.5) {$n$};
			\end{tikzpicture}
		\end{center}
		\caption{The difference between matrix anti-concentration over
			finite fields and reals:
			a matrix that is full rank for all $\alpha \neq 0$,
			but is always ill conditioned.
		}
		\label{fig:BadMatrix}
	\end{figure}
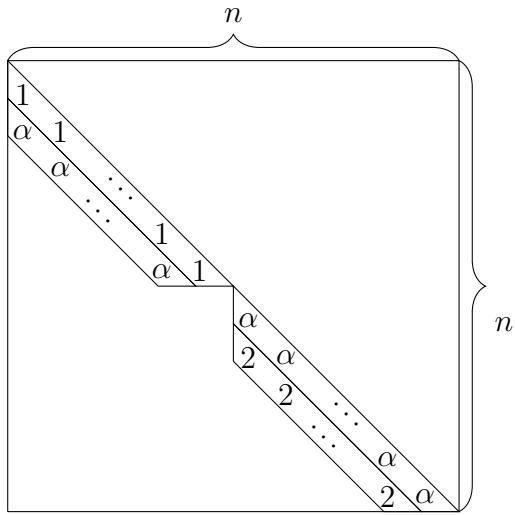
	In the exact case, this matrix is full rank unless $\alpha = 0$,
	even over finite fields.
	On the other hand, its minimum singular value is close to $0$ for
	all values of $\alpha$ because:
	\begin{observation}
		The minimum singular value of a matrix with $1$s on the diagonal,
		$\alpha$ on the entries immediately below the diagonal, and $0$
		everywhere else is at most $|\alpha|^{-(n - 1)}$, due to the test vector
		$[1; -\alpha; \alpha^2; \ldots; (-\alpha)^{n - 1}]$.
	\end{observation}
	Specifically, in the top-left block, as long as $|\alpha| > 3/2$,
	the top left block has minimum singular value at most $(2/3)^{n-1}$.
	On the other hand, rescaling the bottom-right block by $1 / \alpha$
	to get $1$s on the diagonal gives $2 / \alpha$ on the off-diagonal.
	So as long as $|\alpha| < 3/2$, this value is at least $4/3$, which in turn
	implies a minimum singular value of at most $(3/4)^{n-1}$ in the
	bottom right block.
	This means no matter what value $\alpha$ is set to,
	this matrix will always have a singular value that's
	exponentially close to $0$.
	Furthermore, the Gram matrix of this matrix also gives such a counter
	example to symmetric matrices with (non-linearly) correlated entries.
	Previous works on analyzing condition numbers of asymmetric
	matrices also encounter similar difficulties: a more detailed discussion of it
	can be found in Section 7 of Sankar et al.~\cite{SankarST03:journal}.

	In order to bound the bit complexity of all intermediate steps of
	the block Krylov algorithm by $\Otil(m)\cdot \log{\kappa}$,
	we devise a more numerically
	stable algorithm for solving block Hankel matrices, as well as
	provide a new perturbation scheme to quickly generate a
	well-conditioned block Krylov space.
	Central to both of our key components is the close connection
	between condition number and bit complexity bounds.
	
	First, we give a more numerically stable solver for block
	Hankel/Toeplitz matrices.
	Fast solvers for Hankel (and closely related Toeplitz) matrices
	have been extensively studied in numerical analysis,
	with several recent developments on more stable
	algorithms~\cite{XiaXG12}.
	However, the notion of numerical stability studied in these algorithms
	is the more practical variant where the number of bits of precision
	is fixed.
	As a result, the asymptotic behavior of the stable algorithm
	from~\cite{XiaXG12} is quadratic in the number of digits in the condition number,
	which in our case would translate to a prohibitive cost of $\Otil(m^2)$ (i.e., the overall cost would be higher than $n^\omega$).

	Instead, we combine developments in recursive block
	Gaussian elimination~\cite{DemmelDHK07,KyngLPSS16,CohenKKPPRS18}
	with the {\em low displacement rank} representation of Hankel/Toeplitz
	matrices~\cite{KailathKM79,BitmeadA80}.
	Such representations allow us to implicitly express both the Hankel matrix and its inverse by displaced versions of rank $2s$ matrices.
	This means the intermediate sizes of instances arising from recursion
	is $O(s)$ times the dimension, for a total size of $O(n\log{n})$,
	giving a total of $\Otil(n s^{\omega - 1})$ arithmetic operations
	involving words of size $\Otil(m)$.
	We provide a rigorous analysis of the accumulation of round-off errors similar to the analysis of recursive matrix multiplication
	based matrix inversion from~\cite{DemmelDHK07}.
	
	Motivated by this close connection with the condition number of
	Hankel matrices, we then try to initialize with Krylov spaces
	of low condition number.
	Here we show that a sufficiently small perturbation suffices for producing
	a well conditioned overall matrix.
	In fact, the first step of our proof, that
	a small sparse random perturbation to $A$ guarantees good separations
	between its eigenvalues is a direct combination of bounds
	on eigenvalue separation of random Gaussians~\cite{NguyenTV17}
	as well as min eigenvalue of random sparse matrices~\cite{LuhV18}.
	This separation then ensures that the powers of $A$,
	$A^{1}, A^{2}, \ldots A^{m}$, are sufficiently distinguishable
	from each other. Such considerations also come up in the
	smoothed analysis of numerical algorithms~\cite{SankarST03:journal}.
	
	The randomness of the Krylov matrix induced by the initial set of random vectors $B$ is more difficult to analyze: each column of $B$
	affects $m$ columns of the overall Krylov space matrix.
	In contrast, all existing analyses of lower bounds of singular values
	of possibly asymmetric random
	matrices~\cite{SankarST03:journal,TaoV10} rely on the randomness in
	the columns of matrices being independent.
	The dependence between columns necessitates analyzing singular
	values of random linear combinations of matrices, which we handle
	by adapting $\epsilon$-net based proofs of anti-concentration bounds. Here we encounter an additional challenge in bounding the minimum singular value of the block Krylov matrix. 
	We resolve this issue algorithmically: instead of picking a Krylov space
	that spans the entire $\Re^{n}$, we stop things short by picking
	$ms = n - \Otil(m)$
	This set of extra columns significantly simplify the proof of singular
	value lower bounds. This is similar in spirit to the analysis of minimum singular values of random matrices, which is significantly easier for non-square matrices~\cite{RudelsonV10}. 
	In the algorithm, the remaining columns are treated as a separate block that we reduce to via a Schur complement at the very end of the block elimination algorithm. Since the block is small, so is its overhead on the running time.

	\subsection{History and Related Work}
	\label{subsec:Related}
	
	Our algorithm has close connections with multiple lines of research
	on more efficient solvers for sparse linear systems.
	This topic has been extensively studied not only in computer science,
	but also in applied mathematics and engineering.
	For example, in the Editors of the Society of Industrial and
	Applied Mathematics News' `top 10 algorithms of the 20th century',
	three of them (Krylov space methods, matrix decompositions,
	and QR factorizations) are directly related to linear systems solvers~\cite{Cipra00}.

	At a high level, our algorithm is a hybrid linear systems solver.
	It combines iterative methods, namely block Krylov space methods,
	with direct methods that factorize the resulting Gram matrix of the
	Krylov space.
	Hybrid methods have their origins in the incomplete Cholesky method
	for speeding up elimination/factorization based direct solvers.
	A main goal of these methods is to reduce the $\Omega(n^2)$ space
	needed to represent matrix factorizations/inverses.
	This high space requirement is often 
	even more problematic than time when handling large sparse matrices.
	Such reductions can occur in two ways: either by directly dropping
	entries from the (intermediate) matrices, or by providing more succinct
	representations of these matrices using additional structures.
	
	The main structure of our algorithm is based on the latter
	line of work on solvers for structured matrices.
	Such systems arise from physical processes where the interactions
	between objects have invariances (e.g. either by time or space
	differences).
	Examples of such structure include
	circulant matrices~\cite{Gray06:book},
	Hankel/Toeplitz matrices~\cite{KailathKM79,BitmeadA80,XiaXG12,XiXCB14},
	and distances from $n$-body simulations~\cite{CoifmanRW93}.
	Many such algorithms require exact preservation of the
	structure in intermediate steps.
	As a result, many of these works develop algorithms over finite
	fields~\cite{BitmeadA80,BeckermannL94,BostanJMS17}.
	
	More recently, there has been work on developing more numerically
	stable variants of these algorithms for structured matrices,
	or more generally, matrices that are numerically close
	to being structured~\cite{XiaCGL10,LinLY11,XiaXG12,XiXCB14}.
	However, these results only explicitly discussed the entry-wise
	Hankel/Toeplitz case (which corresponds to $s = 1$).
	Furthermore, because they rely on domain-decomposition techniques
	similar to fast multiple methods, they produce one bit of precision
	per each outer iteration loop.
	As the Krylov space matrix has condition number $\exp(\Omega(m))$,
	such methods would lead to another factor of $m$ in the solve
	cost when directly invoked.
	
	Instead, our techniques for handling and bounding numerical errors
	are more closely related to recent developments in provably efficient
	sparse Cholesky factorizations~\cite{KyngLPSS16,KyngS16,Kyng17:thesis,CohenKKPPRS18}.
	These methods generated efficient preconditioners using only the condition
	of intermediate steps of Gaussian eliminatnion, known as Schur complements,
	having small representations.
	They avoided the explicit generation of the dense representations of Schur
	complements by treatment them as operators, and implicitly applied randomized
	tools to directly sample/sketch the final succinct representations, which have
	much smaller algorithmic costs.
	
	On the other hand, previous works on spare Choleskfy factorizations
	required the input matrix to be decomposable
	into a sum of simple elements, often through additional combinatorial
	structure of the matrices.
	In particular, this line of work on combinatorial preconditioning
	was initiated through a focus on graph Laplacians,
	which are built from $2$-by-$2$ matrix blocks corresponding
	to edges of undirected graphs~\cite{Vaidya89,Gremban96:thesis,SpielmanTengSolver:journal,KoutisMP12}.
	Since then, there has been 
	substantial generalizations to the structures amenable to such
	approaches, notably to finite element matrices~\cite{BomanHV08}
	and directed graphs/irreversible Markov chains~\cite{CohenKPPRSV17}.
	However, recent works have also shown that many classes of structures
	involving more than two variables are complete for general linear systems~\cite{Zhang18:thesis}.
	Nonetheless, the prevalence of approximation errors in such algorithms
	led to the development of new ways of bounding numerical round-off errors
	in algorithms that are critical to our elimination routine for block-Hankel
	matrices.
	
	Key to recent developments in combinatorial preconditioning is 
	matrix concentration~\cite{RudelsonV07,Tropp15}.
	Such bounds provide guarantees for (relative) eigenvalues
	of random sums of matrices.
	For generating preconditioners, such randomness arise from whether
	each element is kept, and a small condition number (which in turn
	implies a small number of outer iterations usign the preconditioners)
	corresponds to a small deviation between the original and sampled matrices.
	In contrast, we introduce randomness in order to obtain block Krylov
	spaces whose minimum eigenvalue is large.
	As a result, the matrix tool we need is anti-concentration, which
	somewhat surprisingly is far less studied.
	Previous works on it are mostly related by similar problems
	from numerical precision~\cite{SankarST03:journal,TaoV10},
	and mostly address situations where the entries in the resulting
	matrix are independent.
	Our bound on the min singular value of the random Krylov space
	can yield a crude bound for a sum of rectangluar random matrices,
	but we believe much better matrix anti-concentration bounds
	are possible.
	
	\subsection{Organization}
	\label{subsec:Organization}
	
	The rest of this paper is organized as follows:
	we present the ``outer" algorithm in Section~\ref{sec:Overview},
	and give a detailed outline of its analysis in Section~\ref{sec:Analysis}.
	A breakdown of the main components of the analysis
	is in Section~\ref{subsec:Core}: briefly,
	Sections~\ref{sec:PerturbA} and~\ref{sec:RandKrylov} bound
	the singular values of the block Krylov matrix,
	and Sections~\ref{sec:Solver} and~\ref{sec:Pad} give the
	linear systems solver with block Hankel matrices.
	Some research directions raised by this work, including possible improvements and extensions
	are discussed in Section~\ref{sec:Discussion}.
	
	\section{Algorithm}
	\label{sec:Overview}
	
	We describe the algorithm, as well as the running
	times of its main components in this section.
	To simplify discussion, we assume
	the input matrix $A$ is symmetric,
	and has $poly(n)$ condition number.
	If it is asymmetric (but invertible),
	we implicitly apply the algorithm
	to $A^T A$, using the identity
	$A^{-1} = (A^TA)^{-1} A^T$
	derived from
	$(A^TA)^{-1} = A^{-1} A^{-T}$.
	Also, recall from the discussion after Theorem~\ref{thm:Main} that we
	use $\Otil(\cdot)$ to hide lower order terms in order to simplify runtimes.
	
	Before giving details on our algorithm,
	we first discuss what constitutes
	a linear systems solver algorithm,
	specifically the equivalence between many such algorithms
	and linear operators.

	For an algorithm $\textsc{Alg}$ that takes
	a matrix $B$ as input,
	we say that $\textsc{Alg}$ is linear if 
	there is a matrix $Z_{\textsc{Alg}}$
	such that for any input $B$, we have
	\[
	\textsc{Alg}\left(B\right)
	=
	Z_{\textsc{Alg}}.
	\]
	In this section, in particular in the pseudocode in Algorithm~\ref{fig:solver},
	we use the name of the procedure, $\textsc{Solve}_{A}(b, \delta)$,
	interchangeably with the operator correpsonding to a linear algorithm that
	solves a system in $A$, on vector $b$, to error $\delta > 0$.
	In the more formal analysis, we will denote such corresponding linear
	operators using the symbol $Z$, with subscripts
	corresponding to the routine if appropriate.
	
	This operator/matrix based analysis of algorithms was
	first introduced in the analysis of
	recursive Chebyshev iteration by Spielman and Teng~\cite{SpielmanTengSolver:journal},
	with credits to the technique also attributed to Rohklin.
	It the advantage of simplifying analyses of multiple
	iterations of such algorithms, as we can directly measure
	Frobenius norm differences between such operators and the
	exact ones that they approximate.
	
	Under this correspondence, the goal of producing an algorithm
	that solves $Ax = b$ for any $b$ as input becomes equivalent
	to producing a linear operator $Z_{A}$
	that approximates $A^{-1}$, and then running it on the input $b$.
	For convenience, we also let the solver take as input a matrix 
	instead of a vector, in which case the output is the result of 
	solves against each of the columns of the input matrix. 
	
	The high-level description of our algorithm is in Figure~\ref{fig:solver}.
	To keep our algorithms as linear operators,
	we will ensure that the only approximate steps are
	from inverting matrices (where condition numbers naturally lead
	to matrix approximation errors), and in forming operators using
	fast convolution.
	We will specify explicitly in our algorithms when such
	round-off errors occur.
	
	\begin{figure}[p]
		\fbox{\parbox{\textwidth}{
				{\bf \textsc{BlockKrylov}(
					$\textsc{MatVec}_{A}( x, \delta)$:
					symmetric matrix given as implicit matrix vector muliplication access,
					$\alpha_{A}$: eigenvalue range/separation bounds for $A$ that also doubles as error threshold,
					$m$: Krylov step count
					)} 
				
				\begin{enumerate}
					
					\item (FORM KRYLOV SPACE) 
					\begin{enumerate}
						\item Set $s \leftarrow \lfloor n / m \rfloor - O(m)$,
						$h \leftarrow O(m^{2} \log(1 / \alpha_{A}))$.
						Let $G^S$ be an $n \times s$ random matrix with $G^S_{ij}$
						set to $\mathcal{N}(0, 1)$
						with probability $\frac{h}{n}$, and $0$ otherwise.
						
						\item (Implicitly)
						compute the block Krylov space
						\[
						K
						=
						\left[
						\begin{array}{c|c|c|c|c}
							G^S & A G^S & A^{2} G^S & \ldots & A^{m - 1} G^S
						\end{array}
						\right].
						\]
					\end{enumerate}
					
					\item \label{Step:InvertHankel} (SPARSE INVERSE)
					Use fast solvers for block Hankel matrices to
					obtain a solver for the matrix:
					\[
					M \leftarrow \left( A K\right)^{T} \left( AK\right),
					\]
					and in turn a solve to arbitrary error 
					which we denote $\textsc{Solve}_{M}(\cdot, \epsilon)$.
					
					\item (PAD and SOLVE) 
					\begin{enumerate}
						\item
						Let $r = n - ms$ denote the number of remaining columns.
						Generate a $n \times r$ dense Gaussian matrix $G$,
						use it to complete the basis as: $Q = [K | G]$.
						
						\item Compute the Schur complement of $(AQ)^T AQ$
						onto its last $r = n - ms$ entries
						(the ones corresponding to the columns of $G$) via the operation
						\[
						\left(AG\right)^T AG
						- \left(AG\right)^T
						\cdot
						AK
						\cdot
						\textsc{Solve}_{M}\left(\left(A K\right)^{T}AG, \alpha_{A}^{10m} \right)
						\]
						and invert this $r$-by-$r$ matrix. 
						
						\item Use the inverse of this Schur complement, as well as
						$\textsc{Solve}_{M}(\cdot, \epsilon)$ to obtain a solver
						for $Q^TQ$, $\textsc{Solve}_{Q^{T}Q}(\cdot, \epsilon)$.
					\end{enumerate}
					
					\item (SOLVE and UNRAVEL) Return the operator
					$
					Q \cdot 
					\textsc{Solve}_{(AQ)^{T}AQ} ( (AQ)^{T}x, \alpha_{A}^{10m} )
					$
					as an approximate solver for $A$.
				\end{enumerate}
		}}
		\caption{Pseudocode for block Krylov space algorithm:
			$\textsc{Solve}_{\cdot}(\cdot, \cdot)$ are operators corresponding
			to linear system solving algorithms whose formalization we discuss
			at the start of this section.
		}
		\label{fig:solver}
	\end{figure}
	
	Some of the steps of the algorithm require care for,
	efficiency, as well as tracking the number of words needed
	to represent the numbers.
	We assume the bounds on bit-complexity in the analysis
	(Section~\ref{sec:Analysis}) below, which is $\Otil(m)$
	when $\kappa = poly(n)$, and use this in the brief
	description of costs in the outline of the steps below.
	
	We start by perturbing the input matrix, resulting in a symmetric
	positive definite matrix where all eigenvalues are separated by $\alpha_A$.
	Then we explicitly form a Krylov matrix from sparse Random Gaussians:
	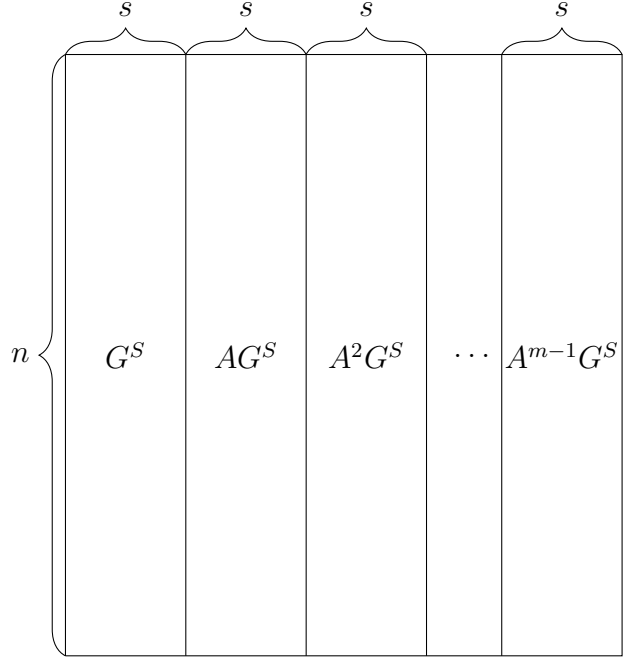
\begin{figure}
		\begin{center}
			\begin{tikzpicture}
				\draw (0,0) -- (7.4,0) -- (7.4,8) -- (0,8) -- (0,0);
				\draw (1.6,0) -- (1.6,8);
				\draw (3.2,0) -- (3.2,8);
				\draw (4.8,0) -- (4.8,8);
				\draw (5.8,0) -- (5.8,8);
				\node at (0.8, 4) (a) {$G^{S}$};
				\node at (2.4, 4) (a) {$AG^{S}$};
				\node at (4, 4) (a) {$A^{2}G^{S}$};
				\node at (6.6, 4) (a) {$A^{m-1}G^{S}$};
				\draw [decorate,decoration={brace,amplitude=10pt}]
				(0, 8) -- (1.6, 8) node (curly_bracket)[black,midway, yshift =- 0.3 cm] 
				{};
				\node at (0.8, 8.6) {$s$};
				\node at (5.4, 4) (blah) {$\ldots$};
				\draw [decorate,decoration={brace,amplitude=10pt}]
				(1.6, 8) -- (3.2, 8) node (curly_bracket)[black,midway, yshift =- 0.3 cm] 
				{};
				\node at (2.4, 8.6) {$s$};
				\draw [decorate,decoration={brace,amplitude=10pt}]
				(3.2, 8) -- (4.8, 8) node (curly_bracket)[black,midway, yshift =- 0.3 cm] 
				{};
				\node at (4, 8.6) {$s$};
				\draw [decorate,decoration={brace,amplitude=10pt}]
				(5.8, 8) -- (7.4, 8) node (curly_bracket)[black,midway, yshift =- 0.3 cm] 
				{};
				\node at (6.6, 8.6) {$s$};
				\draw [decorate,decoration={brace,amplitude=10pt}]
				(0, 0) -- (0, 8) node (curly_bracket)[black,midway, yshift =- 0.3 cm] 
				{};
				\node at (-0.6, 4) {$n$};
			\end{tikzpicture}
		\end{center}
		\caption{Randomized $m$-step Krylov Space Matrix
			with $n$-by-$s$ sparse Gaussian $G^{S}$ as starter.}
		\label{fig:KrylovMatrix}
	\end{figure}
	For any vector $u$, we can compute $A^{i} u$ from
	$A^{i - 1}u$ via a single matrix-vector multiplication in $A$.
	So computing each column of $K$ requires $O(nnz(A))$
	operations, each involving a length $n$ vector with
	words of length $\Otil(m)$.
	So we get the matrix $K$, as well as $AK$,
	in time 
	\[
	\Otil\left(nnz\left(A\right) \cdot n\cdot m\right).
	\]
	
	To obtain a solver for $AK$, we instead solve
	its Gram matrix $(AK)^T(AK)$.
	Each block of $K^TK$ has the form $(G^S)^{T} A^{i} G^S$
	for some $2 \leq i \leq 2m$,
	and can be computed by multiplying $(G^S)^{T}$ and $A^{i}G^S$.
	As $A^{i}G^{S}$ is an $n$-by-$s$ matrix,
	each non-zero in $G^{S}$ leads to a cost of $O(s)$
	operations involving words of length $\Otil(m)$.
	Then because we chose $G^{S}$ to have $\Otil(m^{3})$ non-zeros
	per column, the total number of non-zeros in $G^{S}$ is about
	$\Otil(s \cdot m^{3}) = \Otil(n m^2)$.
	This leads to a total cost (across the $m$ values of $i$) of:
	\[
	\Otil\left(n^{2} m^{3}\right).
	\]
	
	The key step is then Step~\ref{Step:InvertHankel}:
	a block version of the Conjugate Gradient method.
	It will be implemented using a recursive data structure
	based on the notion of displacement rank~\cite{KailathKM79, BitmeadA80}.
	To get a sense of why a faster algorithm may be possible,
	note that there are only $O(m)$ distinct blocks in the
	matrix $(AK)^T (AK)$.
	So a natural hope is to invert these blocks by themselves:
	the cost of (stable) matrix inversion~\cite{DemmelDH07},
	times the $\Otil(m)$ numerical word complexity, would then
	give a total of
	\[
	\Otil\left( m^2 s^{\omega} \right)
	=
	\Otil\left( m^2 \left( \frac{n}{m} \right)^{\omega} \right)
	=
	\Otil\left( n^{\omega} m^{\omega - 2} \right).
	\]
	Of course, it does not suffice to solve these $m$ $s$-by-$s$
	blocks independently.
	Instead, the full algorithm, as well as the $\textsc{Solve}_{M}$
	operator, is built from efficiently convolving such
	$s$-by-$s$ blocks with matrices using Fast Fourier Transforms.
	Such ideas can be traced back to the development of
	super-fast solvers for (entry-wise)
	Hankel/Toeplitz matrices~\cite{BitmeadA80,LabahnS92,XiXCB14}.
	
	Choosing $s$ and $m$ so that $n = sm$ would then give the
	overal running time,
	\textbf{assuming that we can bound the minimum singular value
		of $K$ by $\exp(-\Otil(m))$}.
	This is a major shortcoming of our analysis:
	we can only prove such a bound when $n - sm \geq \Omega(m)$.
	Its underlying cause is that rectangular semi-random matrices can
	be analyzed using $\epsilon$-nets, and thus are significantly
	easier to analyze than square matrices.
	
	This means we can only use $m$ and $s$ such that
	$n - ms = \Theta(m)$, and we need to pad $K$ with
	$n - ms$ columns to form a full rank, invertible, matrix.
	To this end we add $\Theta(m)$ dense Gaussian columns
	to $K$ to form $Q$,
	and solve the system $AQ$, and its associated Gram
	matrix $(AQ)^T(AQ)$ instead.
	These matrices are shown in Figure~\ref{fig:FullMatrix}.
	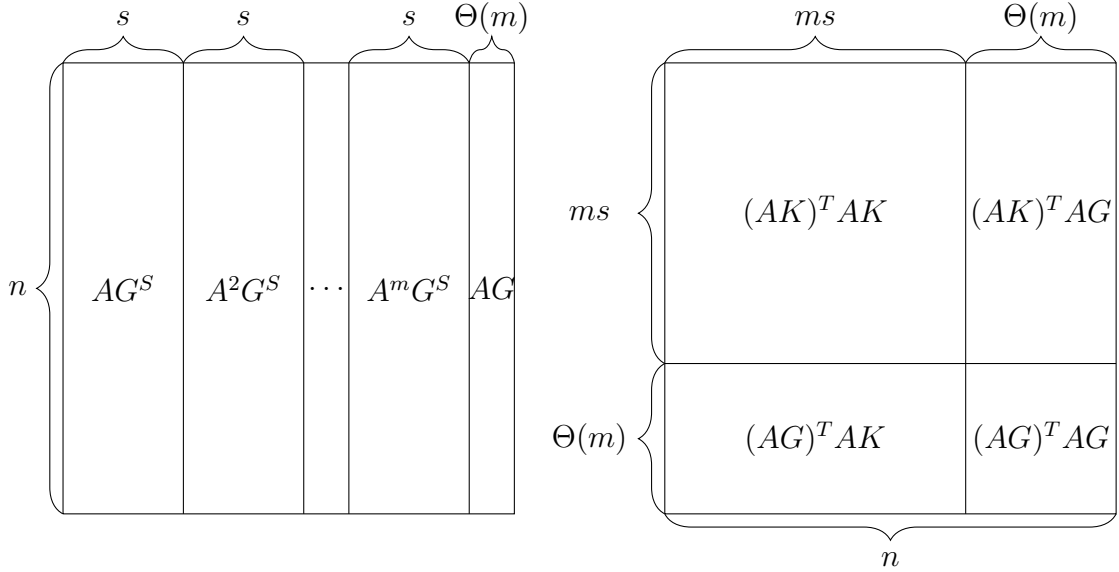
\begin{figure}
		\begin{center}
			\begin{tikzpicture}
				\draw (0,0) -- (6,0) -- (6,6) -- (0,6) -- (0,0);
				\draw (1.6,0) -- (1.6,6);
				\draw (3.2,0) -- (3.2,6);
				\draw (3.8,0) -- (3.8,6);
				\draw (5.4,0) -- (5.4,6);
				\node at (0.8, 3) (a) {$AG^{S}$};
				\node at (2.4, 3) (a) {$A^2G^{S}$};
				\node at (4.6, 3) (a) {$A^{m}G^{S}$};
				\node at (5.7, 3) (a) {$AG$};
				\draw [decorate,decoration={brace,amplitude=10pt}]
				(0, 6) -- (1.6, 6) node (curly_bracket)[black,midway, yshift =- 0.3 cm] 
				{};
				\node at (0.8, 6.6) {$s$};
				\draw [decorate,decoration={brace,amplitude=10pt}]
				(1.6, 6) -- (3.2, 6) node (curly_bracket)[black,midway, yshift =- 0.3 cm] 
				{};
				\node at (2.4, 6.6) {$s$};
				\draw [decorate,decoration={brace,amplitude=10pt}]
				(3.8, 6) -- (5.4, 6) node (curly_bracket)[black,midway, yshift =- 0.3 cm] 
				{};
				\node at (4.6, 6.6) {$s$};
				\draw [decorate,decoration={brace,amplitude=10pt}]
				(5.4, 6) -- (6, 6) node (curly_bracket)[black,midway, yshift =- 0.3 cm] 
				{};
				\node at (5.7, 6.6) {$\Theta(m)$};
				\draw [decorate,decoration={brace,amplitude=10pt}]
				(0, 0) -- (0, 6) node (curly_bracket)[black,midway, yshift =- 0.3 cm] 
				{};
				\node at (-0.6, 3) {$n$};
				
				\node at (3.5, 3) (blah) {$\ldots$};
				
				\draw (8,0) -- (8,6) -- (14,6) -- (14,0) -- (8,0);
				\node at (10,4) {$(AK)^TAK$};
				\node at (13,4) {$(AK)^TAG$};
				\node at (10,1) {$(AG)^TAK$};
				\node at (13,1) {$(AG)^TAG$};
				\draw (12,0) -- (12,6);
				\draw (8,2) -- (14,2);
				\draw [decorate,decoration={brace,amplitude=10pt}]
				(8, 2) -- (8, 6) node (curly_bracket)[black,midway, yshift =- 0.3 cm] 
				{};
				\node at (7, 4) {$ms$};
				\draw [decorate,decoration={brace,amplitude=10pt}]
				(8, 0) -- (8, 2) node (curly_bracket)[black,midway, yshift =- 0.3 cm] 
				{};
				\node at (7, 1) {$\Theta(m)$};
				\draw [decorate,decoration={brace,amplitude=10pt}]
				(8, 6) -- (12, 6) node (curly_bracket)[black,midway, yshift =- 0.3 cm] 
				{};
				\node at (10, 6.6) {$ms$};
				\draw [decorate,decoration={brace,amplitude=10pt}]
				(12, 6) -- (14, 6) node (curly_bracket)[black,midway, yshift =- 0.3 cm] 
				{};
				\node at (13, 6.6) {$\Theta(m)$};
				\draw [decorate,decoration={brace,amplitude=10pt,mirror}]
				(8, 0) -- (14, 0) node (curly_bracket)[black,midway, yshift =- 0.3 cm] 
				{};
				\node at (11, -0.6) {$n$};
			\end{tikzpicture}
		\end{center}
		\caption{Full matrix $AQ$ and its Associated Gram Matrix
			$(AQ)^T (AQ)$.
			Note that by our choice of parameters $m$ is much smaller than
			$s \approx n/m$.}
		\label{fig:FullMatrix}
	\end{figure}
	
	Because these additional columns are entry-wise i.i.d,
	its minimum singular value can be analyzed using
	existing tools~\cite{SankarST03:journal,TaoV10},
	namely lower bounding the dot product of a random vector
	against any normal vector.
	Thus, we can lower bound the minimum singular value of
	$Q$, and in turn $AQ$, by $\exp(-\Otil(m))$ as well.
	
	This bound in turn translates to the minimum
	eigenvalue of the Gram matrix of $AQ$, $(AQ)^T(AQ)$.
	Partitioning its entries by those from $K$ and $G$
	gives four blocks:
	one $(sm)$-by-$(sm)$ block
	corresponding to $(AK)^T(AK)$,
	one $\Theta(m)$-by-$\Theta(m)$ block corresponding to
	$(AG)^T(AG)$, and then the cross terms.
	To solve this matrix, we apply block-Gaussian elimination,
	or equivalently, form the Schur complement onto the
	$\Theta(m)$-by-$\Theta(m)$ corresponding to the columns
	in $AG$.
	
	To compute this Schur complement, it suffices to solve
	the top-left block (corresponding to $(AK)^T(AK)$)
	against every column in the cross term.
	As there are at most $\Theta(m) < s$ columns, this solve
	cost comes out to less than $\Otil(s^{\omega} m)$ as well.
	We are then left with a $\Theta(m)$-by-$\Theta(m)$ matrix,
	whose solve cost is a lower order term.
	
	So the final solver operator costs
	\[
	\Otil\left(nnz(A)\cdot n m + n^{2} m^{3} + n^{\omega} m^{2 - \omega}\right)
	\]
	which leads to the final running time by choosing $m$ to balance the terms.
	This bound falls short of the ideal case
	given in Equation~\ref{eq:BlockKrylovGoal}
	mainly due to the need for a denser $B$
	to the well-conditionedness of the Krylov space matrix.
	Instead of $O(n)$ non-zeros total, or about $O(m)$ per column,
	we need $poly(m)$ non-zero variables per column to ensure the
	an $\exp(-O(m))$ condition number of the block Krylov space matrix $K$.
	This in turn leads to a total cost of $O(n \cdot nnz \cdot poly(m))$
	for computing the blocks of the Hankel matrix, and a worse trade
	off when summed against the $\frac{n ^{\omega}}{m^{\omega - 2}}$ term.
	
	\section{Outline of Analysis}
	\label{sec:Analysis}

	In this section we outline our analysis of the algorithm
	through formal theorem statements.
	We start by formalizing our tracking of convergence, and
	the tracking of errors and roundoff errors.

	\subsection{Preliminaries}
	\label{subsec:Preliminares}

	We will use capital letters for matrices,
	lower case letters for vectors and scalars.
	All subscripts are for indexing into entries
	of matrices and vectors, and superscripts are for 
	indexing into entries of a sequence.
	Our notation is summarized in Table~\ref{table:Notations}
	at the end of the paper.
	
	\paragraph{Norms and Singular Values.}
	Our convergence bounds are all in terms of the Euclidean,
	or $\ell_{2}$ norms.
	For a length $n$ vector $x$, the norm of $x$ is given by
	$\norm{x} = \sqrt{\sum_{1 \leq i \leq n} x_i^2}$.
	Similarly, for a matrix $M$, the norm of its entries treated
	as a vector is known as the Frobenius norm, and we have
	\[
	\norm{M}_{F}
	=
	\sqrt{\sum_{ij} M_{ij}^2}
	=
	\sqrt{\textsc{Trace}\left(M^TM\right)}.
	\]
	
	We will also use $\normi{\cdot}$ to denote entry-wise
	norms over a matrix, specifically $\norm{M}_{\infty}$
	to denote the max magnitude of an entry in $M$.
	Note that $\norm{M}_{F} = \normi{M}_{2}$, so we have
	$\normi{M}_{\infty} \leq \norm{M}_{F} \leq n \normi{M}_{\infty}$.
	
	The minimum and maximum singular values of a matrix $M$
	are then defined as the min/max norms of its product against
	a unit vector:
	\[
	\sigma_{\min} \left( M \right)
	=
	\min_{x} \frac{\norm{Mx}_2}{\norm{x}_2}
	\qquad
	\sigma_{\max} \left( M \right)
	=
	\max_{x} \frac{\norm{Mx}_2}{\norm{x}_2},
	\]
	and the condition number of $M$ is defined
	as $\kappa(M) = \sigma_{\max}(M) / \sigma_{\min}(M)$.
	
	Bounds on the minimum singular value allows us to
	transfer perturbation errors to it to its inverse.
	
	\begin{lemma}
		\label{lem:ErrorInvert}
		If $M$ is a full rank square matrix with min and max singular values
		in the range $[\sigma_{\min}, \sigma_{\max}]$, and $\Mtil$
		is some approximation of it such that
		$\norm{\Mtil - M}_{F} \leq \epsilon$
		for some $\epsilon < \sigma_{\min} / 2$, then
		\begin{enumerate}
			\item All singular values in $\Mtil$ are in the range
			$[ \sigma_{\min} - \epsilon, \sigma_{\max} + \epsilon]$, and
			\item The inverse of $\Mtil$ is close to the inverse of $M$:
			\[
			\norm{\Mtil^{-1} - M^{-1}}_F
			\leq
			10 \sigma_{\min}^{-2} \epsilon .
			\]
		\end{enumerate}
	\end{lemma}
	
	\begin{proof}
		The bound on singular values follows from the
		norm minimization/maximization definition of singular values.
		Specifically, we get that for a unit vector $x$,
		\[
		\abs{\norm{\Mtil x}_2 - \norm{M x}_2}
		\leq
		\norm{\left(\Mtil - M \right) x }_2
		\leq
		\norm{\Mtil - M}_2 \norm{x}_2
		\leq
		\epsilon,
		\]
		which means all singular values can change by at most $\epsilon$.
		
		Note that this implies that $\Mtil$ is invertible.
		For the bounds on inverses, note that
		\[
		\Mtil^{-1} - M^{-1}
		=
		M^{-1} \left( M \Mtil^{-1} - I \right)
		=
		M^{-1} \left( M  - \Mtil \right) \Mtil^{-1}.
		\]
		So applying bounds on norms, as well as
		$\norm{\Mtil^{-1}}_2 \leq (\sigma_{\min} - \epsilon)^{-1}
		\leq 2 \sigma_{\min}^{-1}$ gives
		\[
		\norm{\Mtil^{-1} - M^{-1}}_{F}
		\leq
		\norm{M^{-1}}_2
		\norm{M - \Mtil}_{F}
		\norm{\Mtil^{-1}}_2
		\leq
		2 \sigma_{\min}^{-2} \epsilon.
		\]
	\end{proof}
	
	\paragraph{Error Accumulation.}
	Our notion of approximate operators also compose
	well with errors.
	
	\begin{lemma}
		\label{lem:ErrorCompose}
		If $Z^{(1)}$ and $Z^{(2)}$ are linear operators with
		(algorithmic) approximations $\Ztil^{(1)}$ and $\Ztil^{(2)}$
		such that for some $\epsilon < 0.1$, we have
		\[
		\norm{Z^{\left( 1 \right)} - \Ztil^{\left( 1 \right)}}_{F},
		\norm{Z^{\left( 2 \right)} - \Ztil^{\left( 2 \right)}}_{F}
		\leq
		\epsilon
		\]
		then their product satisfies
		\[
		\norm{Z^{\left( 1 \right)}Z^{\left( 2 \right)} -
			\Ztil^{\left( 1 \right)}\Ztil^{\left( 2 \right)}}_{F}
		\leq
		10 \epsilon \max\left\{1, \norm{Z^{\left(1\right)}}_{2}, \norm{Z^{\left(2\right)}}_{2}\right\}.
		\]
	\end{lemma}
	
	\begin{proof}
		Expanding out the errors gives
		\begin{multline*}
			Z^{\left( 1 \right)}Z^{\left( 2 \right)} -
			\Ztil^{\left( 1 \right)}\Ztil^{\left( 2 \right)}
			=
			\left( Z^{\left( 1 \right)} - \Ztil^{\left( 1 \right)} \right)
			Z^{\left( 2 \right)}
			+
			\left( Z^{\left( 2 \right)} - \Ztil^{\left( 2 \right)} \right)
			Z^{\left( 1 \right)}\\
			+
			\left( Z^{\left( 1 \right)} - \Ztil^{\left( 1 \right)} \right)
			\left( Z^{\left( 2 \right)} - \Ztil^{\left( 2 \right)} \right)
		\end{multline*}
		
		The terms involving the original matrix against the
		error gets bounded by the error times the norm
		of the original matrix.
		For the cross term, we have
		\[
		\norm{\left( Z^{\left( 1 \right)} - \Ztil^{\left( 1 \right)} \right)
			\left( Z^{\left( 2 \right)} - \Ztil^{\left( 2 \right)} \right)}
		\leq
		\norm{Z^{\left( 1 \right)} - \Ztil^{\left( 1 \right)}}_{F}
		\cdot
		\norm{Z^{\left( 2 \right)} - \Ztil^{\left( 2 \right)}}_{2}
		\leq
		\epsilon^2,
		\]
		which along with $\epsilon < 0.1$ gives the overall bound.
	\end{proof}
	
	\paragraph{Randomization and Normal Distributions}
	\label{subsec:Random}
	
	Our algorithms rely on randomly perturbing the input matrices
	to make them non-degenerate, and much of our analysis revolving
	analyzing the effect of such perturbations on the eigenvalues.
	We make use of standard notions of probability, in particular,
	the union bound, which states that for any two events $E_1$
	and $E_2$, $\prob{}{E_1 \cup E_2} \leq \prob{}{E_1} + \prob{}{E_2}$.
	
	Such a bound means that it suffices to show that the failure
	probability of any step of our algorithm is $n^{-c}$ for some
	constant $c$.
	The total number of steps is $poly(n)$, so unioning over
	such probabilities still give a success probability of
	at least $1 - n^{-c + O(1)}$.
	
	We will perturb our matrices using Gaussian random variables.
	These random variables $N(0, \sigma)$ have density function
	$g(x) = \frac{1}{\sigma \sqrt{2 \pi}} e^{-x^2 / 2\sigma^2}$.
	They are particularly useful for showing anti-concentration
	because the sum of Gaussians is another Gaussian, with variance
	equalling to the sum of squares, or $\ell_2^2$-norm,
	of the variance terms.
	That is, for a vector $x$ and a (dense) Gaussian vector with
	entry-wise i.i.d. $N(0, 1)$ normal random variables,
	aka. $g \sim N(0, 1)^{n}$, we have $x^T g \sim N(0, \norm{x}_2)$.
	
	The density function of Gaussians means that their magnitude
	exceed $n$ with probability at most $O(\exp(-n^2))$.
	This probability is much smaller than the $n^{-c}$ failure
	probabilities that we want, so to simplify presentation we will
	remove it at the start.
	\begin{claim}
		\label{claim:GaussianMax}
		We can analyze our algorithm conditioning on any normal random
		variable with variance $\sigma$, $N(0, \sigma)$, having magnitude
		at most $n \sigma$.
	\end{claim}
	
	\paragraph{Tracking Word Length.}
	
	The numerical rounding model that we will use is fixed point
	precision.
	The advantage of such a fixed point representation is that it
	significantly simplifies the tracking of errors during
	additions/subtractions.
	The need to keep exact operators means we cannot omit
	intermediate digits.
	Instead, we track both the number of digits before and
	after the decimal point.
	
	The number of trailing digits, or words after the decimal
	point, compound as follows:
	\begin{enumerate}
		\item Adding two numbers with $L_1$ and $L_2$ words after
		the decimal point each results in a number with $\max\{L_1, L_2\}$
		words after the decimal point.
		\item Multiplying two numbers with $L_1$ and $L_2$ words after
		the decimal point each results in a number with $L_1 + L_2$
		words after the decimal point.
	\end{enumerate}
	As $\max\{L_1, L_2\} \leq L_1 + L_2$ when $L_1$ and $L_2$
	are non-negative,
	we will in general assume that when we multiply matrices
	with at most $L_1$ and $L_2$ words after the decimal point,
	the result has at most $L_1 + L_2$ words after the decimal point.
	In particular, if $Z$ is an operator with $L_Z$ words
	after the decimal point, and its input $B$ has $L_B$
	words after the decimal point, the output has at most
	$L_Z + L_B$ words after the decimal point.
	
	Note that both of these bounds are for exact computations.
	The only round off errors come from round-off errors by dropping
	some of the digits, as the matrices themselves are created.
	
	On the other hand, we need to bound the maximum magnitude
	of our operators.
	The number of digits before the decimal point is given by bounds
	on the magnitude of the numbers themselves.
	Such bounds also propagate nicely along multiplications.
	
	\begin{lemma}
		\label{lem:MaxMagnitude}
		If the maximum magnitude of entries in two matrices $Y$ and $Z$
		with dimension at most $n$ are both at most $\alpha$,
		then all entries in $YZ$ have magnitude at most $n  \alpha^2$ as well.
	\end{lemma}
	
	\begin{proof}
		\[
		\abs{\left( YZ \right)_{ij}}
		=
		\abs{\sum_{k} Y_{ik} Z_{kj}}
		\leq
		\sum_{k} \abs{Y_{ik}} \abs{Z_{kj}}
		\leq
		n \alpha^2
		\]
	\end{proof}
	
	Throughout our analyses, we will often rescale the matrices so that
	their max magnitudes are $n^{-2}$.
	This allows us to absorb any constant factor increases in
	magnitudes from multiplying these matrices by Lemma~\ref{lem:MaxMagnitude}
	because $(c \cdot n^{-2})^2 \leq n^{-2}$.
	
	Finally, by doing FFT based fast multiplications
	for all numbers involved~\cite{CormenLRS09:book,HarveyH19},
	we can multiple two numbers with an $O(\log{n})$ factor
	in their lengths.
	This means that when handling two matrices with $L_1$
	and $L_2$ words after the decimal point, and whose
	maximum magnitude is $\mu$, the overhead caused by the
	word-lengths of the numbers involved is $\Otil(\mu + L_1 + L_2$)

	\paragraph{Random Variables and Probability}
	
	For our perturbations we use standard Gaussian $\mathcal{N}(0,1)$ random variables. 
	\begin{fact}
		For $x \sim \mathcal{N}(0,1)$, we have $\Pr(\abs{x}\ge t)\le \frac{2}{t\sqrt{2\pi}}e^{-t^2/2}$.
	\end{fact}
	Thus, with probability at least $1-\exp(-n/2)$, a standard Gaussian variable is bounded by $\sqrt{n}$. We will use $O(n^2)$ such variables and condition on the event that all their norms are bounded by $\sqrt{n}$. 
	
	\subsection{Main Technical Ingredients}
	\label{subsec:Core}
	
	The main technical components of the analysis can be summarized as follows:
	
	\begin{enumerate}
		\item $A$ can be perturbed so that its eigenvalues 
		are separated (Theorem~\ref{thm:PerturbA}, Section~\ref{sec:PerturbA}).
		\item Such a separation implies a well-conditioned Krylov space,
		when it is initialized with sparse random Gaussian vectors.
		(Theorem~\ref{thm:RandKrylov}, Section~\ref{sec:RandKrylov})
		\item This Krylov matrix can be solved efficiently
		using a combination of low displacement rank solvers and fast
		convolutions (Theorem~\ref{thm:Solver}, Section~\ref{sec:Solver}).
		\item The last few rows/columns can be solved efficiently
		via the Schur complement (Lemma~\ref{lem:PadAndSolve}, Section~\ref{sec:Pad}).
	\end{enumerate}
	
	\paragraph{Anti-Concentration of semi-random matrices.}
	A crucial part of our analysis is bounding the spectrum of semi-random matrices.
	Unfortunately, the highly developed literature on spectral properties of random matrices assumes independent entries or independent columns, which no longer hold in the semi-random case of $K$.
	However, getting tight estimates is not important for us (the running time is affected only by the logarithm of the gap/min value). So we adapt methods from random matrix theory to prove sufficient anti-concentration. 
	
	Specifically, after symmetrizing the potentially asymmetric input $A$
	by implicitly generating the operator $A^TA$, we need to bound
	(1) the minimum eigenvalue gap of the coefficient matrix after
	perturbation by a symmetric sparse matrix and
	(2) the minimum singular value of the block Krylov matrix constructed
	by multiplying with a sparse random matrix.
	
	The first step of showing eigenvalue separation is needed
	because if $A$ has a duplicate eigenvalue, the resulting
	Krylov space in it has rank at most $n - 1$.
	We obtain such a separation by perturbing the matrix randomly:
	its analysis follows readily from recent results on separations of
	eigenvalues in random matrices by Luh and Vu~\cite{LuhV18}.
	In Section~\ref{sec:PerturbA}, we show the following separation bound.
	
	\begin{restatable}{theorem}{PerturbA}
		\label{thm:PerturbA}
		For any $n \times n$ symmetric positive definite matrix $\overline{A}$ with:
		\begin{enumerate}
			\item entries at most $1/n$,
			\item eigenvalues at least $1 / \kappa$ for some $\kappa \geq n^{3}$,
		\end{enumerate}
		and any probability where
		\[
		p \geq \frac{300 \log{\kappa} \log{n}}{n}
		\]
		the symmetrically random perturbed matrix $A$ defined as
		\[
		A_{ij} = A_{ji}
		\defeq
		\begin{cases}
			\bar{A}_{ij} + \frac{1}{n^{2} \kappa} \mathcal{N}\left( 0, 1 \right)
			& \qquad \text{w.p. $p$},\\
			\bar{A}_{ij} & \qquad \text{w.p. $1 - p$},
		\end{cases}
		\]
		with probability at least $1 - n^{-10}$
		has all eigenvalues separated by at least $\kappa^{-5 \log{n}}$.
	\end{restatable}
	
	Given this separation, we show that a random $n$-by-$s$ B gives
	a $m$-step Krylov space matrix, as long as $n - ms = \Omega(m)$.
	Furthermore, we pick this $B$ to be sparse in the columns, so we
	can quickly compute $B^T A^{i} B$.
	
	\begin{restatable}{theorem}{RandKrylov}
		\label{thm:RandKrylov}
		Let $A$ be an $n \times n$ symmetric positive definite matrix with
		entries at most $1 / n$, and $\alpha_{A} < n^{-10}$
		a parameter such that:
		\begin{enumerate}
			\item all eigenvalues of $A$ are at least $\alpha_{A}$, and
			\item  all pairs of eigenvalues of $A$
			are separated by at least $\alpha_{A}$.
		\end{enumerate}
		Let $s$ and $m$ be parameters such that
		$n^{0.01} \leq m \leq n^{\frac{1}{4}}$
		and $s \cdot m \leq n - 5 m$.
		The $n$-by-$s$ sparse Gaussian matrix $G^{S}$ where each
		entry is set to $\normal(0, 1)$ with probability at least
		\[
		\frac{10000 m^3 \log\left( 1 / \alpha_{A} \right)}{n}
		\]
		leads to the Krylov space matrix
		\[
		K
		=
		\left[
		\begin{array}{c|c|c|c|c}
			G^{S} & AG^{S} & A^2G^{S} & \ldots & A^{m - 1}G^{S}
		\end{array}
		\right].
		\]
		With probability at least $1 - n^{-2}$
		(over randomness in $G^{S}$),
		$K$ has maximum singular value at most $n^2$,
		and minimum singular value at least $\alpha_{A}^{5m}$.
	\end{restatable}
	
	We prove this in Section~\ref{sec:RandKrylov}.
	These bounds allow us to bound the length of numbers,
	and in turn running time complexity of solving $(AK)^{T}AK$.
	
	\paragraph{Solvers for block Hankel matrices.} 
	
	An important ingredient in our algorithm is a numerically efficient solver for block Hankel matrices. For this we use the notion of displacement rank by Kailath, Kung and Morf~\cite{KailathKM79}.
	Its key statement is that any Schur Complement of a Toeplitz
	Matrix has displacement rank $2$, and can be uniquely represented
	as the factorization of a rank $2$ matrix. This combined with the fact that the displacement rank of the inverse is the same as that of the matrix is used to compute the Schur complement   
	in the first super-fast/asymptotically fast
	solvers for Toeplitz matrices by Bitmead and Anderson~\cite{BitmeadA80}. Here we extend this to block Toeplitz/Hankel matrices and prove numerical stability, using 
	the natural preservation of singular values of Schur complements.
	Specifically, the analysis from Demmel, Dumitriu, Holtz and
	Kleinberg~\cite{DemmelDHK07} can be readily adapted to this setting.
	
	We give full details on the algorithm in Section~\ref{sec:Solver}.
	
	\begin{restatable}{theorem}{Solver}
		\label{thm:Solver}
		If $H$ is an $sm \times sm$ symmetric $s$-block-Hankel matrix
		and $0 < \alpha_{H} < (sm)^{-100}$ is a parameter such that
		every contiguous square block-aligned minor of $H$
		containing the top-right or bottom left corner
		have minimum eigenvalue at least $\alpha_{H}$:
		\[
		\sigma_{\min}\left(H_{\left\{1:i, \left(m-i+1\right):m\right\}}\right),
		\sigma_{\min}\left(H_{\left\{ \left(m-i+1\right):m, 1:i\right\}}\right)
		\geq
		\alpha_{H}
		\qquad \forall 1 \leq i \leq m
		\]
		and all entries in $H$ have magnitude at most $(sm)^{-2}\alpha_{H}^{-1}$,
		then for any error $\epsilon$, we can pre-process
		$H$ in time $\Otil(m s^{\omega}
		\log(\alpha_{H}^{-1} \epsilon^{-1}) )$ to form $\textsc{Solve}_H(\cdot, \epsilon)$
		that corresponds to a linear operator $Z_{H}$ such that:
		\begin{enumerate}
			\item For any $(ms) \times k$ matrix $B$
			with max magnitude $\normi{B}_{\infty}$ and $L_B$ words after the decimal point, $\textsc{Solve}_H(B, \epsilon)$
			returns $Z_H B$ in time
			\[
			\Otil\left(m
			\cdot \max\left\{s^{\omega - 1}{k}, s^2 k^{\omega - 2}\right\}
			\cdot
			\left(\log\left(
			\frac{\left( 1 + \normi{B}_{\infty}\right)
				ms}{\alpha_{H} \epsilon}\right)
			+ L_B \right)
			\right).
			\]
			\item
			$Z_H$ is a high-accuracy approximation to $H^{-1}$:
			\[
			\norm{Z_{H} - H^{-1}}_{F}
			\leq
			\epsilon.
			\]
			\item the entries of $Z_{H}$ have at most
			$O( \log^2{m} \log(\alpha_{H}^{-1} \epsilon^{-1}))$
			words after the decimal point.
		\end{enumerate}
	\end{restatable}
	
	The overhead of $\log^2{m}$ in the word lengths
	of $Z_{H}$ is from the $O(\log{n})$ layers
	of recursion in Fast Fourier Transform times
	the $O(\log{n})$ levels of recursion in the divide-and conquer
	block Schur complement algorithm.
	Note that the relation between the $i\textsuperscript{th}$ block of $H = K^TK$
	and $K$ itself is:
	\[
	H_{\left\{1:i, (n-i+1):m\right\}}
	=
	K_{\left\{:, 1:i\right\}}^T K_{\left\{:, (n-i+1):m\right\}}
	=
	K_{\left\{:, 1:i\right\}}^T A^{m - i - 1}
	K_{\left\{:, 1:i\right\}}.
	\]
	So the min/max singular values of these matrices off by a
	factor of at most $\alpha_{A}^{m}$ from the singular value
	bounds on $H$ itself.
	
	It remains to pad $K$ with random columns to make it
	a square matrix, $Q$.
	We will bound the condition number of this padded matrix,
	and convert a solver for $M = (AK)^{T}(AK)$
	to a solver for $(AQ)^T(AQ)$ in Section~\ref{sec:Pad},
	Specifically, the following bounds are from combining
	Theorems~\ref{thm:RandKrylov} and~\ref{thm:Solver},
	along with properties of random dense Gaussians.
	
	\begin{restatable}{lemma}{PadAndSolve}
		\label{lem:PadAndSolve}
		Let $A$ be an $n \times n$ symmetric positive definite
		matrix with entries at most $1/n$,
		and $0 < \alpha_{A} < n^{-10}$ a parameter such that:
		\begin{enumerate}
			\item all eigenvalues of $A$ are at least
			$\alpha_{A}$, and at most $\alpha_{A}^{-1}$,
			\item all pairs of eigenvalues of $A$
			are separated by at least $\alpha_{A}$,
			\item all entries in $A$ have magnitude at most $\alpha_{A}$,
			and at most $O(\log( 1 / \alpha_{A}))$ words after the
			decimal point.
		\end{enumerate}
		For any parameter $m$ such that $n^{0.01} \leq m
		\leq 0.01 n^{0.2}$,
		the routine $\textsc{BlockKrylov}$ as shown in
		Figure~\ref{fig:solver} pre-processes $A$ in time:
		\begin{enumerate}
			\item $O(n)$ matrix-vector multiplications of $A$
			against vectors with at most $O(m \log( 1 / \alpha_{A}))$ words 
			both before and after the decimal point,
			\item plus operations that cost a total of:
			\[
			\Otil\left(n^2 \cdot m^{3} \cdot \log^{2} \left(1/\alpha_{A}\right)
			+ n^{\omega}
			m^{2 - \omega} \log\left(1/\alpha_{A}\right) \right).
			\]
		\end{enumerate}
		and obtains a routine $\textsc{Solve}_{A}$
		that when given a vector $b$ with $L_b$ words after
		the decimal point, returns $Z_{A}b$ in time
		\[
		\Otil\left(n^2 m
		\cdot \left(  \log\left(1/\alpha_{A}\right)
		+ \normi{b}_{\infty} + L_b\right) \right) ,
		\]
		for some $Z_{A}$ with at most $O(\log(1 / \alpha_A))$ words
		after the decimal point such that
		\[
		\norm{Z_{A} - A^{-1}}_{F} \leq \alpha_{A}.
		\]
	\end{restatable}
	
	Note that we dropped $\epsilon$ as a parameter
	to simplify the statement of the guarantees.
	Such an omission is acceptable because
	if we want accuracy less than the eigenvalue bounds
	of $A$, we can simply run the algorithm with
	$\alpha_{A} \leftarrow \epsilon$ due to the pre-conditions
	holding upon $\alpha_{A}$ decreasing.
	With iterative refinement (e.g.~\cite{Saad03:book},
	it is also possible to lower  the dependence on
	$\log(1 / \epsilon)$ to linear instead of the cubic
	dependence on $\log(1 / \alpha_{H})$.
	
	\subsection{Proof of Main Theorem}
	\label{subsec:ProofMain}
	
	It remains to use the random perturbation specified
	in Theorem~\ref{thm:PerturbA} and taking the outer product
	to reduce a general system to the symmetric, eigenvalue
	well separated case covered in Lemma~\ref{lem:PadAndSolve}.
	
	The overall algorithm then takes a symmetrized version of
	the original matrix $\overline{A}$, perturbs it,
	and then converts the result of the block Krylov space
	method back.
	Its pseudocode is in Figure~\ref{fig:LEA}
	
	\begin{figure}[t]
		\fbox{\parbox{\textwidth}{
				
				{\bf \textsc{LinearEquationApproximation}(
					$A$, $b$:
					integer matrix/vector pair,
					$\kappa$: condition number bound for $A$,
					$\epsilon$: error threshold.
					)}
				
				\begin{enumerate}
					
					\item Compute $\theta_{A} \leftarrow \normi{A}_{\infty}$.
					
					\item Generate random symmetric matrix $R$ with
					\[
					R_{ij} = R_{ji} = \frac{\epsilon}{n^{10} \kappa^2} \mathcal{N}(0, 1)
					\]
					with probability $\frac{O\left( \log\left( \kappa / \epsilon \right) \log{n} \right)}{n}$.
					
					\item 
					Implicitly generate
					\[
					\Atil
					=
					\frac{1}{n^4\theta_{A}^2}
					A^TA + R
					\]
					and its associated matrix-multiplication operator
					$\textsc{MatVec}_{\Atil}(\cdot, \delta)$.
					
					\item Build solver for $\Atil$ via
					\[
					\textsc{Solve}_{\Atil}
					\leftarrow
					\textsc{BlockKrylov}\left(
					\textsc{MatVec}_{\Atil}\left(\cdot, \delta\right),
					\left( n^{8} \kappa^{2} \epsilon^{-1} \right)^{-5 \log{n}},
					n^{\frac{\omega - 2}{\omega + 1}}
					nnz\left(A\right)^{\frac{\omega - 2}{\omega + 1}}
					\right).
					\]
					
					\item
					\label{step:FinalRescale}
					Return
					\[
					\frac{1}{n^4 \theta_A^2}
					\cdot
					\textsc{Solve}_{\Atil} \left( A^T b \right).
					\]
					
				\end{enumerate}
		}}

		\caption{Pseudocode for block Krylov space algorithm.}
		\label{fig:LEA}
	\end{figure}
	
	\begin{proof}(Of Theorem~\ref{thm:Main})
		Let $\Ahat$ be the copy of $A$ scaled down by $n^2 \theta_{A}$:
		\[
		\Ahat
		=
		\frac{1}{n^2 \theta_{A}}{A}
		=
		\frac{1}{n^2 \norm{A}_{\infty}} A.
		\]
		This rescaling gives us bounds on both the maximum and
		minimum entries of $\Ahat$.
		The rescaling ensures that the max magnitude of an entry
		in $\Ahat$ is at most $1 / n^2$.
		Therefore its Frobenius norm, and in turn max singular value,
		is at most $1$.
		On the other hand, the max singular of $A$ is at least
		$\normi{A}_{\infty} = \theta_{A}$: consider the unit vector that's $1$ in
		the entry corresponding to the column containing the max magnitude entry of $A$, and $0$ everywhere else.
		This plus the bound on condition number of $\kappa$ gives that
		the minimum singular value of $A$ is at least
		\[
		\sigma_{\min}\left( A \right)
		\geq
		\frac{1}{\kappa} \sigma_{\max}\left( A \right)
		\geq
		\frac{\theta_{A}}{\kappa},
		\]
		which coupled with the rescaling by $\frac{1}{n^2 \theta_{A}}$ gives
		\[
		\sigma_{\min}\left( \Ahat \right)
		\geq
		\frac{1}{n^2 \theta_{A}}
		\cdot
		\frac{\theta_{A}}{\kappa}
		=
		\frac{1}{n^2 \kappa}.
		\]
		
		The matrix that we pass onto the block Krylov method,
		$\Atil$, is then the outer-product of $\Ahat$ plus a sparse
		random perturbation $R$ with each entry is set
		(symmetrically when across the diagonal)
		to $\frac{\epsilon}{n^{4} \kappa}N(0, 1)$ with probability
		$O(\log{n} \log(\kappa / \epsilon)) / n$
		\[
		\Atil \leftarrow
		\Ahat^{T} \Ahat
		+
		R.
		\]
		This matrix $\Atil$ is symmetric.
		Furthermore, by Claim~\ref{claim:GaussianMax}, we may assume
		that the max magnitude of an entry in $R$ is at most
		$\epsilon n^{-9} \kappa^{-2}$, which gives
		\[
		\norm{R}_{F}
		\leq
		\frac{\epsilon}{n^{8} \kappa^2}.
		\]
		So we also get that the max magnitude of an entry in
		$\Atil$ is still at most $2 n^{-2}$.
		Taking this perturbation bound into
		Lemma~\ref{lem:ErrorInvert} also gives that all
		eigenvalues of $\Atil$ are in the range
		\[
		\left[
		\frac{1}{n^{5}\kappa^2}
		,1
		\right].
		\]

		By Theorem~\ref{thm:PerturbA},
		the minimum eigenvalue separation in this perturbed matrix
		$\Atil$ is at least
		\[
		\left( n^{8} \kappa^2 \epsilon^{-1} \right)^{-5 \log{n}}.
		\]
		Also, by concentration bounds on the number of entries
		picked in $R$, its number of non-zeros is with high
		probability at most
		$O(n\log{n} \log(\kappa n / \epsilon))
		= \Otil(n \log(\kappa / \epsilon))$.
		
		As we only want an error of $\epsilon$,
		we can round all entries in $A$ to precision
		$\epsilon / \kappa$
		without affecting the quality of the answer.
		
		So we can invoke Lemma~\ref{lem:PadAndSolve}
		with
		\[
		\alpha_{\Atil}
		=
		\left( n^{8} \kappa^2 \epsilon^{-1} \right)^{-5 \log{n}},
		\]
		which leads to a solve operator $Z_{\Atil}$ such that
		\[
		\norm{Z_{\Atil} - \Atil^{-1}}_{F}
		\leq \alpha_{\Atil}
		\leq
		\left( n^{8} \kappa^2 \epsilon^{-1} \right)^{-5 \log{n}}
		\leq \frac{\epsilon}{n^{40} \kappa^{10}}.
		\]
		The error conversion lemma from Lemma~\ref{lem:ErrorInvert}
		along with the condition that the min-singular value
		of $\Atil$ is at least $\frac{1}{n^{5} \kappa^2}$ implies that
		\[
		\norm{Z_{\Atil}^{-1} - \Atil}_{F}
		\leq
		\frac{\epsilon}{n^{30} \kappa^{6}}
		\]
		or factoring into the bound on the size of $R$
		via triangle inequality:
		\[
		\norm{Z_{\Atil}^{-1} - \Ahat^T \Ahat }_{F}
		\leq
		\frac{2 \epsilon}{n^{8} \kappa^{4}},
		\]
		which when inverted again via Lemma~\ref{lem:ErrorInvert}
		and the min singular value bound gives
		\[
		\norm{Z_{\Atil} - \left(\Ahat^T \Ahat\right)^{-1} }_{F}
		\leq
		\frac{\epsilon}{n^2}.
		\]
		
		It remains to propagate this error across the rescaling in
		Step~\ref{step:FinalRescale}.
		Since $\Ahat = \frac{1}{n^2\theta_{A}} A$, we have
		\[
		\left( A^T A \right)^{-1} =
		\frac{1}{n^4 \theta_{A}^2} \left(\Ahat^T \Ahat\right)^{-1},
		\]
		and in turn the error bound translates to
		\[
		\norm{\frac{1}{n^4 \theta_{A}^2} Z - \left(\Ahat^T \Ahat\right)^{-1}}_{F}
		\leq
		\frac{\epsilon}{n^4 \theta_{A}^2}.
		\]
		The input on the other hand has
		\[
		\Pi_{A} b
		=
		A \left( A^T A \right)^{-1} A^T b,
		\]
		so the error after multiplication by $A$ is
		\[
		A \left(\frac{1}{n^4 \theta_{A}^2} Z - \left(\Ahat^T \Ahat\right)^{-1}\right) A^Tb,
		\]
		which incorporating the above, as well as
		$\norm{A}_2 \leq n\theta_{A} $ gives
		\[
		\norm{A
			\left[ \frac{1}{n^4 \theta_{A}^2} Z A^Tb \right]
			- \pi_{A} b}_{2}
		\leq
		\frac{\epsilon}{n^3 \theta_{A}}
		\norm{A^T b}_{2}.
		\]
		On the other hand, because the max eigenvalue of
		$A^TA$ is at most $\norm{A}_{F}^2 \leq n^2 \theta_{A}^2$,
		the minimum eigenvalue of $(A^TA)^{-1}$
		is at least $\theta_{A}^{-2}$.
		So we have
		\[
		\norm{\Pi_{A} b}_{2}
		=
		\norm{A^T b}_{\left(A^TA\right)^{-1}}
		\geq
		\frac{1}{n^2 \theta_{A}} \norm{A^T b}_{2}.
		\]
		Combining the two bounds then gives that the error
		in the return value is at most $\epsilon \norm{\Pi_{A} b}_2$.
		
		For the total running time, the number of non-zeros in $R$
		implies that the total cost of mutliplying $\Atil$ against
		a vector with $\Otil(m \log(1 / \alpha_{A}))
		= \Otil(m \log{n} \log(\kappa / \epsilon))
		= \Otil(m \log(\kappa / \epsilon))$ is
		\[
		\Otil\left( \left( nnz\left( A \right) + n\right) m \log^2 \left( \kappa / \epsilon \right) \right)
		\leq
		\Otil\left( nnz\left( A \right) m \log^2 \left( \kappa / \epsilon \right) \right),
		\]
		where the inequality of $nnz(A) \leq n$ follows
		pre-processing to remove empty rows and columns.
		So the total construction cost given in
		Lemma~\ref{lem:PadAndSolve} simplifies to
		\[
		O\left( n^2 m^3 \log^2\left( \kappa / \epsilon \right)
		+ n^{\omega} m^{2 - \omega} \log\left( \kappa / \epsilon \right)
		+ n \cdot nnz \left( A \right) \cdot m \log\left( \kappa / \epsilon \right) \right)
		\]
		
		The input vector, $\frac{1}{\theta_{Y}} y$ has max magnitude
		at most $1$, and can thus be rounded to
		$O(\log(\kappa / \epsilon))$ words after the decimal point as well.
		This then goes into the solve cost with
		$\log( \normi{b}_{\infty}) + L_b
		\leq O(\log(n \kappa / \epsilon))$, which gives
		a total of $\Otil(n^2 m \log( \kappa / \epsilon))$,
		which is a lower order term compared to the construction cost.
		The cost of the additional multiplication in $A$
		is also a lower order term.
		
		Optimizing $m$ in this expression above
		based on only $n$ and $nnz(A)$ gives that we should
		choose $m$ so that
		\[
		\max\left\{ n \cdot nnz\left( A \right) m, n^2 m^3\right\}
		=
		n^{\omega} m^{2 - \omega},
		\]
		or
		\[
		\max\left\{ n \cdot nnz\left( A \right) m^{\omega - 1}, n^2 m^{
			\omega + 1} \right\}
		=
		n^{\omega}.
		\]
		The first term implies
		\[
		m \leq \left( n^{\omega - 1} \cdot nnz\left(A \right)^{-1} \right)^{\frac{1}{\omega - 1}}
		=
		n \cdot nnz\left( A \right)^{\frac{-1}{\omega - 1}}
		\]
		while the second term implies
		\[
		m \leq n^{\frac{\omega - 2}{\omega + 1}}.
		\]
		Substituting the minimum of these two bounds
		back into $n^{\omega} m^{2 - \omega}$ and noting that
		$2 - \omega \leq 0$ gives that the total runtime dependence
		on $n$ and $nnz(A)$ is at most
		\[
		n^{\omega}
		\cdot
		\max\left\{
		n^{\frac{\left( \omega - 2 \right) \left( 2 - \omega \right) }{\omega + 1}},
		n^{2 - \omega} \cdot nnz\left( A \right)^{
			\frac{- \left( 2 - \omega \right)}{\omega - 1}}
		\right\}
		=
		\max\left\{
		n^{\frac{5 \omega - 4}{\omega + 1}},
		n^2 \cdot nnz\left( A \right)^{\frac{\omega - 2}{\omega - 1}}
		\right\}.
		\]
		Incorporating the trailing terms,
		then gives the the bound stated in Theorem~\ref{thm:Main},
		with $c$ set to $2$ plus the number of log factors hidden in the $\Otil$.
	\end{proof}
	
	\input{PerturbA}

	\input{EpsNet}

	\input{Solver}

	\input{PadAndSolve}
	
	\section{Discussion}
	\label{sec:Discussion}
	
	We have presented a faster solver for linear systems with moderately sparse  coefficient matrices under bounded word-length arithmetic with logarithmic dependence on the condition number.
	This is the first separation between the complexity of matrix multiplication and solving linear systems in the bounded precision setting.
	While both our algorithm and analysis are likely improvable,
	we believe they demonstrate that there are still many
	sparse numerical problems and algorithms that remain to
	be better understood theoretically. We list a few avenues for future work.
	
	\paragraph{Random Matrices.}
	
	The asymptotic gap between our running time of about
	$n^{2.33}$ and the $O(n^{2.28})$ running time of computing inverses
	of sparse matrices over finite fields~\cite{EberlyGGSV06} is
	mainly due to the overhead our minimum singular value bound
	from Theorem~\ref{thm:RandKrylov}, specifically the requirement
	of $\Omega(m^3)$ non-zeros per column on average.
	We conjecture that a similar bound holds for $\Otil(m)$
	non-zeros per column, and also in the full Krylov space case.
	\begin{enumerate}
		\item Can we lower bound the min singular value of
		a block Krylov space matrix generated from a random matrix with
		$\Otil(m)$ non-zeros per column?
		\item Can we lower bound the min singular value of
		a block Krylov space matrix where $m \cdot s = n$
		for general values of $s$ (block size)
		and $m$ (number of steps)?
	\end{enumerate}
	The second improvement of removing
	the additional $\Omega(m)$ columns would not give
	asymptotic speedups.
	It would however remove the need for the extra steps (padding with random Gaussian columns)
	in Section~\ref{sec:Pad}.
	Such a bound for the square case would likely 
	require developing new tools for analyzing matrix anti-concentration. 
	\begin{enumerate}
		\item[3.] Are there general purpose bounds on the
		min singular value of a sum of random matrices,
		akin to matrix concentration bounds (which focus
		on the max singular value)~\cite{RudelsonV10,Tropp15}.
	\end{enumerate}
	The connections with random matrix theory can also be leveraged in the reverse direction:
	\begin{enumerate}
		\item[4.] Can linear systems over random matrices with i.i.d. entries be solved faster?
	\end{enumerate}
	An interesting case here is sparse matrices with
	non-zeros set to $\pm 1$ independently.
	Such matrices have condition number $\Theta(n^2)$ with
	constant probability~\cite{RudelsonV10}, which means that
	the conjugate gradient algorithm has bit complexity
	$O(n \cdot nnz)$ on such systems.
	Therefore, we believe these matrices present a natural starting
	point for investigating the possibility of faster algorithms
	for denser matrices with $nnz > \Omega(n^{\omega - 1}$).
	
	\paragraph{Numerical Algorithms.}
	
	The bounded precision solver for block Hankel matrices
	in Section~\ref{sec:Solver} is built upon the earliest
	tools for speeding up solvers for such structured
	matrices~\cite{KailathKM79,BitmeadA80}, as well as the
	first sparsified block Cholesky algorithm for solving
	graph Laplacians~\cite{KyngLPSS16}.
	We believe the more recent developments in solvers
	for Hankel/Toeplitz matrices~\cite{XiXCB14} as well
	as graph Laplacians~\cite{KyngS16} can be incorporated
	to give better and more practical routines for solving
	block-Hankel/Toeplitz matrices.
	\begin{enumerate}
		\item[5.] Is there a superfast solver under bounded
		precision for block Hankel/Toeplitz matrices
		that does not use recursion?
	\end{enumerate}
	It would also be interesting to investigate whether
	recent developments in randomized numerical linear
	algebra can work for Hankel/Toeplitz matrices.
	Some possible questoins there are:
	\begin{enumerate}
		\item[6.] Can we turn $m \times s^{\omega}$ into $O(ms^2 + s^{\omega})$ using more recent developments sparse projections
		(e.g. CountSketch / sparse JL / sparse Gaussian instead of a dense Gaussian).
		\item[7.] Is there an algorithm that takes a rank $r$ factorization of $I - XY \in \Re^{n \times n}$, and computes in time $\Otil(n \mbox{poly}(r))$ the a rank $r$ factorization/approximation of $I - YX$?
	\end{enumerate}
	
	Another intriguing question is the extensions
	of this approach to the high condition number,
	or exact integer solution, setting.
	Here the current best running time bounds are via $p$-adic
	representations of fractions~\cite{Dixon82},
	which are significantly less understood compared to decimal point
	based representations.
	In the dense case, an algorithm via shifted $p$-adic numbers
	by Storjohann~\cite{Storjohann05} achieves an $O(n^{\omega})$ bit complexity.
	Therefore, it is natural to hope for a similar $\Otil(n \cdot nnz)$ bit
	complexity algorithm for producing exact integer solutions.
	A natural starting point could be the role of low-rank sketching in solvers
	that take advantage of displacement rank, i.e., extending
	the $p$-adic algorithms to handle low rank matrices:
	\begin{enumerate}
		\item[8.] Is there an $O(n \cdot r^{\omega - 1})$ time
		algorithm for exactly solving linear regression problems
		involving an $n$-by-$n$ integer matrix with rank $r$?
	\end{enumerate}
	
	Finally, we note that the paper by Eberly et al.~\cite{EberlyGGSV06} that
	proposed block-Krylov based methods for matrix inversion
	also included experimental results that demonstrated good
	performances as an exact solver over finite fields.
	It might be possible to practically evaluate
	block-Krylov type methods for solving general systems
	of linear equations.
	Here it is worth remarking that even if one uses naive $\Theta(n^3)$
	time matrix multiplication, both the Eberly et al.
	algorithm~\cite{EberlyGGSV07} (when combined with $p$-adic
	representations), as well as our algorithm, still take sub-cubic time.
	
	\section*{Acknowldgements}
	Richard Peng was supported in part by NSF CAREER award 1846218, and Santosh Vempala by NSF awards AF-1909756 and AF-2007443.
	We thank Mark Giesbrecht for bringing to our attention the literature
	on block-Krylov space algorithms;
	Yin Tat Lee for discussions on random linear systems; 
	Yi Li, Anup B. Rao, and Ameya Velingker for discussions about
	high dimensional concentration and anti-concentration bounds;
	Mehrdad Ghadiri, He Jia, Kyle Luh, Silvia Casacuberta Puig, 
	and anonymous reviewers for comments on earlier versions of this paper.
	
	\bibliographystyle{alpha}
	\bibliography{ref}
	
	\begin{appendix}
		
		\input{Vandermonde}

		\input{DisplacementRankDetails}
		
		\input{Convolution}

		\pagebreak
		
		\input{Notations}

	\end{appendix}
	
\end{document}

%% file: PerturbA.tex
\section{Separating Eigenvalues via Perturbations}
\label{sec:PerturbA}

In this section we show that a small symmetric random perturbation
to $A$ creates a matrix whose eigenvalues are well separated.
The main result that we will prove is Theorem~\ref{thm:PerturbA},
which we restate below.

\PerturbA*

Our proof follows the outline and structure of the bounds
on eigenvalue separations in random matrices by
Nguyen, Tao and Vu~\cite{NguyenTV17} for the dense case,
and by Luh and Vu~\cite{LuhV18} and
Lopatto and Luh~\cite{LopattoL19} for the sparse case.
The separation of eigenvalues needed for
Krylov space methods was mentioned as a motivating application
in~\cite{NguyenTV17}.
Our proof can be viewed as a more basic variant of the
proof on gaps of random sparse matrices by Lopatto and
Luh~\cite{LopattoL19}.
Because we only need to prove a polylog bound on the number
of digits, we do not need to associate the $O(\log{n})$ levels
of $\epsilon$-nets using chaining type arguments.

On the other hand, we do need to slightly modify the arguments
in Lopatto and Luh~\cite{LopattoL19} to handle the initial
$A$ matrix that's present in addition the random terms.
These modifications we make are akin to the ones needed
in the dense case by Nguyen, Tao and Vu~\cite{NguyenTV17}
to handle arbitrary expectations.

The starting point is the interlacing theorem,
which relates the eigenvalue separation to the inner product of the
last column with any eigenvector of the principle minor without it
and its corresponding row, which we denote as $A-$.

\begin{lemma}
	\label{lem:Interlacing}
	Let $A$ be a symmetric $n$-by-$n$ matrix,
	and denote its $(n - 1)$-by-$(n- 1)$ leading minor by $A-$:
	\[
	A-
	=
	A_{1:\left(n - 1\right), 1:\left(n - 1\right)}
	\]
	Then the eigenavlues of $A$ in sorted order,
	$\lambda(A, 1) \leq
	\lambda(A, 2) \leq
	\ldots
	\leq \lambda(A, n)$
	and the eigenvalues of $A-$ in sorted order,
	$\lambda(A-, 1),
	\leq \lambda(A-, 2) \leq
	\ldots
	\leq \lambda(A-, n - 1)$
	interlace each other when aggregated:
	\[
	\lambda\left(A, 1\right)
	\leq
	\lambda\left(A-, 1 \right)
	\leq
	\lambda\left(A, 2\right)
	\leq
	\ldots
	\lambda\left(A, n - 1\right)
	\leq
	\lambda\left(A-, n - 1\right)
	\leq
	\lambda\left(A, n\right)
	.
	\]
\end{lemma}

In particular, this allows us to lower bound
the gap between $\lambda_{i - 1}(A)$ and $\lambda_{i}(A)$
by lower bounding the gap between the $i$th eigenvalues
of $A$ and $A-$.
This bound can in turn be computed by left and right
multiplying $A-$, the principle minor, against the corresponding
eigenvectors of $A-$ and $A$.

\begin{lemma}
\label{lem:EigenGap}
Let $v(A, i)$ denote the eigenvector of $A-$ corresponding
to $\lambda_i(A)$, and
$v(A-, i)$ denote the eigenvector of $A-$ corresponding
to $\lambda_i(A-)$.
Then we have
\[
\abs{\lambda_i\left( A \right) - \lambda_i\left( A- \right) }
\geq
\abs{v\left(A, i\right)_n}
\abs{v\left(A-, i\right)^{T} A_{1:\left( n - 1\right), n}}.
\]
\end{lemma}

\begin{proof}
Taking the condition of $v(A, i)$ being an eigenvector for $A$:
\[
A v\left( A, i \right)
= 
\lambda \left( A, i \right) v\left( A, i \right)
\]
and isolating it to the first $n - 1$ rows gives:
\[
\lambda \left( A, i \right) v\left( A, i \right)_{1:\left( n - 1 \right)}
=
\left[ A v\left( A, i \right) \right]_{1:\left( n - 1 \right)}
=
\left(A-\right) v\left( A, i \right)_{1:\left( n - 1 \right)}
+ A_{1:\left(n - 1\right), n} v\left( A, i \right)_{n}.
\]

Left multiplying this equality by $v(A-, i)^{T}$, and invoking
substituting in $v(A-, i)^T A- = \lambda (A-, i) v(A-, i)$ given by the fact that $v(A-, i)$ is an eigenvector for $A-$ gives:
\begin{multline*}
\lambda \left( A, i \right) \cdot
v\left( A-, i \right)^{T}
v\left( A, i \right)_{1:\left( n - 1 \right)}\\
=
v\left( A-, i \right)^{T} \left(A-\right) v\left( A, i \right)_{1:\left( n - 1 \right)}
+ v\left( A-, i \right)^{T} A_{1:\left(n - 1\right), n} v\left( A, i \right)_{n}\\
=
\lambda \left( A-, i \right) \cdot
v\left( A-, i \right)^{T}
v\left( A, i \right)_{1:\left( n - 1 \right)}
+
v\left( A, i \right)_{n} \cdot
v\left( A-, i \right)^{T} A_{1:\left(n - 1\right), n} 
 .
\end{multline*}
Moving the two terms involving the dot product
$v(A-, i)^{T} v(A, i)_{1:(n - 1)}$ onto the same
side then gives an expression involving the difference
of eigenvalues:
\[
\left( \lambda\left( A, i \right) - \lambda\left( A -, i \right)\right)
\cdot
v\left( A-, i \right)^{T}
v\left( A, i \right)_{1:\left( n - 1 \right)}
=
v\left( A, i \right)_{n} \cdot
v\left( A-, i \right)^{T} A_{1:\left(n - 1\right), n}.
\]
or upon taking absolute values and dividing by the dot product
\[
\abs{\lambda\left( A, i \right) - \lambda\left( A -, i \right)}
=
\frac
{
	\abs{
		v\left( A, i \right)_{n} \cdot
		v\left( A-, i \right)^{T} A_{1:\left(n - 1\right), n}
	}
}
{
	\abs{
		v\left( A-, i \right)^{T}
		v\left( A, i \right)_{1:\left( n - 1 \right)}
	}
}.
\]

As both $v(A-, i)$ and $v(A, i)$ are unit vectors,
we can also upper bound the denominator above by $1$:
\[
\abs{v\left( A-, i \right)^{T} v\left( A, i \right)_{1:\left( n - 1 \right)}}
\leq
\norm{v\left( A-, i \right)}_2
\norm{v\left( A, i \right)_{1:\left( s - 1 \right)}}_2
\leq
\norm{v\left( A-, i \right)}_2
\norm{v\left( A, i \right)}_2
\leq 1,
\]
which then gives the lower bound on eigengap
by the dot product of $v(A-, i)$ against the
last column of $A$ minus its last entry.
\end{proof}

Note that because $v(A, i)$ is unit lengthed,
it has at least one entry with magnitude at least $n^{-1/2}$.
By permutating the rows and columns,
we can choose this entry to be entry $n$ when invoking
Lemma~\ref{lem:Interlacing}.
So it suffices to lower bound the dot product of every
column of $A$ against the eigenvectors of the minor
with that row/column removed.
When the columns are dense Gaussians, this is simply a Gaussian
whose variance is the norm of the eigenvector, which is $1$.

For the sparse case, we need to prove that
the eigenvector is dense in many entries.
We will invoke the following lemma for every principle
$(n - 1)$-by-$(n - 1)$ minor of $A$.

\begin{lemma}
\label{lem:EigenvectorDense}
Let $\bar{A}$ and $A$ be the initial and perturbed matrix
as given in Theorem~\ref{thm:PerturbA}.
That is, $\bar{A}$ is symmetric, has all entries with magnitude at most
$1 / n$, and minimum eigenvalue at least $1 / \kappa$ for some
$\kappa \geq n^{3}$; and $A$ is formed by adding a symmetric sparse
Gaussian matrix where each entry is set to
$\frac{1}{n^2 \kappa} \cdot \normal(0, 1)$
with probability $300 \log{n} \log\kappa / n$, and $0$ otherwise.
Then with probability at least $1 - n^{-20}$,
all eigenvectors of $A$ have at least
$\frac{n}{40 \log{n}}$ entries with magnitude
at least $\kappa^{ - 3 \log{n}}$.
\end{lemma}

The proof is by inductively constructing $O(\log{n})$ 
levels of $\epsilon$ nets.
A similar use of this technique for bounding the density of 
null space vectors in matrices with completely independent entries
is in Lemma~\ref{lem:SparseGaussianNullSpace}.

\begin{proof}	
First, by Claim~\ref{claim:GaussianMax},
we may assume that all entries in the perturbation generated
are bounded by $1/(\kappa n) \leq \frac{1}{n}$.
This in turn gives $\norm{A-\bar{A}}_F\le 1/(n\kappa)$,
and so by Lemma~\ref{lem:ErrorInvert}, all eigenvalues
of $A$ are between $\frac{1}{2\kappa}$ and $2$.

We will prove by induction on $i$
that for all $i \in [0, \log(\frac{n}{1000 \log{n}})]$,
with probability at least $1 - i \cdot n^{-21}$,
all unit length vectors $x$ such that $A x = \lambda x$
for some $\lambda$ with absolute value in the range
$[\frac{1}{2\kappa}, 2]$ has at least $2^{i}$ entries
with magnitude at least $\kappa^{-100 (i + 1)}$.
	
The base case of $i = 0$ follows that the norm of
$x$ being $1$: at least one of $n$ entries
in $x$ must have magnitude at least $1 / n$.
	
For the inductive case, we construct an $\epsilon$-net
consisting of all vectors with entries that are integer multiples of
$\epsilon_i$, which we set to
\[
\epsilon_i
\leftarrow
\kappa^{-3 \left(i + 1\right)}.
\]
	
For any vector $x$, rounding every entry in $x$ to the nearest
multiple of $\epsilon_i$ produces a vector $\xhat$ such that
\[
\norm{x - \xhat}_2
\leq
\sqrt{n} \epsilon_i
\]
which when multiplied by the norm of $A$ gives
\[
\norm{Ax - A\xhat}_2
\leq
\norm{A}_2
\cdot
\norm{x - \xhat}_2
\leq
2\sqrt{n} \epsilon_i.
\]
Combining these then gives
\[
\norm{A \xhat - \lambda \xhat}_{2}
\leq
\norm{A - \lambda x} + 4 \sqrt{n} \epsilon_{i}.
\]
This means it suffices to show that with good probability,
for all $\xhat$ with:
\begin{enumerate}
    \item at most $2^{i}$ non-zero entries,
    \item all non-zeros are integer multiple of $\epsilon_i$,
    \item overall norm at most $2$,
\end{enumerate}
there does not exist some
$\lambda$ with $\abs{\lambda} \in [\frac{1}{2\kappa}, 2]$
such that
\[
\norm{A \xhat - \lambda \xhat}_2
\leq
10 \sqrt{n} \epsilon_i.
\]
	
By the inductive hypothesis, it suffices to consider $x$
with at least $2^{i - 1}$ entries with magnitude at least $\epsilon_{i - 1}$.
Such vectors in turn round to $\xhat$ with at least
$2^{i - 1}$ entries whose magnitudes are at least
\[
\epsilon_{i - 1}
-
\epsilon_i
\geq
\frac{1}{2} \epsilon_{i - 1}.
\]
Let this subset of entries be $\textsc{Large}$.
	
We will bound the probability of $\xhat$ being close to an eigenvector
by applying the approximation condition to its zeros.
The choice of $i$ so that
$2^{i} \leq \frac{n}{1000 \log{n}} \leq \frac{n}{2}$ means $\xhat$ has at least $n / 2$ entries that are $0$.
For each such entry $i$, the fact that $\lambda \xhat_{i} = 0$
for any choice of $\lambda$ means it suffices to upper bound
the probability over $A$ of
\[
\abs{A_{i, :} \xhat}
\leq
10 \sqrt{n} \epsilon_i
\]
	
As these entries with zeros are disjoint from the ones in
$\textsc{Large}$,
the corresponding entries in $A$ are independent from each other.
So we can consider the rows corresponding to each such $i$
independently.
The probability that no entry in $A_{i, \textsc{Large}}$
is chosen to be a non-zero is
\[
\left(1 - p \right)^{2^{i - 1}}
=
\left(1 - \frac{300 \log{n} \log{\kappa}}{n} \right)^{2^{i - 1}}
\geq
\exp\left( \frac{- 300 \log{n} \log{\kappa} 2^{i - 1}}{n} \right).
\]
On the other hand, if one such entry is picked,
the resulting Gaussian corresponding to $A_{i, :} \xhat$ has
variance at least
\[
\frac{\epsilon_{i - 1}}{2} \cdot \frac{1}{n^2 \kappa},
\]
so its probability of being in an interval of size at most
$10 \sqrt{n} \epsilon_i$ is at most
\[
10 \sqrt{n} \epsilon_{i}
\frac{2 n^2 \kappa}{\epsilon_{i  - 1}}
=
20 n^{2.5} \kappa \frac{\epsilon_{i}}{\epsilon_{i - 1}}
\leq
\kappa^{-1},
\]
where the second inequality uses
$\epsilon_{i + 1} = \kappa^{-3} \epsilon_{i}$ and $\kappa \geq n^3$.
Combining these two then gives
\[
\prob{A}{\abs{A_{i, :} \xhat} \leq n^2 \epsilon_i}
\leq
\max\left\{\exp\left( \frac{- 300 \log\kappa 2^{i - 1}}{n} \right), \kappa^{-1} \right\},
\]
which compounded over the at least $n/2$ choices of $i$
gives that the overall probability of $A$ being picked so that
$\xhat$ could be close to an eigenvector is at most
\[
\max\left\{\exp\left( - 300 \log\kappa  2^{i - 2}\right), \kappa^{-n/2} \right\}.
\]
	
So it remains to take a union bound of this probability
over the size
of the $\epsilon$-net.
Recall that we considered all $\xhat$ with
at most $2^{i}$ nonzeros, and each such non-zero is an
integer multiple of $\epsilon_i$ in the range $[-2, 2]$.
This means the total size of the $\epsilon$-net can be upper
bounded by:
\[
\binom{n}{2^{i}} \cdot \left( \frac{4}{\epsilon_i} \right)^{2^{i}}
\leq
n^{2^{i}}
\cdot
\kappa^{3 \left( i + 1 \right) \cdot 2^{i}}
\leq
\kappa^{7 i \cdot 2^{i}},
\]
where the last inequality uses the assumption of $i \geq 1$.

We then invoke union bound over this $\epsilon$ net.
That is, we upper bound the overall failure probability
by multiplying the two failures probabilities for individual
vectors obtained earlier against this size bound above.
For the first term, we get
\[
\kappa^{7 \cdot i \cdot 2^{i}}
\cdot
\exp\left( - 300 \log{\kappa} \log{n} 2^{i - 2} \right)
=
\kappa^{7 i 2^{i} - 300 \log{n} 2^{i - 2}},
\]
which is less than $n^{-22}$ because $i < \log{n}$.

For the second term, we get
\[
\kappa^{7 i \cdot 2^{i}} \cdot \kappa^{-n/2}
=
\kappa^{7 \cdot i \cdot 2^{i} - n/2},
\]
which is at most $\kappa^{-22} \leq n^{-22}$
when $2^{i} \leq \frac{n}{20\log{n}}$.
Thus the overall failure probability is at most
$2 \cdot n^{-22} \leq n^{-21}$,
and the inductive hypothesis holds for $i$ as well by union bound.

The choice of $i$ being powers of $2$s means that
the last value picked is at least half the upper bound,
aka. at least $\frac{n}{40 \log{n}}$.
So we get that with probability at least $1 - n^{-20}$,
at least this many entries have 
entries have magnitude at least $\kappa^{-3 \log{n}}$.
\end{proof}

Applying Lemma~\ref{lem:EigenvectorDense} to every principle minor
that remove one row and column of $A$,
and incorporating the dot product based lower bound on eigenvalue gap
from Lemma~\ref{lem:EigenGap} then gives the overall bound.

\begin{proof} (of Theorem~\ref{thm:PerturbA})
First, consider all the $s$ principle minors of $A$
formed by removing one row and column (which we denote by $k$),
\[
A_{\left[n\right]\setminus k, \left[n\right]\setminus k}.
\]
	
In order to invoke Lemma~\ref{lem:EigenGap},
we need to first lower bound the value of column $k$
when multiplied by all the eigenvectors of this principle minor,
\[
\abs{
	A_{\left[n\right]\setminus k, k}^{T}
	v\left(A_{\left[n\right]\setminus k, \left[n\right]\setminus k}, i\right)
}.
\]
Note that the entries in this matrix are also perturbed in the same
manner as the overall matrix,
Therefore, Lemma~\ref{lem:EigenvectorDense} gives that with probability
at least $1 - n^{-20}$, for each $i$, all entries of
$v(A_{[n]\setminus k, [n]\setminus k}, i)$
$\frac{n}{40\log{n}}$ entries with magnitude
at least $\kappa^{ - 3 \log{n}}$.
	
On the other hand, the entries of $A_{[n] \setminus k}$ are
each perturbed by $\frac{1}{n^2 \kappa} \normal(0, 1)$
with probability at least $\frac{30 \log{\kappa} \log{n}}{n}$.
This means the probability of one of these perturbations
occurring on a large magnitude entry of
$v(A_{[n]\setminus k, [n]\setminus k}, i)$ is at least
\begin{multline*}
1 - \left( 1 - \frac{300 \log{\kappa} \log{n}}{n} \right)
^{\frac{n}{40 \log{n}}}
\geq
1 - \exp\left( - \frac{300 \log{\kappa} \log{n}}{n}
\cdot \frac{n}{40 \log{n}}\right)\\
=
1 - \exp\left( - \frac{300 \log{\kappa}}{40} \right)
\geq
1 - n^{-20},
\end{multline*}
where the last inequality follows from $\kappa \geq n^{3}$.
Thus, we have that with probability at least
$1 - n^{-20}$, the variance in the dot product between
$A_{[n]\setminus k, k}$ and $v(A_{[n]\setminus k, [n]\setminus k}, i)$
is at least
\[
\kappa^{-3 \log{n}} \cdot \frac{1}{n^2 \kappa},
\]
which in turn implies with probability at least $1 - n^{-20}$,
this dot product is at least $\kappa^{-4 \log{n}}$ in magnitude.
	
Taking a union bound over all $n$ choices of $k$,
and all $n - 1$ eigenvectors of $A_{[n] \setminus k, [n] \setminus k}$
gives that with probability at least $1 - n^{-10}$, we have
\[
\abs{
	A_{\left[n\right]\setminus k, k}^{T}
	v(A_{\left[n\right]\setminus k, \left[n\right]\setminus k}, i)
}
\geq
\kappa^{ -4 \log{n}}
\qquad
\forall k, i.
\]

Then consider any $v(A, i)$.
Since it has norm $1$, there is some entry $k$ with
$|v(A, i)_{k}| \geq n^{-1/2}$.
Plugging this into Lemma~\ref{lem:EigenGap} gives:
\[
\abs{
	\lambda\left( A_{\left[s\right] \setminus k, \left[s\right] \setminus k} , i\right)
	-
	\lambda\left( A, i \right)
}
\geq
n^{-1/2}
\cdot
\kappa^{-4 \log{n}}
\geq
\kappa^{-5 \log{n}}.
\]
Combining this with interlacing from Lemma~\ref{lem:Interlacing},
namely
$
	\lambda( A, i )
	\leq
		\lambda( A_{[n] \setminus k, [n] \setminus k}, i )
	\leq
	\lambda( A, i + 1 )	
$
then lower bounds the gap between $\lambda_{i}(A, i)$
and $\lambda(A, i + 1)$ by the same amount.
\end{proof}

%% file: EpsNet.tex
\section{Condition Number of a Random Krylov Space}
\label{sec:RandKrylov}

In this section, we use $\epsilon$-nets to bound the condition number
of a Krylov space matrix generated from a random sparse matrix.
	
\RandKrylov*

Our overall proof structure builds upon the eigenvalue lower
bounds for random matrices with independent entries~\cite{SankarST03:journal,TaoV10}.
This approach is also taken in Lemma~\ref{lem:DensePerturb}
(with suboptimal parameter trade-offs)
to analyze the min singular value of a
matrix perturbed by dense Gaussian matrices.

The main difficulty we need to address when adapting this
approach is the dependence between columns of the Krylov space matrix.
For a particular $i$, the columns of $A^{i} G^{S}$
are still independent, but for a single column of $G^{S}$,
$g^{S}$ and two different powers $i_1$ and $i_2$,
$A^{i_1}g^{S}$ and $A^{i_2}g^{S}$ are dependent.
Such dependencies requires us to analyze the $m$ columns
generated from each column of $G^{S}$ together.
Specifically, the $m$ columns
$g^{S}, Ag^{S}, \ldots, A^{m - 1} g^{S}$ produced by each such vector $g^{S}$.

We then consider the space orthogonal to rest of the Krylov space,
which we denote using $W$.
Our matrix generalization for showing that the columns
corresponding to $g^{S}$ having large projection into $W$
is to lower bound the right singular values of the matrix
\[
W^T \left[ g^{S}, Ag^{S}, \ldots, A^{m - 1} g^{S} \right],
\]
or equivalently, showing that its products against all
length $m$ unit vectors have large norms.

The condition of $m \times s \leq n - 5 m$ represents the key
that allows the use of $\epsilon$ nets.
It implies that $W$ has at least $5m$ rows, and thus the
randomness in $g^{S}$ has more dimensions to show up
than the dimensions of length $m$ unit vectors.
In Section~\ref{subsec:OneColumn}, we transform the
product for some $\xhat$ in the $\epsilon$ net,
\[
W^T
\left[ g^{S}, Ag^{S}, \ldots, A^{m - 1} g^{S} \right]
\xhat
\]
into the product of a matrix with high numerical rank
(about $\Omega(m)$) against the non-zero entries of $g^{S}$.
This is done through the connection between Krylov
spaces and Vandermonde matrices, which we formalize
in the rest of this section and Appendix~\ref{sec:Vandermonde}.

However, for even a single vector to work,
we need to ensure that the columns of $W$ are not
orthogonal to $A^{i} g^{S}$.
Consider the following example: $A$ is diagonal,
and all the non-zeros in the columns of $W$ are in about $O(m)$ rows.
Then $A^{i} g^{S}$ is non-zero only if the non-zeros of $g^{S}$
overlap with the non-zero rows of $W$.
As we choose entry of $g^{S}$ is set to non-zero with probability
about $m^{3} / n$, this overlap probability works out to about
$O(m^{4} / n) < 0.1$.
In other words, $W^T [g^{S}, A g^{S}, \ldots A^{m - 1} g^{S}]$
would be $0$ with constant probability.

Therefore, we need to rule out $W$ that are `sparse'.
For $A$ with general eigenspace structures, this condition becomes any
vector in the column space of $W$ having large dot products with many
(about $n / poly(m)$) eigenvectors of $A$.
We will formalize the meaning of such density through the spectral
decomposition of $A$ in the rest of this section.
Then in Section~\ref{subsec:NullSpace} we use a multi-layer
$\epsilon$-net argument similar to the eigenvector
density lower bound from Lemma~\ref{lem:EigenvectorDense} to
conclude that any vector in the column span of $W$ has dense
projections in the eigenspaces.

We start by taking the problem into the eigenspaces
of $A$, and consider the effect of each
column of $G^{S}$ as a column block.
Since $A$ is symmetric, let its spectral decomposition be
\[
A = U^T \diag{\sigma} U
\]
where $\sigma$ is the list of eigenvalues.
Then for a single column $g^{S}$, the resulting vectors are
\[
\left[
\begin{array}{c|c|c|c|c}
U^T U g^{S}&
U^T \diag{\sigma}^1 U g^{S}&
U^T \diag{\sigma}^2 U g^{S}&
\ldots &
U^T \diag{\sigma}^{m - 1} U g^{S}
\end{array}
\right].
\]
Since $U^T$ is a unitary matrix, we can remove it
from consideration.
Furthermore, note that for any vectors $x$ and $y$,
\[
\diag{x} y
=
\diag{y}x
\]
so we can switch the roles of $\sigma$ and $U g^{S}$,
by defining the Vandermonde matrix $V \in \Re^{n \times m}$
with entries given by:
\[
V_{ij}
=
\sigma_{i}^{j - 1}.
\]
With this matrix, the term corresponding to $g^{S}$ can be
written as
\[
\diag{U g^{S}}
V.
\]
More globally, we are now considering a matrix formed by
taking $s$ copies (one per column of $G^{S}$) of $V$
with rows rescaled by $\diag{U g^{S}}$, and putting the columns
beside each other.

As mentioned before, $V$ is a Vandermonde matrix.
Its key property that we will use is below.
We will prove this for completeness in Appendix~\ref{sec:Vandermonde}.
\begin{restatable}{lemma}{Vandermonde}
\label{lem:Vandermonde}
Let $\sigma$ be a length $n$ vector with all entries in
the range $[\alpha, \alpha^{-1}]$ for some $\alpha < 1 /n$,
and any two entries of $\sigma$ at least $\alpha$ apart.
The $n \times m$ Vandermonde matrix with entries given as
\[
V_{ij}
=
\sigma_{i}^{j - 1}.
\]
has the property that for any unit vector $x$,
at least $n - m + 1$
entries in $Vx$ have magnitude at least $\alpha^{3m}$.
\end{restatable}

We will treat this property of $V$ as a black box.
Specifically, the main technical result that we will show
using $\epsilon$-nets is:

\begin{restatable}{lemma}{BlockSigmaMin}
\label{lem:BlockSigmaMin}
Let $V$ be an $n \times m$ matrix with:
\begin{enumerate}
    \item $m < n^{1/4}$,
    \item max entry magnitude at most $1$,
    \item for any unit $x$,
        $Vx$ has at least $n - m + 1$ entries with magnitude
        at least $\alpha_{V}$ for some $\alpha_{V} \leq n^{-5 \log{n}}$.
\end{enumerate} 
and let $U$ be an $n \times n$ orthonormal matrix.
Let $G^{S}$ be an $n \times s$ sparse Gaussian matrix
where each entry is set i.i.d. to $\normal(0, 1)$
with probability $h / n$,
where $s \geq \frac{n}{2 m}$,
$h \geq 10000 m^2 \log(1/\alpha_{V})$,
and $W$ be an $n \times r$ orthonormal matrix
orthogonal to $G^{S}$ for some $r \geq 5 m$, i.e., $W^T G^{S} = 0$.
    
Then for another column vector $g^{S}$ generated from
this sparse Gaussian distribution with density $h$,
we have with probability at least $1 - n^{-3}$:
\[
    \sigma_{\min} \left( W^T \diag{U g^{S}} V\right)
    \geq
    \alpha_{V}^{4}.
\]
\end{restatable}

We will prove this bound in the next subsection
(Section~\ref{subsec:OneColumn}).
Before we do so, we first formally verify that this local bound
suffices for globally lower bounding the min
singular value of the block Krylov space matrix.

\begin{proof}(of Theorem~\ref{thm:RandKrylov})

Claim~\ref{claim:GaussianMax} allows us to assume that the
max magnitude of an entry in $G^{S}$ is at most $n$.
Applying Lemma~\ref{lem:MaxMagnitude} inductively on
$A^{i} G^{S}$ then gives that the magnitudes of all entries
in each such matrix is at most $n$ as well.
Thus the max magnitude of an entry in $K$ is at most $n$,
and the max singular value of $K$ is at most $n^2$.

We now turn attention to the min singular value.
Since $U$ is unitary,
it suffices to lower bound the min singular vlaue of
\[
\Khat
=
U K
=
\textsc{HCat}_{1 \leq j \leq s}
\left(\diag{U G^{S}_{:, j}} V\right).
\]
where $\textsc{HCat}(\cdot)$ denotes horizontal concatenation
of matrices (with the same numbers of rows),
and $V$ is $m$-step Vandermonde matrix generated from the eigen
values of $A$.
Since the Frobenius norm of $A$ is at most $1$, all eigenvalues,
and all their powers, have magntiude at most $1$ as well.
Combining this with Lemma~\ref{lem:Vandermonde} then gives
that $V$ has max entry-wise magnitude at most $1$,
and for any unit vector $x$, $Vx$ has at least $n - m + 1$
entries with magnitude at least $\alpha_{A}^{3m}$.

We denote the block of $m$ columns corresponding to
$G^{S}_{:, j}$ as $[j]$.
Taking union bound over all column blocks,
we get that Lemma~\ref{lem:BlockSigmaMin} holds
for each of the blocks.
Under this assumption, we proceed to lower bound
bound $\norm{\Khat x}_2$ for all unit vectors $x$.
For each such vector, because there are at most $s$ blocks,
there is some block $x[j]$ with $\norm{x[j]}_2 \geq s^{-1/2}$.

The energy minimization extension definition of Schur complements
(e.g. Appendix A.5.5. of~\cite{BoydV04:book}) gives
\[
x^T \left( \Khat^T \Khat \right) x
\geq
x\left[j\right]^{T}
\textsc{Sc}\left( \Khat^T \Khat, \left[j\right]\right)
x\left[j\right]
\]
This Schur complement of $\Khat^T \Khat$ onto the block $[j]$ can in turn
be written as
\[
\Khat_{:, \left[j\right]}^T \Khat_{:, \left[j\right]}
-
\Khat_{:, \left[j\right]}^T \Khat_{:, \overline{\left[j\right]}}
    \left( \Khat_{:, \overline{\left[j\right]}}^{T}
			\Khat_{:, \overline{\left[j\right]}} \right)^{\dag}
\Khat_{:, \overline{\left[j\right]}}^{T} K_{:, \left[j\right]}
=
\Khat_{:, \left[j\right]}^T
    \left( I -
		\Khat_{:, \overline{\left[j\right]}}
		\left( \Khat_{:, \overline{\left[j\right]}}^{T}
		\Khat_{:, \overline{\left[j\right]}} \right)^{\dag}
		\Khat_{:, \overline{\left[j\right]}}^{T}
		\right)
\Khat_{:, \left[j\right]}
\]
where $\dag$ denotes the pseudo-inverse.

Note that the middle term is precisely a projection matrix
onto the space orthogonal to the columns of $K_{:, \overline{[j]}}$.
So in particular, picking $W(\overline{[j]})$ to be an
orthogonal basis of a subset of this space can only decrease the operator:
\[
\textsc{Sc}\left( \Khat^T \Khat, \left[j\right]\right)
\succeq
\Khat_{:, \left[j\right]}^{T}
W\left( \overline{\left[j\right]}\right)
W\left( \overline{\left[j\right]}\right)^{T}
\Khat_{:, \left[j\right]}
\]
Applying this matrix inequality to the vector $x[j]$ then gives:
\begin{multline*}
\norm{\Khat x}_2
\geq
\sqrt{
x\left[j\right]^{T}
\Khat_{:, \left[j\right]}^{T}
W\left( \overline{\left[j\right]}\right)
W\left( \overline{\left[j\right]}\right)^{T}
\Khat_{:, \left[j\right]}
x\left[j\right]
}
=
\norm{W\left(\overline{\left[j\right]}\right)^{T}
				\Khat_{:, \left[j\right]} x\left[j\right]}_2\\
\geq
\sigma_{\min} \left( W\left(\overline{\left[j\right]}\right)^{T}
				\Khat_{:, \left[j\right]}\right)
\cdot
\norm{x\left[j\right]}_2
\geq
n^{-1/2}
\cdot
\sigma_{\min} \left( W\left(\overline{\left[j\right]}\right)^{T}
				\Khat_{:, \left[j\right]}\right).
\end{multline*}
Then the result follows from the lower bound on the min
singular value from Lemma~\ref{lem:BlockSigmaMin}.
\end{proof}

\subsection{Least Singular Value of Krylov Space of One Column}
\label{subsec:OneColumn}

We now prove the lower bound on min singular value
of a single block as stated in Lemma~\ref{lem:BlockSigmaMin}.

\BlockSigmaMin*

The proof will be via an $\epsilon$-net argument.
We want to show that for all $x \in \Re^{m}$,
$\norm{W^T \diag{U g^{S}} Vx}$ is large.
To do so, we will enumerate over all vectors,
with entries rounded to nearest multiple of $\epsilon$.
We will denote these rounded vectors with the hat superscript,
i.e., $\xhat$ for $x$.

By claim~\ref{claim:GaussianMax}, we may assume that
the entries in $Ug^{S}$ have magnitude at most $2\sqrt{n}$.
Combining this with the fact that $W$ is orthonormal,
and the given condition on max magnitude of $V$ gives
that for vectors $x$ and $\xhat$ such that
$\norm{x - \xhat}_{\infty} \leq \epsilon$, we have
\[
\norm{W^T \diag{Ug^{S}} V \xhat
-
W^T \diag{Ug^{S}} V x
}
\leq
n^3 \epsilon.
\]
So we can consider a finite number of such $\xhat$
vectors, as long as $\epsilon$ is set to be smaller than
the lower bound on products we want to prove.

For such a vector $\xhat$, we let
\[
\yhat = V \xhat.
\]
As $\frac{1}{2} \leq \norm{\xhat}_2 \leq 2$,
the max magnitude of an entry in $\yhat$ is at most $2n$
by the bound on magnitude of entries of $V$;
and at least $n - m$ entries in $\yhat$ have magnitude
at least $\alpha_{V} / 2$ by the given condition on $V$.

We want to show for a fixed triple of $W^T$, $U$, and $\yhat$,
it's highly unlikely for $g^{S}$
to result in a small value of $\norm{W^T \diag{U g^{S}} \yhat}$.
Consider the effect of a single entry of $g^{S}$, $g^{S}_{i}$ on the result.
The vector produced by $g^{S}_{i}$ in $U^{T} g^{S}$ is
\[
\left( U^{T} \right)_{:, i}
=
\left( U_{i, :} \right)^{T}.
\]
This vector is then multiplied entry-wise by $\yhat$:
\[
\diag{U_{i, :}^{T}} \yhat
= 
\diag{\yhat} \left( U_{i, :} \right)^{T},
\]
and then multiplied against $W^{T}$.
So the contribution of $g^{S}_{i}$ to the overall sum
is a vector in $\Re^{k}$ given by
\[
    g^{S}_i \cdot W^{T} \diag{\yhat} U_{i, :}^{T}.
\]
	
Thus, once we fix $\xhat$ and $\yhat$ from the $\epsilon$-net,
we can form the $\Re^{k \times n}$ matrix that directly measures
the contribution by the entries in $g^{S}$.
\[
    W^{T} \diag{\yhat} U^{T}.
\]
So our goal is to prove that a random subset of columns
of this matrix picked independently with probability
$h / n$ per column has high numerical rank.
	
To do so, we will show that for a fixed $\yhat$, any vector
in the span of $W$ is likely to have non-zero dot products
with most columns of the form of $\diag{\yhat} U_{i, :}^{T}$.
This relies on proving a global lower bound on the density
of null space vectors of $G^{S}$.
So we will carry it out as a global argument in 
Section~\ref{subsec:NullSpace}

\begin{restatable}{lemma}{NullSpace}
\label{lem:NullSpace}
Let $U$ be an $n$-by-$n$ orthornomal basis,
$\Ycal$ be a family of length $n$ vectors with magnitude at most
$n$, at least $n - m$ entries
with magnitude at least $\alpha_{Y} < n^{-4 \log{n}}$ for some $m < n^{\frac{1}{4}}$,
and let $G^{S}$ be an $n \times d$ sparse Gaussian matrix
with each entry set to $\normal(0, 1)$ with probability $h/n$
for some $h$ such that $dh > 20 n \log{n} \log(1 / \alpha)$
and $d > 80 m$.
Then with probability at least 
\[
1 - n^{-9} - \abs{\Ycal} \alpha_{Y}^{\frac{d}{10}},
\]
the matrix $G^{S}$ has the property that 
for any vector $w$ orthogonal to the columns of $G^{S}$,
and any vector $\yhat \in \Ycal$, the vector
$w^{T} \diag{\yhat} U^{T}$ has at least $\frac{d}{80}$
entries with magnitude at least $\alpha_{Y}^3$.
\end{restatable}

We then convert this per-entry bound to an overall bound on
the minimum
singular value of $W^{T} \diag{y} U^{T}_{:, S}$, where $S$
is the subset of non-zeros picked in $g^{S}$.
This is once again done by union bound over an $\epsilon$-net.
Here the granularity of the net is again dictated by
the minimum dot product that we want to show.
This value is in turn related to the entry-wise magnitude
lower bound of the matrix that we can assume, which is $\alpha_Y^{3}$.

\begin{lemma}
\label{lem:SampleColumns}
    Let $Y$ be an $r \times n$ matrix with $r \leq n$ and per-entry
    magnitude at most $n^{10}$ such that for any unit vector
    $z \in \Re^{r}$, the vector $z^T Y$ has at least 
    $t$ entries with magnitude at least $\alpha_{Y} < n^{-12}$.
    Then a random sample of the columns of $Y$ with
    each column chosen independently with probability
    $\frac{h}{n}$ gives a subset $S$ such that for all
    unit $z \in \Re^{k}$,
    \[
    \norm{z^T Y_{:, S}}_{2} \geq \frac{\alpha_{Y}}{2}
    \]
    with probabitliy at least
    \[
    1 - \alpha_{Y}^{-2r} \exp\left( - \frac{t \cdot h}{n} \right)
    \]
\end{lemma}

\begin{proof}

The magnitude upper bound of $Y$
ensures for any $z$ the rounded vector $\zhat$ has
\[
\norm{z^T Y_{:, S} - \zhat^T Y_{:, S}}_{2}
\leq
\norm{z  - \zhat}_{2}
\cdot
\norm{Y_{:, S}}_2
\leq
\sqrt{n} \norm{z - \zhat}_{\infty}
\cdot
\sqrt{n}
\normi{Y_{:, S}}
\leq
n^{11} \norm{z - \zhat}_{\infty}
\]
where the second last inequality follows from bounding
the $\ell_2$ norm by $\ell_{\infty}$ norms.
So an $\epsilon$-net over $z$ with granularity
\[
\epsilon_{Z} = \frac{\alpha_{Y}}{2 n^{11}}.
\]
ensures that for any $z$, there is some $\zhat$ in the net such that
\[
\norm{z^T Y_{:, S} - \zhat^T Y_{:, S}}_{2}
\leq
\frac{\alpha_{Y}}{2}.
\]

So it suffices to bound the probability that for
all $\zhat$, $S$ contains at least one of the large
entries in the vector $\zhat^T Y$.
As each of the $t$ large entries is picked with probability
$h/n$, the probability that none gets picked is
\[
\left(1 - \frac{h}{n} \right)^{t}
\leq
\exp\left( - \frac{ht}{n} \right).
\]
While on the other hand,
because every entry of a unit vector $z$ has value between $-1$
and $1$, the rounding to granularity $\epsilon_Z$ gives at most
\[
\frac{2}{\epsilon_{Z}}
\leq
4 n^{11} \alpha_{Y}^{-1}
\leq
\alpha_{Y}^{-2}
\]
values per coordinate.
Multiplied over the $r$ coordinates of $z$ then gives
that the total size of the $\epsilon$-net is at most $\alpha_{Y}^{-2r}$.
Taking a union bound over all these vectors then gives the overall bound.
\end{proof}

This ensures that the sub-matrix corresponding to the non-zero
entries of $g^{S}$ is well conditioned.
We can then apply {\em restricted invertibility}~\cite{BourgainT87}
to obtain a subset of columns $J \subseteq S$ such that
the covariance of those Gaussian entries is well conditioned.
The formal statement of restricted invertibility that we will
use is from~\cite{SpielmanS12}.
\begin{lemma}
	\label{lem:RestrictedInvertibility}
	Let $Y$ be a $n \times r$ matrix such that
	$Y^{T} Y$ has minimum singular value at least $\xi$.
	Then there is a subset $J \subseteq [n]$ of size
	at least $r / 2$ such that
	$Y_{J, :}$ has rank $|J|$, and
	minimum singluar value at least $\frac{\xi}{\sqrt{10n}}$,
	or equivalently:
	\[
	Y_{J, :} Y_{J, :}^{T}
	\succeq
	\frac{\xi^2}{10 n}I.
	\]
\end{lemma}

\begin{proof}
The isotropic, or $Y^{T} Y = I$ 
case of this statement is Theorem 2 from \cite{SpielmanS12}
instantiated with:
\begin{itemize}
	\item $v_{i}$ being the $i$th row of $Y$,
	\item $n \leftarrow s$, $m \leftarrow n$,
	\item $L = I$,
	\item $\epsilon = 1/2$.
\end{itemize}
	
For the more general case, consider the matrix
\[
M \leftarrow  Y^{T} Y.
\]
and in turn the matrix
$Y M^{-1/2}$.
This matrix satisfies
\[
\left( Y M^{-1/2} \right)^{T} Y M^{-1/2}
=	M^{-1/2} Y^{T} Y M^{-1/2}
= M^{-1/2} M M^{-1/2}
= I.
\]
So we can apply the special case of restricted invertibility
(for the case where the Gram matrix is identity)
mentioned above to get a row subset $J$ such that:
\[
\frac{1}{10 n}
\preceq
Y_{J, :} M^{-1/2} \left( Y_{J, :} M^{-1/2} \right)^{T}\\
=
Y_{J, :} M^{-1} Y_{J, :}^{T}
\]

On the other hand, the singular value bound on $Y$ implies
$M \succeq \xi^2 I$ and therefore
$M^{1/2} \succeq \xi I$.
So we get for any unit vector $u \in \Re^{|J|}$,
\[
\norm{Y_{J, :}^{T} u}_2
\geq
\xi \norm{M^{-1/2} Y_{J, :}^{T} u}_2
\geq
\xi \sqrt{\frac{1}{10 n}},
\]
which when squared gives the desired bound on singular values.
\end{proof}

\begin{proof} (of Lemma~\ref{lem:BlockSigmaMin})
As adding to columns of $W$ can only increase $\norm{W^T z}$,
it suffices to consider the case of $r = 5m$.
Consider taking an $\epsilon$-net over all unit $x \in \Re^{m}$
with granularity
\[
\epsilon_{x} \geq n^{-10} \alpha_{V}^3.
\]
The size of this net is at most
\[
\left( \frac{2}{n^{-10} \alpha_V^{3}} \right)^{m}
\leq
\alpha_{V}^{-4m}.
\]
Let the associated set of $\yhat = V \xhat$
vectors be $\Ycal$.

Then Lemma~\ref{lem:NullSpace} gives that with probability
at least
\[
1 - n^{-9} - \alpha_{V}^{-4m}
\alpha_{V}^{\frac{d}{10}}
\geq
1 - n^{-8}.
\]
the $W$ matrix orthogonal to $G^{S}$ has the property that
for all $w$ in its column span, and for all $\yhat \in \Ycal$, 
\[
w^T \diag{\yhat} U^T
\]
has at least
\[
\frac{s}{80}
\geq 
\frac{n}{160 m}
\]
entries with magnitude at least $\alpha_{V}^2$,
Here the last inequality in the probability bound
follows from the assumption of $m < n^{1/4}$.

Consider each individual $\yhat$.
Let $Y(\yhat)$ be the associated matrix
\[
Y\left(\yhat\right)
=
W^T \diag{\yhat} U^{T}
\]
Lemma~\ref{lem:SampleColumns} gives that the (globally)
picked subset $S$ give that
$W^T \diag{\yhat} U^{T}_{:, S} = Y(\yhat)_{:, S}$
has minimum singular value
at least $\alpha_V^2/2$ with probability at least
\[
1 - \left( 2\alpha_{V} \right)^{-4m} \exp\left( - \frac{h}{n} \cdot \frac{n}{160 m} \right)
=
1 - \alpha_{V}^{-4m} \exp\left( - \frac{h}{200 m} \right).
\]
Should such bound hold,
by restricted invertibility as stated in Lemma~\ref{lem:RestrictedInvertibility},
we get that there is a set $J(\yhat)$ with size
at least $r / 2$ such that
\[
Y\left(\yhat\right)_{:, J\left( \yhat \right)}^T
Y\left(\yhat\right)_{:, J\left( \yhat \right)}
\succeq
\frac{\alpha_{V}^2}{10 n} I.
\]

This means it suffices to bound the probability over the random choice of
$g^{S}_{S}$ which lead to $\norm{Y(\yhat)g^{S}_{S}}_2$ being small.
Specifically, we can do a worst-case bound over entries
not in $S$, call that vector $b$, to get:
\[
\prob{G^{S}_{S}}
{\norm{Y\left(\yhat\right)_{:, S} g^{S}_{S}}_{2} \leq \alpha_{V}^4}
\leq
\max_{b}
\prob{G^{S}_{J\left( \yhat \right)}}
{\norm{Y\left(\yhat\right)_{:, J\left( \yhat \right)}
    g^{S}_{J\left( \yhat \right)} - b}_{2} \leq \alpha_{V}^4}.
\]
multiplying by $Y(\yhat)_{:, J( \yhat )}^T$
and $(Y(\yhat)_{:, J( \yhat )}^T Y(\yhat)_{:, J( \yhat )})^{-1}$
then gives
\begin{multline*}
\norm{Y\left(\yhat\right)_{:, J\left( \yhat \right)}
    g^{S}_{J\left( \yhat \right)} - b}_{2}
\geq
n^{-10}
\norm{
Y\left(\yhat\right)_{:, J\left( \yhat \right)}^T
Y\left(\yhat\right)_{:, J\left( \yhat \right)}
    g^{S}_{J\left( \yhat \right)} -
Y\left(\yhat\right)_{:, J\left( \yhat \right)} b
}_{2}\\
\geq
\alpha_V^{3}
\norm{
    g^{S}_{J\left( \yhat \right)} -
\left(Y\left(\yhat\right)_{:, J\left( \yhat \right)}^T
Y\left(\yhat\right)_{:, J\left( \yhat \right)}\right)^{-1}
Y\left(\yhat\right)_{:, J\left( \yhat \right)} b
}_{2}.
\end{multline*}
Here the first inequality follows from the magnitude upper
bound on $Y(\yhat)_{:, J( \yhat )}^T$, and the bound on
min singular value above.

Substituting this lower bound back in,
with a different choice of the worst-case vector, gives
\[
\prob{g^{S}_{S}}
{\norm{Y\left(\yhat\right)_{:, S} g^{S}_{S}}_{2} \leq \alpha_{V}^4}
\leq
\max_{\bhat}
\prob{g^{S}_{J\left( \yhat \right)}}
{\norm{ g^{S}_{J\left( \yhat \right)} - \bhat}_2
\leq
\alpha_{V}^{-3} \cdot \alpha_{V}^{4}}
\]
Simplifying to coordinates gives
\[
\leq
\prod_{j \in J\left( \yhat \right)}
\max_{\bhat_{j}}
\prob{g^{S}_{j}}
{\abs{g^{S}_j - \bhat_{j}} \leq \alpha_{V}}
\leq
\alpha_{V}^{\frac{r}{2}}.
\]
Here the last inequality follows from the density of
a standard Gaussian being at most $1$.

So the overall failure probability is, by union bound, at most
\[
\alpha_{V}^{-4m}
\cdot \alpha_{V}^{r}
+
\alpha_{V}^{-5r}
\exp\left( - \frac{h}{200 m } \right),
\]
where the first term is from  union bounding over the entire net
of the above probability,
and the second term follows from the invocation of Lemma~\ref{lem:SampleColumns}.

Substituting in $r = 5m$ gives that the first term
is at most $\alpha_{V}^{m} \leq n^{-20}$.
For the second term,
incorporating $h \geq 10000 m^2 \log ( 1 / \alpha_{V})$ gives
\[
\alpha_{V}^{-5r}
\exp\left( - \frac{h}{200 m } \right)
\leq
\alpha_{V}^{-25 m}
\cdot
\alpha_{V}^{50 m}
\leq
n^{-20}.
\]
Summing these then gives the overall failure probability.
\end{proof}

\subsection{Density Lower Bound on Vectors in Null Space}
\label{subsec:NullSpace}

It remains to rule out all vectors such that
$w \diag{\yhat} U^{T}$ is sparse for some $\yhat$
corresponding to some $\xhat$ in the $\epsilon$-net.
We do so by leveraging the initial Krylov block $G^{S}$.
We show that any particular vector is highly unlikely
to be orthogonal to all $s$ columns of $G^{S}$.
For a dense $n$-by-$s$ Gaussian $G$, the probability of
any particular vector being (nearly) orthogonal to all
columns is about $\exp(-O(s))$ due to all columns being
independent.
We will show that sparse Gaussians also produce a similar
behavior.
The proof then utilizes this per-vector bound together
with an $\epsilon$-net argument on vectors that
are sparse in their representation under the basis $U$,
with an additional case to take the few small entries of
$\yhat$ into account.

\NullSpace*

We first show that $G^{S}$ rules out any particular
vector with reasonably large probability.
This is again done in two steps: first by ruling out
all sparse vectors using a successive $\epsilon$-net
argument similar to the proof of Lemma~\ref{lem:EigenvectorDense}.
This proof is slightly simplified due to all entries of $G^{S}$
being completely independent, instead of correlated across the
diagonal as in Lemma~\ref{lem:EigenvectorDense}.

\begin{lemma}
\label{lem:SparseGaussianNullSpace}
Let $G^{S}$ be an $n \times d$ sparse Gaussian
matrix with each entry set to $\normal(0, 1)$ with probability $h / n$,
and $0$ otherwise.
If $hd \geq 20 n \log^2{n}$,
then with probability at least $1  - n^{-10}$, all unit vectors
$w$ satisfying $(G^{S})^{T} w = 0$
(i.e., orthogonal to all columns $G^{S}$)
has at least
$\frac{d}{40 \log{n}}$ entries with magnitude at least $n^{-4 \log{n}}$.
\end{lemma}

\begin{proof}
By Claim~\ref{claim:GaussianMax}, we may assume that all entries in
$G^{S}$ have magnitude at most $n$.
		
We will prove by induction for $i = 0 \ldots \lfloor \log_{2}(\frac{d}{40 \log{n}})\rfloor$
that with probability at least $1 - i n^{-11}$,
all unit vectors in the null space of $G^{S}$ has at least $2^{i}$
entries with magnitude at least $n^{-4 (i + 1)}$.
		
The base case of $i = 0$ follows from a length $n$ unit vector having
an entry of magnitude at least $n^{-1/2}$.
		For the inductive case, we will build an $\epsilon$-net with granularity
\[
\epsilon_i \leftarrow n^{-4 \left(i + 1 \right)}.
\]
That is, we only consider vectors whose entries are integer
multiples of $\epsilon_i$.
For a generic vector $w$, we can round each entry of it toward
$0$ to form $\what$ such that:
\begin{itemize}
	\item $\norm{\what}_{2} \leq 1$.
	\item $(G^{S})^T \what$ is entry-wise small:
	\[
	\norm{(G^{S})^T \what}_{\infty}
	\leq
	n \norm{G^{S}}_{\infty}
	\norm{w - \what}_{\infty}
	\leq
	n^2 \epsilon_i.
	\]
\end{itemize}

So it suffices to show that the probability of $G^{S}$
having an entry-wise small product with any $\what$
with at most $2^{i}$ non-zeros is small.
This is because any entry with magnitude less than $\epsilon_i$
will get rounded to $0$.
Furthermore, by the inductive hypothesis, it suffices to consider
only $w$ with at least $2^{i - 1}$ entries with magnitude at least
$\epsilon_{i - 1}$.
Each such entry, when perturbed by $\epsilon_i$, has magnitude
at least $\epsilon_{i - 1} - \epsilon_{i} \geq \epsilon_{i - 1} / 2$.
	
We will do so by union bound over all such vectors $\what$.
Because the columns of  the probability that any column of $G^{S}$
picks none of $t$ entries is at most
\[
\left( 1 - \frac{h}{n} \right)^{2^{i - 1}}
\leq
\exp\left( - \frac{h \cdot 2^{i - 1}}{n} \right).
\]	
Furthermore, if one of these entries are picked,
the resulting Gaussian corresponding to the product of that column of
$G^{S}$ against $\what$ has variance at least $\epsilon_{i - 1} / 2$.
Which means that it's in an interval of size at most $n^{3} \epsilon_i$
with probability at most
\[
\frac{n^{2} \epsilon_{i}}
{\epsilon_{i - 1} / 2}
\leq
n^{-1},
\]
where the inequality follows from the choice of
$\epsilon_{i} = n^{-4} \epsilon_{i - 1}$.
Taking union bound over these two events gives that the probability
of a column of $G^{S}$ having small dot product against $\what$ is at most
\[
\exp\left( - \frac{h 2^{i - 1}}{n} \right)
+
n^{-1}
\leq
2 \exp\left( -
\min\left\{\frac{h \cdot 2^{i - 1}}{n}, \log{n} \right\}
\right)
\]
which compounded over the $d$ columns gives,
and substituting in the assumption of $h \cdot d \geq 20 n \log^2{n}$
gives an overall probability of at most
\[
2\exp\left( -
\min
\left\{\frac{d \cdot h \cdot 2^{i - 1}}{n}, d \log{n}  \right\}
\right)
\leq
2\exp\left( -
\min\left\{ 10 \log^{2}{n} \cdot 2^{i}, d \log{n} \right\}
\right).
\]

On the other hand, the number of vectors with $2^{i}$
non-zeros, norm at most $2$, and entries rounded to integer
multiplies of $\epsilon_i$ is at most
\[
\binom{n}{2^{i}}
\cdot
\left(4 / \epsilon_i \right)^{2^{i}}
\leq
\left( 4n / \epsilon_i \right)^{2{i}}
\leq
\exp \left( 5 \cdot \left(i + 1 \right) \cdot \ln{n} \cdot 2^{i} \right)
\leq
\exp \left( 5 \log^2{n} \cdot 2^{i} \right).
\]
		
Matching this against the two terms means we need:
\begin{itemize}
\item For any $i \geq 0$, we have
    $5 \log^{2}n 2^{i} \leq 5 \log^2{n} 2^{i}$,
	and this term is minimized when $i = 1$.
So the first term is at most $2 \exp(- 5 \log^{2}n) \leq n^{-10}$.
\item For the second term to be small, substituting in
$2^{i} \leq \frac{d}{40 \log{n}}$ gives
\[
\exp \left( 5 \log^2{n} \cdot 2^{i} \right)
\cdot
2 \exp\left( - d \log{n} \right)
\leq
2\exp\left( \frac{d \log{n}}{8} - d \log{n} \right)
\leq
2\exp\left( - \frac{7 d \log{n}}{8} \right),
\]
which is at most $n^{-10}$ when $d$ is larger than some absolute constant.
\end{itemize}
\end{proof}

This means under a global event that happens with
probability at least $1 - n^{-10}$,
we only need to consider dense vectors in the column
space of $W$.
For each such vector, we can use its density to show that
it's highly unlikely to have small product against $G^{S}$.

\begin{lemma}
\label{lem:ClobberOneVector}
Let $0 < \alpha < n^{-8 \log{n}}$ be a threshold,
and let $G^{S}$ be a $n \times d$ sparse Gaussian matrix with
each entry set to $\normal(0, 1)$ with probability $h / n$
for some $hd > 20 n \log{n} \log(1 / \alpha)$.
Then any unit vector $w$ with at least $\frac{d}{40 \log{n}}$ entries
with magnitude at least $n^{-4 \log{n}}$ satisfies
\[
\prob{G^{S}}{\norm{\left( G^{S} \right)^{T} w}_2 < \alpha}
<
\alpha^{\frac{d}{5}}.
\]
\end{lemma}
	
\begin{proof}
Consider each column of $G^{S}$, $g^{S}$.
The probability that $g^{S}$ picks none of the large entries in $w$
is at most
\[
\left( 1 - \frac{h}{n} \right)^\frac{d}{40 \log{n}}
=
\exp\left( - \frac{h d}{n \cdot 40 \log{n}} \right)
\leq
\exp\left( - \frac{20 n \log{n} \log\left(1 / \alpha\right) }
  {40 n \log{n}} \right)
=
\alpha^{\frac{1}{2}}.
\]

In the case such an entry is picked, the resulting product with the
Gaussian has variance at least $n^{-4 \log{n}}$.
So is in an interval of size $\alpha$ with probability at most
\[
\frac{\alpha}{n^{-4 \log{n}}}
\leq
\alpha^{\frac{1}{2}},
\]
where the last inequality follows from the assumption
of $\alpha < n^{-8 \log{n}}$.
By union bound, the probability of $g^{S}$ having small
dot product against $w$ is at most $\alpha^{1/5}$.
Compounding this over the $d$ columns then gives the overall bound.
\end{proof}

The rest of the proof is an $\epsilon$-net based argument
on all $w$ for which $w^T \diag{\yhat} U$ is sparse.
For each $\yhat$, we want to generate some set of vectors $\Wcal(\yhat)$
such that if $w$ is a vector where $w^T \diag{\yhat} U^{T}$ is sparse,
there is some $\what \in \Wcal$ such that
\[
\norm{\what - w}_2 \leq \alpha_{Y}.
\]
After that, it suffices to show that all $\what \in \Wcal(\yhat)$
has large $\norm{\what^T G^{S}}_2$ via Lemma~\ref{lem:ClobberOneVector}.
Since $U$ is invertible, most of $w$ is recoverable from
the $w^T \diag{\yhat} U$ vector via the operation
\[
\left( w^T \diag{\yhat} U \right) U^{T}
=
w^T \diag{\yhat}.
\]
So up to a small number of coordinates corresponding to the small
magnitude entries in $\yhat$, we can enumerate over the possible
$\what$s by enumerating over the possible sparse $w^T U$ vectors.
For the non-zero coordinates in both the original and spectral domains,
it also suffices to consider entries that are integer multiples of poly$(\alpha_{Y})$.

\begin{proof}(of Lemma~\ref{lem:NullSpace})
Claim~\ref{claim:GaussianMax} allows us to assume that the maximum
magnitude in $G^{S}$ is $n$.

We will use an $\epsilon$-net argument,
with the goal of invoking Lemma~\ref{lem:ClobberOneVector} on all
vectors in the net.
In order to do so, we first invoke Lemma~\ref{lem:SparseGaussianNullSpace},
and pay for the global failure probability of $n^{-10}$ once
for all vectors $\yhat$.

We also generate the set of vectors that
are $t$-sparse in the $U$ basis representation, with granularity
\[
\epsilon_{q}
\leftarrow
\alpha_{Y}^3.
\]
Formally, let $\Qcal$ denote all the vectors $\qhat$ with norm
at most $2$ and at most $t$ non-zeros,
all of which are integer multiples of $\epsilon_{q}$.

Now consider some vector $\yhat \in \Ycal$.
Let $BIG$ be the subset of entries in $y$ that are
at least $\alpha_{Y}$, and $\overline{BIG}$ its complement.
We generate $\Wcal(\yhat)$ by considering all vectors of the
form
\begin{align*}
    \what_{BIG} & = \diag{\yhat_{BIG}}^{-1} U^T \qhat
    \qquad \text{for some $\qhat \in \Qcal$}\\
    \what_{\overline{BIG}} & = \text{integer multiples of $\alpha_{Y}^2$}
\end{align*}

Then we want to show that any $w$ for which
$w^T \diag{\yhat} U^T$ is $t$-sparse is close to some
$\what \in \Wcal(\yhat)$.
For such a $w$, consider the vector
\[
q = U \diag{\yhat} w,
\]
and suppose we rounded its entries to the nearest multiples of
$\epsilon_{q}$ for some $\epsilon_{q}$, giving $\qhat$
such that $\norm{\qhat - q}_{\infty} \leq \epsilon_{q}$.
Because $U$ is an orthonormal matrix,
any error in $q$ translates to an error in $U^T q - \diag{\yhat} w$ as well:
\[
\norm{U^T \qhat - \diag{\yhat} w }_2
\leq
\norm{\qhat - U \diag{\yhat} w}_2
=
\norm{\qhat - q}_{2}
\leq
n^{1/2} \epsilon_{q}.
\]
This error can in turn be carried across all entries in
$\yhat_{BIG}$, using the entry-wise lower bounds.
Specifcally, for all $i \in BIG$ we have
\[
\abs{\left( U^T \qhat \right)_{i} \yhat_{i}^{-1} - w_{i} }
\leq
\abs{\yhat_{i}^{-1}}
\abs{\left( U^T \qhat \right)_{i} - \yhat_{i} w_{i} }
\leq 
\alpha_{Y}^{-1} \cdot 
\norm{U^T \qhat - \diag{\yhat} w}_{2}
\leq
\alpha_{Y}^{-1} n^{1/2} \epsilon_{q}.
\]
Thus the $\what$ that corresponds to this $w$ is the one
from this $\qhat$, plus having all entries in $\overline{BIG}$
rounded explicitly.
That is, for any such $w$ with $w^T \diag{\yhat} U^T$
$t$-sparse, there is some $\what \in \Wcal(\yhat)$ such that
\[
\norm{w - \what}_{2}
\leq
n^2 \alpha_Y^{-1} \epsilon_{q}
\leq
\frac{\alpha_{Y}}{n^3},
\]
where the inequality follows from the choice of $\epsilon_{q}$
and the assumption of $\alpha_{Y} \leq n^{-8 \log{n}}$.
Combining with the assumption of max magnitude in $G^{S}$
being $n$ from Claim~\ref{claim:GaussianMax} also gives
\[
\norm{G^{S} w - G^{S} \what}_{2}
\leq
\frac{\alpha_{Y}}{2}.
\]

Thus, to rule out all such $w$, it suffices to show that
all $\what \in \Wcal(\yhat)$ have $\norm{(G^{S})^T \what}_{2} \geq \alpha_{Y}$.
We do so by taking union bound over the entire $\epsilon$-net.

Since $w$ is a unit vector, $U^T w$ also has norm at most $1$.
Combining this with the assumption of the entries of
$\yhat$ having magnitude at most $n$ gives that the max
magnitude of an entry in $\qhat$ is at most $2n$.
As each of the $t$ non-zeros in $\overline{BIG}$ is explicitly
enumerated with magnitude at most $\epsilon_{q} = \alpha_{Y}^3$,
we have:
\[
\abs{\Qcal}
\leq
\binom{n}{t} \cdot \left( 4 n \alpha_{Y}^{-3} \right)^{t}
\leq
\left( 4 n^2 \alpha_{Y}^{-3} \right)^{t}
\leq
\alpha_{Y}^{-4t}
\]
which combined with the $m$ entries of $\what$
being explicitly generated as multiples of $\epsilon_{q}$ gives
\[
\abs{\Wcal\left( \yhat \right)}
\leq
\alpha_{Y}^{-4 \left( t + m \right)}.
\]
We remark that the key in this step is that the overhead
from generating terms in $\what$ has $m$ in the exponent
instead of $n$.

As we've already globally conditioned on all vectors in the null
space of $G^{S}$ being dense, we get that each vector
$\what \in \Wcal$ gives a small dot product with probability
at most $\alpha_{Y}^{-\frac{d}{5}}$.
Taking union bound over all $\abs{\Ycal} \cdot \alpha_{Y}^{-4(t + m)}$
vectors then gives a failure probability of at most
\[
\abs{\Ycal} \cdot \alpha_{Y}^{-4 \left( t + m \right) + \frac{d}{5}}.
\]
When $m, t \leq \frac{d}{80}$, this is at most
\[
\abs{\Ycal} \cdot \alpha_{Y}^{- \frac{d}{10} + \frac{d}{5}}
=
\abs{\Ycal} \cdot \alpha_{Y}^{\frac{d}{10}},
\]
which is the desired bound.
\end{proof}

Note that the need to rule out all sparse vectors in the null space
precludes us from applying this bound separately for each vector
$\yhat$ in the $\epsilon$-net.
Instead, we lower bound the density of all null space vectors
via Lemma~\ref{lem:SparseGaussianNullSpace} once for all $\yhat$
in the $\epsilon$-net, and then invoke Lemma~\ref{lem:ClobberOneVector}
for each of the nets generated for each $\yhat$.

%% file: Solver.tex
\section{Solver for Matrices with Low Displacement Rank}
\label{sec:Solver}

We now prove the running time of the solver for Hankel
matrices.
Our notation of matrices will revolve around block
matrices throughout that section: we use $s$ to denote
the size of a block, and $m$ to denote the number of blocks.
When there are multiple matrices that form natural
sequences, we will index into them using superscripts.

The algorithm here has much similarities with the hierarchical
matrix based solver by Xia, Xi, and Gu~\cite{XiaXG12}.
The main difference is that we work in the matrix domain instead of
the Fourier domain,
and our algorithm is optimized for inputs with arbitrarily lengthed
simulated floats.
Xia, Xi, and Gu~\cite{XiaXG12} does most of what we do: after transforming
the problem to the Fourier domain, they write the matrix as a Cauchy matrix,
which they in turn view as a hierarchical matrix multiplied by complex coefficients
of the form of $\frac{1}{z^{i} - z^{j}}$ ($z$ is a complex root of unity).
By leveraging well-spaced decompositions similar to the fast multipole method,
they are able to use stable solvers for hierarchical matrices to extract
the solution to the overall Cauchy matrix.
To our knowledge, directly invoking this algorithm would lead to an extra factor of $m$.
This is because our condition number, and sizes of the numbers involved,
are all $\exp(\Otil(m))$.
The fast multipole method is only able to extra one digit per iteration
due to its reliance on the Taylor expansion, so would give a total running time of
$s^{\omega} m^3 > n^{\omega}$, which is too big.

We will use the $\{\cdot\}$ notation to index into subsets
of blocks, in the same manner as indexing into row/column
indices.
\begin{definition}
\label{def:BlockIndexing}
Given block size $s$ and a set of indices $S \subseteq [m]$,
we use $\{S\}$ to denote the entries in the corresponding blocks.
That is, if we arrange the indices of the $m$ blocks sequentially,
we have:
\[
\left\{S\right\}
=
\bigcup_{i \in S} \left[ \left(i - 1\right)s + 1, is\right].
\]
\end{definition}

\Solver*

Crucial to our analysis is the small displacement rank
property of the Hankel matrix.
This fact relies on the $s$-block-Hankel matrix is identical to
itself shifted down and to the left by $s$ entries each.
This down/left shift however is a bit more cumbersome to
represent notationally, as the shifts occur in different
directions along the rows and columns.
So instead, we work with $s$-block-Toeplitz matrices,
which is formed by reversing the order of columns of $H$.

\begin{definition}
\label{def:Toeplitz}
A $s$-block-Toeplitz matrix with $m$ blocks
is an $ms$-by-$ms$ matrix $T$ where
\[
T_{\left\{i, j\right\}}
=
M^{\left( i - j \right)}
\]
where $M^{(-m + 1)}\ldots M^{(m - 1)}$
is a sequence of $s$-by-$s$ matrices.
\end{definition}

Notationally we will use $T$ to denote the
matrices that we operate on to emphasize the connection/motivation
with $s$-block-Toeplitz matrices.

\subsection{Displacement Rank Based Representations}

We can then define the shift-down by $s$ operator.
Its transpose is the shift-right by $s$ operator
when right multiplied to the matrices.
We will fix this definition for our choice of block size of $s$.
An illustration of it is in Figure~\ref{fig:DisplacementMatrix}.
\begin{definition}
\label{def:Displacement}
For any choice of block size $s$ and block number $m$,
the square displacement operator $\Delta(s)$
(whose dimension we assume to be implicit to the matrix
we use it against)
is the matrix with 
$1$s on all entries $s$ below the diagonal, and $0$ everywhere else.
\[
\Delta_{ij}
=
\begin{cases}
1\left(s\right) & \qquad \text{if $i = j + s$}\\
0 & \qquad \text{otherwise}.
\end{cases}
\]
Then for an $n$-by-$n$ matrix $M$,
the $s^{+}$/$s^{-}$-displaced versions of $M$ are given by:
\begin{align*}
    s^{+}\left( M \right)
    & = M - \Delta\left( s \right) M \Delta\left( s \right)^{T} , \\
	s^{-}\left( M \right)
	& = M - \Delta\left( s \right)^{T} M \Delta\left( s 	\right).
\end{align*}
and its $+s$/$-s$-displaced ranks are:
\begin{align*}
    \textsc{Rank}_{+s}\left( M \right)
    &= \textsc{Rank}\left( s^{+}\left(M \right) \right)
    = \textsc{Rank} \left( M - \Delta\left( s \right) M
        \Delta\left( s \right)^{T} \right), \\
	\textsc{Rank}_{-s}\left( M \right)
	&= \textsc{Rank}\left( s^{-}\left(M \right) \right)
	= \textsc{Rank} \left( M - \Delta\left( s \right)^{T} M
        \Delta\left( s \right) \right).
\end{align*}
\end{definition}

\begin{figure}
\begin{center}
\begin{tikzpicture}
\draw (0,0) -- (5,0) -- (5,5) -- (0,5) -- (0,0);
\draw (0,3.1) -- (0,3.9) -- (3.9,0) -- (3.1,0) -- (0,3.1);
\node at (0.3, 3.2) (a) {$1$};
\node at (0.8, 2.7) (a) {$1$};
\node at (1.3, 2.2) (a) {$1$};
\node at (2.7, 0.8) (a) {$1$};
\node at (3.2, 0.3) (a) {$1$};
\node[rotate = -45] at (2, 1.5) (blah) {$\ldots~\ldots$};

\draw [decorate,decoration={brace,amplitude=10pt,mirror}]
(3.2, 0) -- (5.0, 0) node (curly_bracket)[black,midway, yshift =- 0.3 cm] 
{};
\node at (4.1, -0.6) {$s$};
\draw [decorate,decoration={brace,amplitude=10pt}]
(0, 3.2) -- (0, 5.0) node (curly_bracket)[black,midway, yshift =- 0.3 cm] 
{};
\node at (-0.6, 4.1) {$s$};
\draw [decorate,decoration={brace,amplitude=10pt}]
(0, 5) -- (5.0, 5) node (curly_bracket)[black,midway, yshift =- 0.3 cm] 
{};
\node at (2.5, 5.6) {$n$};
\draw [decorate,decoration={brace,amplitude=10pt, mirror}]
(5, 0) -- (5.0, 5) node (curly_bracket)[black,midway, yshift =- 0.3 cm] 
{};
\node at (5.6, 2.5) {$n$};
\end{tikzpicture}
\end{center}
\caption{Displacement marix $\Delta(s) \in \Re^{n \times n}$}
\label{fig:DisplacementMatrix}
\end{figure}

Observe that if $T$ is a $s$-block-Toeplitz matrix,
then every leading principle minor of $T$ has $s^{+}$-displacement
rank at most $2s$.
The key property of displacement matrices is that
the inverse of a full rank matrix has the same displacement
rank under a sign flip.
The following is an adapation of
Theorem~1 from~\cite{KailathKM79} to the
more general displacement setting.
\begin{restatable}{lemma}{DisplacementRankInverse}
\label{lem:DisplacementRankInverse}
For any invertible matrix $M$ and any shift value $s$,
we have
\[
\textsc{Rank}_{+s} \left( M \right)
=
\textsc{Rank}_{-s} \left( M^{-1} \right)
\]
\end{restatable}
Note that this Lemma applied with $M^{-1}$ instead of $M$
also gives $\textsc{Rank}_{-s} ( M )
    = \textsc{Rank}_{+s} ( M^{-1} )$.
The $s = 1$ case of this is also the reason behind the
representation of inverses of Toeplitz matrices
known as the Gohberg-Krupnik Formula~\cite{LabahnS92,GohbergK72}.

Lemma~\ref{lem:DisplacementRankInverse} allows us to have
rank-$s$ representations of inverses of leading minors
of $T$, as well as the Schur complements formed when
inverting onto a subset of the entries.

Also observe that matrix-vector multiplications involving
$\Delta(s)$ and $\Delta(s)^T$ take linear time:
it's merely shifting all entries.
So given multiplication access to $M$,
we can also obtain multiplication access to both $s^{+}(M)$
and $s^{-}(M)$.
Such a translation in representations is also error-preserving.
We check that errors in the orignal matrix,
or the displaced versions, translate naturally to each toher.
\begin{lemma}
\label{lem:ErrorTransfer}
For any $n$-by-$n$ matrices $M$ and $\Mtil$
\[
n^{-2} \norm{M - \Mtil}_{F}
\leq
\norm{s^{+}\left( M \right) - s^{+}\left( \Mtil \right)}_{F}
\leq
n^2 \norm{M - \Mtil}_{F}
\]
and similarly for the differences of the negatively displaced versions.
\end{lemma}

\begin{proof}
In the forward direction, we have for all $i, j \geq s$,
\begin{multline*}
\abs{
s^{+}\left( M \right)_{ij} - s^{+}\left( \Mtil \right)_{ij}
}
=
\abs{M_{ij} - M_{i - s, j - s} - \left( \Mtil_{ij} - \Mtil_{i - s, j - s} \right)}\\
\leq
\abs{M_{ij} - \Mtil_{ij}}
+
\abs{M_{i - s, j - s} - \Mtil_{i - s, j - s}}
\end{multline*}
and the rest of the entries are the same.
So each entry in the difference $M - \Mtil$ contributes
to at most two entries.

In the reverse direction, we get
\[
M_{ij}
=
\sum_{0 \leq k \leq \left\lfloor i / s \right\rfloor}
    s^{+}\left( M \right)_{i - ks, j - ks}
\]
which subtracted against the same formula for $\Atil$ gives
\begin{multline*}
\abs{M_{ij} - \Mtil_{ij}}
=
\abs{\sum_{k} s^{+}\left( M \right)_{i - ks, j - ks} -
  s^{+}\left( \Mtil \right)_{i - ks, j - ks}}\\
\leq
\sum_{k} \abs{s^{+}\left( M \right)_{i - ks, j - ks} -
  s^{+}\left( \Mtil \right)_{i - ks, j - ks}}.
\end{multline*}
So the contributions of errors on each entry get
amplified by a factor of at most $n$.
\end{proof}

We will treat this representation as a black-box, and formalize
interactions with it using the following lemma.

\begin{restatable}{lemma}{DisplacementRepresentation}
\label{lem:DisplacementRepresentation}
Given block size $s$, block count $m$,
$(ms)$-by-$r$ matrices $X$ and $Y$,
the $(ms)$-by-$(ms)$ matrix $M$ such that
\[
M - \Delta\left(s\right) M \Delta\left( s \right)^{T}
=
X^TY
\]
has a unique solution.

Furthermore, there is a routine
$\textsc{ImplicitMatVec}$
that for any accuracy $\delta < (ms)^{-10}$
corresponds to a linear operator $\Ztil_{XY \rightarrow M, \delta}$
such that for any $ms$-by-$k$ matrix $B$ with at most $L_B$
words after decimal place,
$\textsc{ImplicitMatVec}(X, Y, B, \delta)$ takes time
(measured in number of word operations)
\begin{multline*}
O\left(m \log^3{m} \cdot
\max\left\{r, s\right\}
\max\left\{s^{\omega - 2}k, sk^{\omega - 2} \right\}
\right.
\\
\left. \cdot
\left( L_B + \log\left(
\left( 1 + \normi{X}_{\infty} \right)
\left( 1 + \normi{Y}_{\infty} \right)
\left( 1 + \normi{B}_{\infty} \right)
ms / \delta \right) \right)
\right)
\end{multline*}
and outputs the $(ms)$-by-$k$ matrix
\[
\Ztil_{XY \rightarrow M, \delta}B
\]
where $\Ztil_{XY \rightarrow M, \delta}$ is a matrix with at most
$O(\log{m} \log((1 + \normi{X}_{\infty})(1 + \normi{Y}_{\infty}) ms / \delta))$
words after the decimal point such that
\[
\norm{\Ztil_{XY \rightarrow M, \delta} - M}_{F}
\leq
\delta
\]
The same guarantees and runtime bounds also hold for the
negative displacement case where $M$ is implicitly specified
as $M - \Delta(s)^{T} M \Delta(s) = XY^{T}$,
as well as matrix-vector multiplications with $M^{T}$.
\end{restatable}
\subsection{Recursive Schur Complement}

The main algorithm is to invoke this succinct representation
during intermediate steps of (block) Gaussian elimination.
Specifically, for a subset of coordinates of $C$
and its complement set $\Cbar$,
we want to directly produce (a high accuracy approximation of)
the low displacement rank factorization of
\[
\textsc{SC}\left( T, C \right)
=
T_{CC} - T_{C \Cbar} T_{\Cbar \Cbar}^{-1} T_{\Cbar C}.
\]

Before proceeding with the algorithm
we first must show that the Schur complement has
small displacement rank.
For this, we need the following characterization
of Schur complements as minors of inverses.

\begin{fact}
\label{fact:SchurInverse}
If $M$ is a full rank matrix, then for any
subset of coordinates $C$, we have
\[
\textsc{SC}\left( M, C \right)^{-1}
=
\left[ M^{-1} \right]_{CC}.
\]
\end{fact}

Combining this with the fact that positive/negative
displacement ranks work well under taking leading/trailing
principle minors gives the following bounds
on displacement ranks.
We will in general use $C$ to denote the remaining coordinates,
which in our recursive algorithm will be a suffix of the indices.
Then the leading portion of coordinates will be
denoted using $\Cbar$.

\begin{lemma}
\label{lem:SchurComplementClosure}
If $M$ is a symmetric full rank matrix,
$C$ and $\overline{C}$ are a coordinate wise suffix/prefix
split of the indices, then we have:
\[
\textsc{Rank}_{+ s} \left( M_{\Cbar, \Cbar} \right),
\textsc{Rank}_{+ s} \left( \textsc{SC}\left( M, C \right) \right)
\leq
\textsc{Rank}_{+ s} \left( M \right)
\]
\end{lemma}

\begin{proof}
To bound the displacement rank of the leading minor $M_{\Cbar \Cbar}$,
note that because $\Cbar$ is a prefix of the coordinates,
\[
\Delta\left( s \right) M_{\Cbar \Cbar} \Delta\left( s \right)^{T}
=
\left( M - \Delta\left( s \right) M \Delta\left( s \right)^{T} 
\right)_{\Cbar \Cbar}
\]
which when substituted back in gives
\[
M_{\Cbar \Cbar} - \Delta\left( s \right) M_{\Cbar \Cbar} \Delta\left( s \right)^{T}
= 
\left( M - \Delta\left( s \right) M \Delta\left( s \right)^{T} 
\right)_{\Cbar \Cbar}.
\]
The inequality then follows from minors of matrices
having smaller ranks.
		
For the Schur Complement onto the suffix of indices $C$,
we combine the inverse representation
of Schur Complements from Fact~\ref{fact:SchurInverse} above
with the relation between positive and negative displacement ranks
from Lemma~\ref{lem:DisplacementRankInverse}.

Because $C$ is a suffix of the indices, we get
\[
M^{-1}_{CC} - \Delta\left( s \right)^{T} M^{-1}_{CC} \Delta\left( s \right)
=
\left( M^{-1} - \Delta\left( s \right) M^{-1} \Delta\left( s \right)^{T}\right)_{CC},
\]
which implies
\[
\textsc{Rank}_{+s}\left( \textsc{SC}\left( M, C \right) \right)
=
\textsc{Rank}_{-s}\left( M^{-1}_{CC} \right)
\leq
\textsc{Rank}_{-s}\left( M^{-1} \right)
=
\textsc{Rank}_{+s}\left( M \right)
\]
where the first and last equalities follow from the
equivalences between positive and negative
displacement ranks given by Lemma~\ref{lem:DisplacementRankInverse}.
\end{proof}

This observation, plus the low displacement representation
given in Lemma~\ref{lem:DisplacementRepresentation} leads
to a recursive routine.
We repeatedly partitioning up the matrix into two even
halves of leading/trailing coordinates.
As our partitions (as well as initial eigenvalue conditions)
are on the blocks, we will overload notation and use
$\{\Cbar\}$ and $\{C\}$ to denote the splits into the
corresponding (size $s$) blocks.

Then we recursively find the inverse of the top-left half $\{\Cbar\}$
(in the succinct representation given by
Lemma~\ref{lem:DisplacementRepresentation}),
and use it to compute the Schur complemt of the bottom-right half on $\{C\}$.
Then we once can recurse again on the Schur complement on $\{C\}$,
which also has low displacement rank representation.
Combining its inverse with the inverse of $T_{\{\Cbar, \Cbar\}}$
then gives the overall inverse.

Top-level pseudocode of this method is in Figure~\ref{fig:RecursiveSC}
Its main algorithmic difficulties lie in directly computing
the displaced factorizations of the Schur complement,
and the overall inverse.
We present this algorithm in
Section~\ref{subsec:LowRank} before returning to the overall proof.

\begin{figure}

\begin{algbox}
$\textsc{Solve}_{T}\left(\cdot\right)
=
\textsc{RecursiveSC}(m, s, X, Y, \epsilon)$\\
\underline{Input}: $(ms) \times r$ matrices $X$ and $Y$ that
describe an $(ms)$-by-$(ms)$ $s$-block-matrix $T$ with $m \times m$ blocks
that has $s^{+}$-displacement rank $r$.
Error threshold $\epsilon$.\\
\underline{Output}:
A routine $\textsc{Solve}_{T}\left(\cdot\right)$
that corresponds to a linear operator $Z(T)$.
		\begin{enumerate}
			\item If $m = 1$, return the explicit inverse of $T$.
			\item Let $\{C\}$ the last $\lfloor m / 2 \rfloor$ blocks,
			and its complement $\{\Cbar\}$ be the first $\lceil m / 2\rceil$ blocks.
			
			\item Use Lemma~\ref{lem:DisplacementRepresentation}
			to generate, using $X$ and $Y$,
			$\epsilon$-error multiplications involving
			$T_{C, C}$, $T_{\Cbar, C}$, and $T_{C, \Cbar}$.
			
			\item \label{line:FirstRecursive}
			Recursively on $\Cbar$ to obtain operator $Z(\Cbar)$
			that corresponds to
			\[
			\textsc{Solve}_{T_{\left\{\Cbar, \Cbar\right\}}}\left( \cdot \right)
			\leftarrow \textsc{RecursiveSC}\left(\lceil m / 2\rceil,
			    s, X_{\{\Cbar\}, :}, Y_{\{\Cbar\}, :}, \epsilon\right).
			\]
			
			\item \label{line:GenSC}
			Implicitly generate
			matrix multiplication functions for the Schur complement
			\[
			\widetilde{SC} = T_{\{C, C\}}
			    - T_{\{C, \Cbar\}} Z(\Cbar) T_{\{\Cbar, C\}},
			\]
			\item Factorize the $s^{+}$-displaced
			approximate Schur Complement onto $\{C\}$:
			\[
			X(\textsc{SC}), Y(\textsc{SC}) \leftarrow
			\textsc{LowRankApprox}(\abs{C}, s, \textsc{Mult}_{s^{+}(\widetilde{SC})}(\cdot), 
			\textsc{Mult}_{s^{+}(\widetilde{SC})^{T}}(\cdot), \epsilon).
			\]
			\item \label{line:SecondRecursive}
			Recurse on $\textsc{SC}(T, C)$
			to obtain operator $Z(SC)$ that corresponds to
			\[
			\textsc{Solve}_{SC}\left( \cdot \right)
			\leftarrow
			\textsc{RecursiveSC}(\left\lfloor m / 2\right\rfloor,
			    2s, X(\textsc{SC}), Y(\textsc{SC}), 
			\epsilon).
			\]
			\item Use $Z(\Cbar)$ and $Z(SC)$ to
			implicitly generate multiplication operators for the $s^{-}$-displaced version of
			the approximate inverse operator
            \[
                Z = 
				\left[ \begin{array}{cc}
					I & - Z\left(\Cbar\right) T_{\left\{\Cbar, C\right\}}\\
					0 & I 
				\end{array} \right]\\
				\left[ \begin{array}{cc}
				Z\left( \Cbar \right)
				& 0 \\
				0 &
				Z\left( \textsc{SC} \right)
				\end{array} \right]	\\
				\left[ \begin{array}{cc}
					I & 0\\
					- T_{\left\{C, \Cbar\right\}} Z\left( F \right) & I 
				\end{array} \right],
			\]
			\item Compute
			\[
			X^{INV}, Y^{INV}
			\leftarrow \textsc{LowRankApprox}(m, 2s, \textsc{Mult}_{s^{-}\left( Z\right)}(\cdot), 
			\textsc{Mult}_{s^{-}\left( Z\right)^{T}}(\cdot), \epsilon)
			\]
			and return the multiplication operator given
			by Lemma~\ref{lem:DisplacementRepresentation}
			with error $\epsilon$.
		\end{enumerate}
	\end{algbox}

    \caption{Pseudocode for Recursive Schur Complement Algorithm}
    \label{fig:RecursiveSC}
\end{figure}	

\subsection{Implicit Low Rank Factorizations}
\label{subsec:LowRank}

Key to the efficiency of this algorithm is the ability to
encode the $(sm)^2$ numbers of an inverse with $ms^2$ numbers instead.
We show that given multiplication access to a matrix $T$.
we can directly compute low rank factorizations of $s^{+}(T)$
and $s^{-}(T)$ by only calling the multiplication routine against
$(ms)$-by-$O(s)$ sized matrices.
For this section, we set $n=ms$, and $r=s$,
so we work with $n$-by-$n$ matrices that are close
to rank $r$ within some very small error.

The square case of computing rank revealing factoriaztions was studied
in conjunction with the stability of fast matrix operations~\cite{DemmelDH07}.
However, we need to obtain running times sublinear in the matrix sizes.

We use random Gaussian projections, a method first analyzed by
Sarlos~\cite{Sarlos06}.
The quality of this random projection has been 
studied extensively in randomized numerical liner algebra.
More recent progress show that a sparse random projection
to about $\Otil(r)$ columns captures most of the column space
information~\cite{Sarlos06,DrineasMM08,KannanV17,Woodruff14:book}.
However, because our matrix is given as black-box access,
the density of the vectors do not affect the performance
of our algorithm.

Specifically, we multiply the matrix to be factorized
with a random Gaussian matrix with about $r$ columns,
and use that $n$-by-$\Otil(r)$ matrix to compute a good
column basis for the entire matrix.
The projection guarantees that we will use is the following
variant of Theorem 14 from~\cite{Sarlos06}, as well as
Theorem 45 from~\cite{Woodruff14:book}.

\begin{lemma}
\label{lem:RandProj}
Let $M$ be an $n \times n$ matrix, $r$ any rank parameter,
and $M(r)$ the best rank $r$ approximation to $M$.
Let $S$ be a random $n$-by-$O(r)$ matrix with entries set
i.i.d. to $N(0, 1)$,
and $\Pi_{MS}$ the projection operator onto the column
space of $MS$.
Then we have with probability at least $1 - n^{-10}$,
\begin{enumerate}
\item The projection of $M$ onto the row space of $SM$
has small distance to $M(r)$:
\[
  \norm{\Pi_{MS} M - M\left( r \right)}_{F}
  \leq
  n^{30} \norm{M - M \left( r \right)}_{F}.
\]
\item there is an $\# cols(S)$-by-$n$ matrix $R$
with entry-wise magnitude at most $O(n^4)$ such that
\[
  \norm{MSR - M\left( r \right)}_{F}
  \leq
  n^{30} \norm{M - M \left( r \right)}_{F}
\]
\end{enumerate}
\end{lemma}

Note that the success probability is set to $1 - n^{-10}$
instead of $1/2$ as in~\cite{Sarlos06}:
this is at the expense of a larger error parameter.
This modification is obtained by invoking the Markov inequality
toward the end of the proof in~\cite{Sarlos06} with
a larger error threshold.

To show the existence of $R$ with small entry-wise magnitudes,
we open up the proof of Lemma 45 of~\cite{Woodruff14:book},
but transform it to work with matrices that reduce number of columns.
In order to do so, we make use of the following two
properties of the random Gaussian projections shown in~\cite{Woodruff14:book}.

\begin{lemma}
\label{lem:ProjProperties}
There is an absolute constant such that for any $r$,
the $n$-by-$O(r)$ dense Gaussian matrix $G$ satisfies:
\begin{itemize}
    \item (Subspace Embedding Property) For any $r$-by-$n$ matrix $M$,
    with probability at least $1 - n^{-20}$ we have
    \[
    0.9 \norm{x^T M}_{2}
    \leq
    \norm{x^T M S}_2
    \leq
    1.1 \norm{x^T M}_2
    \qquad \forall x \in \Re^{r}
    \]
    \item (Approximate Matrix Multiplication Property) For any two matrices
    $M(1)$ and $M(2)$ with $n$ rows each,
    with probability at least $1 - n^{-20}$ we have
    \[
    \norm{M\left( 1 \right)^T S S^T M\left( 2 \right) - M\left( 1 \right)^T M\left( 2 \right)}_{F}
    \leq
    n^{21} \norm{M\left( 1 \right)}_{F}
    \norm{M\left(2\right)}_{F}
    \]
\end{itemize}
\end{lemma}

\begin{proof}(Of Lemma~\ref{lem:RandProj})
Let the SVD of $M$ be
\[
M = U \Sigma V^T
\]
and let the top $r$ singular vectors/values be $U(r) \in \Re^{n \times s}$,
$\Sigma(r) \in \Re^{r \times r}$, and $V(r) \in \Re^{n \times s}$ respectively.
The given condition of the distance from $M$
to its rank $r$ approximation being at most $\epsilon$
means we have
\[
\norm{M - U\left( r \right) \Sigma\left( r \right) V\left( r \right)^T}_{F}
\leq
\epsilon
\]

The matrix used to bound the distance is then
\[
R
=
\left(V\left( r \right)^T S \right)^{\dag} V\left(r \right)^T.
\]

The subspace embedding property gives that
because $V(r)$ has $r$ rows,
$V(r)^T S$ has full row rank $r$, and
\[
\left(V\left( r \right)^T S \right)^{\dag}
=
\left( V\left( r \right)^T S \right)^T Q^{-1}
=
S^T V\left( r \right) Q^{-1}
\]
for some $r$-by-$r$ matrix $Q = V(r)^T S S^T V(r)$
whose eigenvalues are in the range $[0.8, 1.3]$.
This in turn means that the entry-wise magnitude
of both $Q$ and $Q^{-1}$ are at most $O(1)$.

Combining this with the entry-wise magnitude of $1$
for $V(r)$ (because it's an orthonormal basis)
and $n$ for $S$ (due to Claim~\ref{claim:GaussianMax})
gives that $R$ has entry-wise magnitude at most
$O(n^4)$, and Frobenius norm at most $O(n^{5})$.

It remains to bound the error term.
Let $\Mhat$ be the matrix corresponding to singular
values $r + 1 \ldots n$ in the SVD:
\[
\Mhat
=
M - U\left( r \right) \Sigma\left(  r\right)
V \left( r \right)^T
=
U\left( n - r \right) \Sigma\left( n - r\right)
V \left( n - r \right)^T.
\]
The fact that $Q$ perfectly inverts $V(r)^T S$ means
$U(r) \Sigma(s) V(r)^T S R = U(r) \Sigma(s) V(r)$.
So the error is only on the $\Mhat$ term:
\[
M - M S R
=
\Mhat S R
=
\Mhat S S^T V\left( r \right) Q^{-1} V \left( r \right)^{T}.
\]
Furthermore, the fact that $V(r)$ is an orthonormal basis,
and that $Q^{-1}$ has max eigenvalue $2$, means that we
can remove them from consideration:
\[
\norm{\Mhat SR}_{F}
\leq
\norm{\Mhat S S^T V\left( r \right)}_{F}
\norm{Q^{-1}}_2
\norm{V \left( r \right)}_{2}
\leq
2
\norm{\Mhat S S^T V\left( r \right)}_{F}.
\]

For this last term, the orthgonality of singular vectors gives
\[
0
=
V \left( n - r \right)^T
V \left( r \right)
=
U\left( n - r \right) \Sigma\left( n - r\right)
V \left( n - r \right)^T
V \left( r \right)
=
\Mhat V \left( r \right).
\]
So applying the Approximate Matrix Multiplication Property gives
\[
\norm{\Mhat SS^T V\left( r \right)}_{F}
\leq
n^{21}
\cdot \norm{\Mhat}_{F}
\cdot \norm{V\left(r\right)^T}_{F}
\leq
n^{22} \epsilon.
\]
Here the last inequality follows from $V(r)$ being an orthonormal basis.
Incorporating the additional factor of $2$ above then gives the result.
\end{proof}

We remark that this $R$ matrix is used throughout randomized
algorithms for computing low rank approximations~\cite{KannanV17}.
The bound on the max magnitude of entries of $R$ here allows us to
perturb $MS$ slightly so that its minimum singular value
bounded away from $0$, which in turn gives bounds on the
bit-complexity of computing projections into its column space.
Specifically, we invoke the dot-product against null space
idea inherent to analyses of min singular values of entry-wise
i.i.d. matrices~\cite{SankarST03:journal,TaoV10},
which is also the starting point of the proof of 
Theorem~~\ref{thm:RandKrylov} in Section~\ref{sec:RandKrylov}.

\begin{lemma}
\label{lem:DensePerturb}
Let $\Mhat$ be any $n$-by-$d$ matrix with $d < n$,
and $\Mtil$ a perturbation of $\Mhat$ formed
by adding $\epsilon N(0, 1)$ to every entry.
Then with probability at least $1 - n^{-11}$,
the minimum singular value of $\Mtil$ is at least
$\epsilon n^{-20}$.
\end{lemma}

\begin{proof}
For each column $j$, consider a vector $w(\setminus j)$
in the null space of the rest of the columns of $\Mtil$.
The dense Gaussian added to $\Mtil$ gives that with
probability at least $1 - n^{-14}$, we have
\[
\abs{w\left( \setminus j\right)^T \Mtil_{: j}} \geq \epsilon n^{-15}.
\]
Taking union bound over all columnsn $j$ gives that
this holds for all columnns with probability at least
$1 - n^{-13}$.

We now lower bound $\norm{\Mtil x}_2$ for all unit vectors $x \in \Re^{d}$.
For any such unit $x$,
there is some $j$ such that $\abs{x_j} \geq n^{-1/2} > n^{-1}$.
Then we get:
\[
\norm{\Mtil x}_2
\geq
\abs{x_j} \cdot \abs{w\left(\setminus j \right)^T \Mtil_{:, j}}
\geq
\epsilon n^{-20}
\]
where the last inequality follows from the dot-product
lower bound above.
\end{proof}

Pseudocode of algorithm that utilizes this projection is in Figure~\ref{fig:LowRankApprox}.
Its guarantees are given below in Lemma~\ref{lem:LowRankApprox}.
	
\begin{figure}[ht]
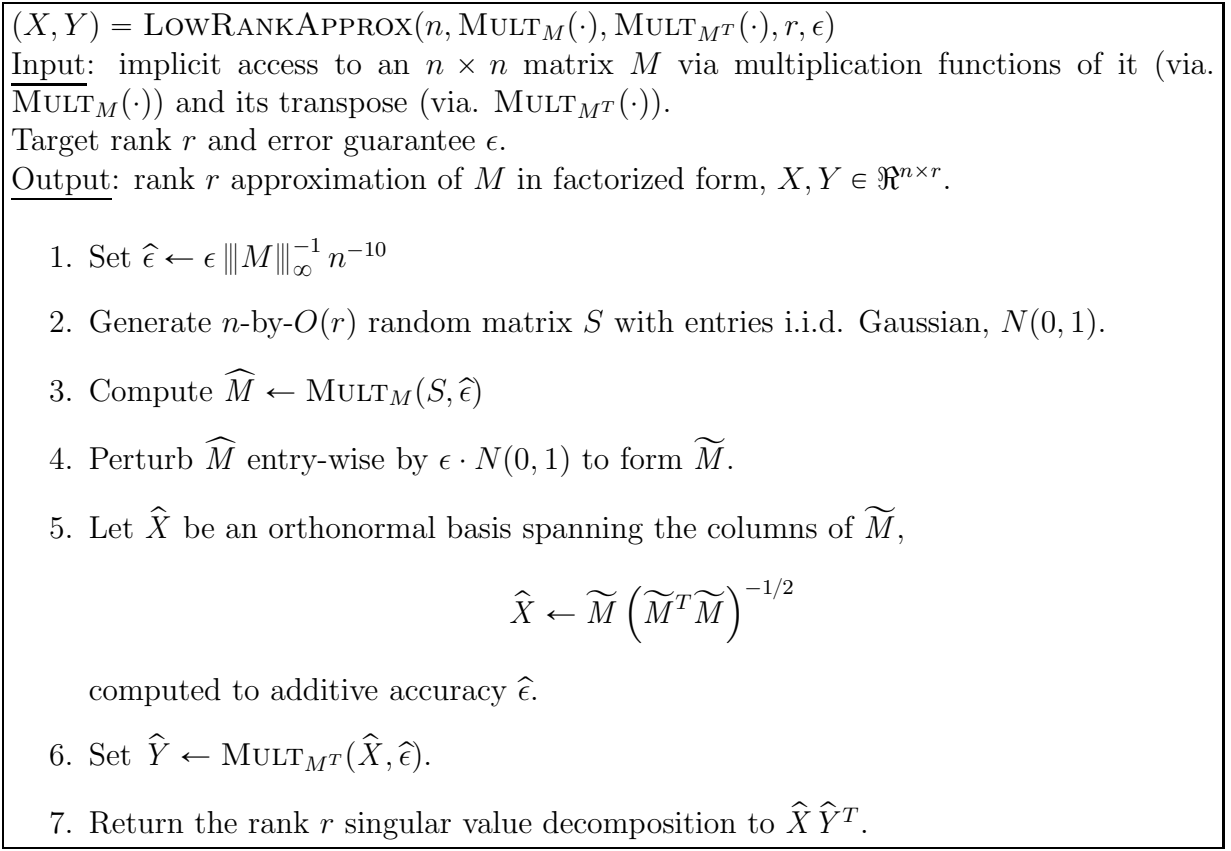

	\begin{algbox}
		$(X, Y) = \textsc{LowRankApprox}(n, 
		\textsc{Mult}_{M}(\cdot),
		\textsc{Mult}_{M^{T}}(\cdot), r, \epsilon)$\\
		\underline{Input}: implicit access to an $n \times n$ matrix $M$
		via multiplication functions of it (via. $\textsc{Mult}_{M}(\cdot)$)
		and its transpose (via. $\textsc{Mult}_{M^{T}}(\cdot)$).\\
		Target rank $r$ and error guarantee $\epsilon$.\\
		\underline{Output}: rank $r$ approximation of $M$
		in factorized form, $X, Y \in \Re^{n \times r}$.
		
		\begin{enumerate}
		    \item Set $\epsilonhat \leftarrow \epsilon \normi{M}_{\infty}^{-1} n^{-10}$
			\item Generate $n$-by-$O(r)$ random matrix $S$
			with entries i.i.d. Gaussian, $N(0, 1)$.
			\item Compute $\Mhat \leftarrow \textsc{Mult}_{M}(S, \epsilonhat )$
			\item Perturb $\Mhat$ entry-wise by $\epsilon \cdot N(0, 1)$ to form $\Mtil$.
			\item Let $\Xhat$ be an orthonormal basis spanning the columns of $\Mtil$,
			\[
			\Xhat \leftarrow \Mtil \left( \Mtil^T \Mtil \right)^{-1/2}
			\]
			computed to additive accuracy $\epsilonhat$.
			\item Set $\Yhat \leftarrow
			    \textsc{Mult}_{M^{T}}( \Xhat, \epsilonhat)$.
			\item Return the rank $r$ singular value decomposition
			to $\Xhat \Yhat^{T}$.
		\end{enumerate}
	\end{algbox}
	\caption{Pseudocode for computing a rank $r$ factorization of
	a matrix using only accesses to its products against vectors.}
    \label{fig:LowRankApprox}
\end{figure}
	
\begin{lemma}
\label{lem:LowRankApprox}
For an $n \times n$ matrix $M$ with $L_M$ words after the decimal point,
and Frobenius norm distance at most $\epsilon$ to some rank $r$ matrix,
given via implicit access to multiplication operators
$\textsc{Mult}_{M}$ and $\textsc{Mult}_{M^{T}}$,
$ \textsc{LowRankApprox}(n, 
\textsc{Mult}_{M}(\cdot),
\textsc{Mult}_{M^{T}}(\cdot), r)$ returns
with probability at least $1 - n^{-10}$, 
$X, Y \in \Re^{n \times r}$ such that
\begin{enumerate}
\item $X$ and $Y$ have magnitude at most
$n^2 \normi{M}_{\infty}$, and at most
$\log( (1 + \normi{M}_{\infty}) n / \epsilon)$
words after the decimal point,
\item The approximation error satisfies
\[
\norm{M - X Y^{T}}_{F}
\leq
n^{20} \epsilon
\]
\end{enumerate}
and the total cost of the algorithm is:
\begin{enumerate}
\item the cost of calling $\textsc{Mult}_{M}(\cdot)$
and $\textsc{Mult}_{M^{T}}(\cdot)$ for
$O( r\log{n})$ vectors with entry magnitude at most $1$ and at most
$O( \log( (1 + \normi{M}_{\infty}) n/ \epsilon))$
words after the decimal point, plus
\item an additional overhead of
$\Otil(n r^{\omega - 1} 
(\log((1 + \normi{M}_{\infty})n / \epsilon))$.
\end{enumerate}

\end{lemma}

\begin{proof}
Lemma~\ref{lem:RandProj} gives that there is some $R$
with entry-wise magnitude at most $O(n^4)$ such that
\[
\norm{\Mhat R - M}_{F}
\leq
n^{30} \epsilon.
\]
Furthermore, the entry-wise perturbation is at most
$n\epsilon$ with probability $\exp(-n)$, so we get
\[
\norm{\Mhat R - \Mtil R}_{F}
\leq
n^5 \epsilon
\]
which in turn gives
\[
\norm{\Mtil R - M}_{F}
\leq
n^{40} \epsilon,
\]
or that $M$ has a good approximation in the column space
of $\Mtil$ as well.

On the other hand, Lemma~\ref{lem:DensePerturb} gives that
$\Mtil$ has full column rank, and has minimum singular value
at least $\epsilon n^{-20}$.
This means that the bit complexity needed to compute
$\Xhat$ to the specified accuracy of $\epsilonhat$ is at most
$\epsilonhat \cdot \epsilon \cdot n^{-O(1)}$, or that
$\log(\normi{M}_{\infty} n / \epsilon)$ words of precision
suffices for these calculations.
This gives the round-off error needed when generating
$S$, and in turn the bounds on the word lengths after
decimal points in the inputs given to $\textsc{Mult}_{M}$
and $\textsc{Mult}_{M^{T}}$.

We then bound the costs of the other steps.
First, note that because the max magnitude of entries in
$S$ is $1$, the max magnitude of entries in $\Mhat$
is at most $O(n^2 \normi{M}_{\infty})$.
Also, because $\Xhat$ is an orthonormal basis,
the max magnitude of an entry in it is at most $1$,
which means the max magnitude of an entry in $\Yhat$
is at most $O(n^2 \normi{M}_{\infty})$ as well.

Factoring in the round-off errors we incurred gives an additional error of
\[
\epsilonhat
\cdot 
O\left(n^2 \normi{M}_{\infty}\right)
\leq
n^{5} \epsilon.
\]
which in turn implies
\[
\norm{\Xhat \Yhat^{T} - M\left( r \right)}_{F}
\leq
n^{45} \epsilon.
\]
with probability at least $1 - n^{-10}$,
and in turn that the rank $r$ decomposition of
$\Xhat \Yhat^{T}$ has error at most $n^{45} \epsilon$.

So it remains to bound the overhead caused by running
high accuracy singular value decomposition on
$\Xhat \Yhat^{T}$.
Note that the outer product of this matrix is:
\[
\left( \Xhat \Yhat^{T} \right)^{T}
\Xhat \Yhat^{T}
=
\Yhat \left(\Xhat^T \Xhat\right) \Yhat,
\]
and the middle matrix can be computed in
$O(n r^{\omega - 1})$ operations involving
numbers whose magnitude are at most $O(n^{10} \normi{M}_{\infty})$
and with at most
$O(\log(n(1 + \normi{M}_{\infty})/\epsilon)) + L_M)$
words after decimal place.

Taking square root of this positive semi-definite
matrix, and rounding to error
\[
\frac{\epsilon}{\normi{M}_{\infty}  n^{10}}
\]
then allows us to write this matrix as
\[
\left( P \Yhat \right)^{T} P\Yhat
\]
for some positive semi-definite matrix $P = (\Xhat^T \Xhat)^{1/2}$.
Note that the magnitude of $P\Yhat$ is less than the
magnitude of $Y$ due to square-rooting only decreasing
eigenvalues above $1$.

The equivalence of SVDs of outer and inner products
means we can take the SVD of $\Yhat^T P^{T} P \Yhat$,
which is an $r$-by-$r$ matrix.
The complexity of this step was shown to be
$\Otil(r^{\omega})$
times the lengths of the initial words
in Sections 6.1. and 6.2. of~\cite{DemmelDH07}.
\footnote{The actual stated running time in~\cite{DemmelDH07}
is $O(r^{\omega + \eta})$ for any $\eta > 0$.
We observe this bound is $\Otil(r^{\omega})$ operations involving
words whose lengths equal to initial matrix entries:
the only overhead on the matrix multiplication steps
in~\cite{DemmelDH07} are from the $O(\log{r})$ layers of recursion.
}

It can then be converted back to rank $k$ bases for $P \Yhat$,
and then $\Xhat$ by multiplying through via these matrices:
each of which incurrs an error of at most $O(n^2 \normi{M}_{\infty}^2)$.
Thus, it suffices to keep the round-off errors in all of these
steps at most
\[
\frac{\epsilon}{n^{50} \left( 1 + \normi{M}_{\infty} \right)^4}
\]
which gives a word-length of at most
$\log(n (1 + \normi{M}_{\infty}) / \epsilon)$.
Incorporating this, along with the upper bound of all numbers
of $O(n^2 \normi{M})$ together with the
$O(n r^{\omega - 1})$ operations then gives the total cost.
\end{proof}

We remark that we cannot directly return just $\Xhat$ and $\Yhat$.
Due to the recursive invocation of such routines,
such an overhead will accumulate to $O(1)^{d}$
over $d$ layers of the recursion.
Such error would in turn necessitate stopping the recursion earlier, and in turn at least an $n^{o(1)}$ overhead in running time.

Also, note that the eigen-gap we need for the $s$-by-$s$
matrix is around $\epsilon$:
we can get away with using things like matrix exponential
operators from~\cite{OrecchiaSV12,SachdevaN13}
instead of the full eigen-decomposition routine from~\cite{DemmelDH07}.

On the other hand, we do remark that this ability 
to compute a low rank factorization in sublinear time is critical to our overall running time gains.
In our setting with $\log(\norm{M}_{\infty} / \epsilon)
= \Otil(m)$ and $L_{M} = \Otil(m)$,
the total cost of computing this factorization is about
$s^{\omega} \cdot m^2$ word operations,
so the reduction from $n$ to $s$ is necesasry for gains over $n^{\omega} \approx (sm)^{\omega}$.

\subsection{Error Propagation}
\label{subsec:Error}

Our analysis of the error of the recursive algorithm given
in Figure~\ref{fig:RecursiveSC} relies on bounds on extreme singular
values of all principle minors of $T$ containing block-prefixes of indices.
Such blocking is natural due to $T_{\{\Cbar, \Cbar\}}$ taking prefixes of blocks.

Below we provide tools needed to handle the accumulation of errors.
Lemma~\ref{lem:ErrorInvert} implies that singular value bound allow us to
carry errors between a matrix and its inverse, losing a factor of
$poly(n, \alpha^{-1})$ at each step where $\alpha$ is the min singular value.
In order to use this though, we need to bound the singular values
of all intermediate matrices that arise from such recursion.

Note that all the matrices that we work with consist of a
prefix of blocks, Schur complemented onto a suffix of blocks.
Formally, we let $\{S\}$ denote a prefix of the block indices,
and $\{C\}$ denote the suffix that we Schur complement onto,
and want to bound the min/max singular values of the matrix
\[
\textsc{SC}(T_{\left\{S, S\right\}}, \left\{C\right\}).
\]
The given conditions on $H$ transferred to $T$
ensures that $T_{\{S, S\}}$ has good singular values.

We show via the below lemma that this condition ensures
that the Schur complement also has good singular values.
Once again, although we only work with block-alligned
indices, we prove the bounds for any subset of indices.
\begin{lemma}
	\label{lem:ConditionSchur}
	If $M$ is a $n \times n$
	square full rank matrix, and $C$/$\Cbar$ is a split of the indices
	such that both $M$ and $M_{\Cbar \Cbar}$ have singular values in the range
	$[\sigma_{\min}, \sigma_{\max}]$, then the singular values of
	$\textsc{SC}(M, C)$ are in the range
	$[ n^{-2} \sigma_{\min}, n^{10} \sigma_{\max}^{2} \sigma_{\min}^{-1} ]$.
\end{lemma}
	
\begin{proof}
	Since $M$ is square and full rank, the Schur Complement
	$\textsc{SC}(M, C)$ is also full rank.
	Its min singular value is at least $n^{-2}$ times the inverse of
	the magnitude of the max entry in its inverse.
	This entry is in turn at most the max magnitude of an entry in
	$M^{-1}$, which is at most $\sigma_{\min}^{-1}$.
		
	The max magnitude in $\textsc{SC}(M, C)$
	follows from bounding the max magnitude of
	\[
	M_{C \Cbar} M_{\Cbar \Cbar}^{-1} M_{\Cbar C}.
	\]
	The magnitude of each entry in these three matrices
	are $\sigma_{\max}$, $\sigma_{\min}^{-1}$, and $\sigma_{\max}$ respectively,
	and there are a total of $n^4$ tuples to consider
	(because every Schur complement that arise in a recursive
	Schur complement routine is the Schur complement onto
	a suffix block sets of a prefix of blocks).
\end{proof}

Lemmas~\ref{lem:ConditionSchur} along with the given initial conditions
on $T$ then gives that all block Schur complements,
as well as their block minors, are well conditioned.

\begin{corollary}
\label{cor:BlocksAreWellConditioned}
Let $H$ be a block-Hankel matrix with $m$ blocks of size $s$,
entry magnitude at most $1 / (sm)$
and min block singular values at least $\alpha_{H}$, that is
\[
\sigma_{\min}\left( H_{\left\{1:i, (m-i+1):m\right\}}\right),
\sigma_{\min}\left( H_{\left\{(m-i+1):m, 1:i\right\}}\right)
\geq
\alpha_{H}
\qquad
\forall 1\leq i \leq m.
\]
Let $T$ be the Toeplitz matrix $T$ formed by reversing
the order of the column blocks of $H$.
For every leading prefix of block coordinates
$\{S\}$, and every trailing blocks of $\{S\}$, $\{C\}$,
the Schur complement $\textsc{SC}(T_{\{S, S\}}, \{C\})$ has
singular values in the range
\[
\left[n^{-10} \alpha_{H}, n^{10} \alpha_{H}^{-3} \right].
\]
\end{corollary}

Note that this covers all leading block minors of Schur complements as well.
Specifically, for $\{C\}$ that's a suffix of entries of $\{S\}$,
and $\{\Shat\}$ that's a prefix of $\{C\}$, we have
\[
\textsc{SC}\left( T_{\left\{S \setminus C \cup \Shat,
    S \setminus C \cup \Shat\right\}},
    \left\{\Shat\right\} \right) 
= 
\textsc{SC}\left( T_{\left\{S, S\right\}}, \left\{C\right\}\right)
    _{\left\{\Shat, \Shat\right\}},
\]

Combining Lemma~\ref{lem:ErrorInvert} and~\ref{lem:ConditionSchur}
also gives that a small error to $M$ implies small errors to
all its Schur complements as well.

\begin{lemma}
\label{lem:ErrorSchur}
	Let $M$ be a full rank square matrix,
	and $C$ a subset of coordinates (with complement $\Cbar$)
	such that both $M$ and $M_{\Cbar\Cbar}$
	have singular values in the range
	$[\sigma_{\min}, \sigma_{\max}]$.
	For any approximation $\Mtil$ such that
	\[
	\norm{M - \Mtil}_{F} \leq \epsilon
	\]
	for some $\epsilon < 0.1 n^{-10} \sigma_{\min} ( \sigma_{\max} / \sigma_{\min})^{-2}$, we have:
	\[
	\norm{\textsc{SC}\left( M, C \right) - \textsc{SC}\left( \Mtil, C \right)}_{F}
		\leq 
		n^{30} \left(\sigma_{\max} / \sigma_{\min}\right)^{4} \epsilon
    \]
\end{lemma}
	
\begin{proof}
	By Fact~\ref{fact:SchurInverse}, which gives that inverses
	of Schur complements are subsets of inverses, we can
	use Lemma~\ref{lem:ErrorInvert} to get
	\[
	\norm{\textsc{SC}\left( M, C \right)^{-1}
			-\textsc{SC}\left( \Mtil, C \right)^{-1}}_F
	\leq
	\norm{M^{-1} - \Mtil^{-1}}_F\\
    \leq
    10 \sigma_{\min}^{-2} \epsilon.
	\]
	On the other hand, inverting the singular value bounds on
	Schur complements from Lemma~\ref{lem:ConditionSchur}
	gives that all singular values of $\textsc{SC}(M, C)^{-1}$
	are in the range
	\[
	\left[
	n^{-10} \sigma_{\max}^{-2} \sigma_{\min},
	n^{2} \sigma_{\min}^{-1}
	\right].
	\]
	The given condition on $\epsilon$ then gives
	\[
	10 \sigma_{\min}^{-2} \epsilon
	\leq
	n^{-10} \sigma_{\max}^{-2} \sigma_{\min}.
	\]
	So we can invoke Lemma~\ref{lem:ErrorInvert} once again with
	$\textsc{SC}(M, C)$ as the original matrix, and
	$\textsc{SC}(\Mtil, C)$ as the perturbed matrix gives
	an overall error bound of
	\[
	10
	\cdot
	\left(n^{10}  \sigma_{\max}^{-2} / \sigma_{\min}\right)^{-2}
	\cdot
	\left( 10 \sigma_{\min}^{-2} \epsilon \right)
	\leq
	n^{30} \cdot \left(\sigma_{\max} / \sigma_{\min}\right)^{4} \epsilon.
	\]
\end{proof}

This kind of compounding necessitates a global bounding of errors.
We need to show (inductively) that
all matrices entering into all stages of the recursion
are close to the corresponding matrices of $M$.
Schur complements on the other hand are just as stable to perturbations.

\subsection{Analysis of Overall Recursion}
\label{subsec:Recursion}

We now analyze the overall algorithm by inductively showing that
all the matrices produced are close to their exact versions.
Our overall error bound for the recursion is as follows:

\begin{lemma}
\label{lem:Recursion}
Let $T$ be a matrix s.t. the singular values of any
block minor Schur complement are in the range $[\alpha_{T}, \alpha_{T}^{-1}]$
for some $\alpha_{T} < (ms)^{-100}$.
For 
$\epsilon < (ms)^{-1} \alpha_{T}^{10  \log{m}}$,
the output of
\[
X^{INV}, Y^{INV} = \textsc{Solve}\left(m, s, X, Y, \epsilon\right)
\]
corresponds to a matrix $Z$ with
$s^{-}(Z) = X^{INV} (Y^{INV})^T$ such that
\[
\norm{T^{-1} - Z}_{F}
\leq
\alpha_{T} ^{-10 \log{m}} \epsilon.
\]
\end{lemma}

\begin{proof}
The proof is by induction on $m$.
The base case of $m = 1$ follows from the guarantees
of (fast) matrix inversion~\cite{DemmelDHK07}.

The top-left blocks of this matrix given by $\{\Cbar\}$,
as well as their associated block Schur complements,
satisfy the same singular value bounds due to them being
Schur complements of leading minors of $T$.

Therefore, by the inductive hypothesis, we get that the matrices produced
by the first recursive call on Line~\ref{line:FirstRecursive} gives
additive inverse at most 
which incorporating the $\epsilon$ additive error
from operator generation from Lemma~\ref{lem:DisplacementRepresentation} gives
\[
\norm{T_{\left\{\Cbar, \Cbar\right\}}^{-1}
    - Z\left( \Cbar\right) }_{F}
\leq
\alpha_{T}^{-10 \left( \log{m} - 1 \right) } \epsilon,
\]
		
Then since the max entry in $T_{\{\Cbar, C\}}$ $T_{\{C, \Cbar\}}$, and $T_{\Cbar, \Cbar}^{-1}$
are at most $\alpha_{T}^{-1}$, 
and the approximate operators for multiplying by
them have error at most 
$\alpha_{T}^{-10 ( \log{m} - 1)} \epsilon$,
the approximate Schur complement onto $\{C\}$ using this
approximate inverse $Z(\Cbar)$ and the approximate
multiplication operators in $T$ satisfies
\[
\norm{\widetilde{\textsc{SC}}
    - \textsc{SC}\left( T, \left\{C\right\} \right)}_{F}
\leq
3 \cdot 2
\alpha_{T}^{-1}
\cdot \alpha_{T}^{-10 \left( \log{m} - 1 \right) } \epsilon
\leq
\alpha_{T}^{-2}
\cdot \alpha_{T}^{-10 \left( \log{m} - 1 \right) } \epsilon.
\]

Lemma~\ref{lem:SchurComplementClosure} gives that
$SC(T, \{C\})$ has $s^{+}$-displacement rank at most $2s$.
So combining with the error transfer from
Lemma~\ref{lem:ErrorTransfer} gives that the
distance from $s^{+}(\widetilde{\textsc{SC}})$ to a
rank $2s$ matrix is at most
\[
\left( ms \right)^2 \cdot
\alpha_{T}^{-2}
\cdot \alpha_{T}^{-10 \left( \log{m} - 1 \right) } \epsilon
\]
Then the guarantees of the low rank approximation procedure
\textsc{LowRankApprox} from Lemma~\ref{lem:LowRankApprox}
gives that the resulting factorization has error bounded by
\begin{multline*}
\norm{X\left(SC\right) Y\left(SC\right)^T
-
s^{+}\left(\textsc{SC}\left( T, \left\{C\right\} \right)\right)
}_{F}\\
\leq
\left( ms \right)^{20}
\cdot 
\left( ms \right) ^2
\alpha_{T}^{-2}
\cdot \alpha_{T}^{-10 \left( \log{m} - 1 \right) } \epsilon
+
\left( ms \right) ^2
\alpha_{T}^{-2}
\cdot \alpha_{T}^{-10 \left( \log{m} - 1 \right) } \epsilon\\
\leq
n^{30}
\alpha_{T}^{-2}
\cdot \alpha_{T}^{-10 \left( \log{m} - 1 \right) } \epsilon.
\end{multline*}
Applying the error transfer across displacement
given by Lemma~\ref{lem:ErrorTransfer} gives
that the approximate Schur complement represented by
$X(C)$ and $Y(C)$ produced on Line~\ref{line:GenSC},
which we denote as $\widehat{\textsc{SC}}$ satisfies
\[
\norm{\widehat{\textsc{SC}} - \textsc{SC}\left( T, \left\{C\right\} \right)}_{F}
\leq
n^{40}
\alpha_{T}^{-2}
\cdot \alpha_{T}^{-10 \left( \log{m} - 1 \right) } \epsilon
\leq
\alpha_{T}^{-3} \cdot \alpha_{T}^{-10 \left( \log{m} - 1 \right) } \epsilon.
\]

By assumption, the exact Schur complement $\textsc{SC}( T, C)$
has all consecutive principle minors and Schur complements with
singular values in the range $[\alpha_{T}, \alpha_{T}^{-1}]$.
So the perturbation statement from Lemma~\ref{lem:ErrorInvert}
gives that as long as
$\alpha_{T}^{-3} \cdot \alpha_{T}^{-10 ( \log{m} - 1 ) } \epsilon < \alpha_{T}$
(which is met by the initial assumption on $\epsilon$).
Thus, we get that the singular values of all block minors of
$\widehat{\textsc{SC}}$ are in the range
\[
\left[ \left(1 - \frac{1}{m}\right) \alpha_{T}, \left(1 + \frac{1}{m}\right)\alpha_{T}^{-1}\right].
\]
Since $\frac{1}{2} \leq (1 - \frac{1}{m})^{O(\log{n})} $
and $(1 + \frac{1}{m})^{O(\log{n})} \leq 2$,
this change in $\alpha_{T}$ is less than the difference made by
going from $\log{m}$ to $\log{m} - 1$ in the smaller recursive instance.
So the choice of $\epsilon$ is satisfactory for the
recursive call on $X(C)$ and $Y(C)$ as well.

So applying the inductive hypothesis on the second recursive call
on Line~\ref{line:SecondRecursive} gives
\[
\norm{
Z\left(SC\right) - \widehat{\textsc{SC}}^{-1}
}_{F}
\leq
\alpha_{T}^{-10 \left( \log{m} - 1 \right) } \epsilon
\]
which compounded with the differences with the exact
Schur complement gives
\[
\norm{Z\left(SC\right) - \textsc{SC}\left(T, \left\{C\right\}\right)}_{F}
\leq
\alpha_{T}^{-4} \alpha_{T}^{-10 \left( \log{m} - 1 \right) } \epsilon.
\]

The result for the inductive case
then follows from observing that in all the operators
for these approximate inverses that these two block matrices
are multiplied with have magnitude at most $\alpha_{T}^{-1}$
and there are at most three such operators composed.
Specifically, we get
\[
\norm{T^{-1} - Z}
\leq
3n^2 \cdot \alpha_{T}^{-1} \cdot \alpha_{T}^{-4} \alpha_{T}^{-10 \left( \log{m} - 1 \right) }
\leq 
n^{-100} \alpha_{T}^{-10 \log{m} },
\]
Then applying the low rank approximation guarantees
from Lemma~\ref{lem:LowRankApprox} amplifies this by
a factor of $n^{20}$;
and the implicit operator construction from Lemma~\ref{lem:DisplacementRepresentation} incurs
an additive $\epsilon$.
Incorporating these errors then gives that the inductive
hypothesis holds for $m$ as well.
\end{proof}

This means that all the matrices that arise in intermediate
steps are close to their exact versions.
It allows us to bound the max-magnitude of all the numbers
involved, and in turn the number of digits after decimal
places.
This gives the total cost in number of word operations,
which we also bring back to the original statement for solving a well-conditioned block Hankel matrix.

\begin{proof}(of Theorem~\ref{thm:Solver})

The proof takes two steps: we first use bounds on the
word complexity of $X$ and $Y$ to provide bounds on
the word sizes of all numbers involved/generated,
including the implicit operators.
We then incorporate these lengths into the running time
costs of the routines to obtain bounds on overall running time.

Corollary~\ref{cor:BlocksAreWellConditioned}
gives that the singular values of a block
are in the range $[n^{-10} \alpha_{H}, n^{10} \alpha_H^{-3}]$.
So we can set
\[
\alpha_{T} = \alpha_{H}^{4}
\]
and obtain errors bounds on the recursion via Lemma~\ref{lem:Recursion}.
Specifically, to obtain an overall error of $\epsilon$,
we need to reduce it by a factor of
\[
\alpha_{T}^{10 \log{m} }
\geq
\alpha_{H}^{O\left( \log{m} \right)}.
\]

This means that the additive difference in any
of the intermediate matrices we get is at most
\[
\alpha_{T}^{-10 \log{m}} \epsilon \leq \alpha_{T}^{-1},
\]
which combined with the bounds on the matrices themselves
of $O(n^{10} \alpha_T^{-3}) \leq \alpha_T^{-4}$ by
Lemma~\ref{lem:ConditionSchur} gives that the max magnitude
of a matrix that we pass into Lemma~\ref{lem:LowRankApprox}
to factorize is at most $\alpha_T^{-5}$.
It in turn implies that the magnitude of all factorized
matrices ($X$s and $Y$s) are at most $\alpha_T^{-6}$,
and they have at most $O(\log{m} \log(\alpha_{T}^{-1}\epsilon^{-1}))$
words after the decimal point.
Furthermore, because the initial $T$ has min singular value
at least $\alpha_{H}$, we can truncate the initial $X$ and
$Y$ to $\alpha_{H}^3$ without affecting the errors.

The implicit multiplication operators
generated using fast convolutions via
Lemma~\ref{lem:DisplacementRepresentation} have magnitude
at most $O(\alpha_{T}^6)$, and at most
$O(\log^2{m} \log(\alpha_{T}^{-1} \epsilon^{-1}))$
words after the decimal point.
As the implicitly generated operators for $\widetilde{SC}$,
and the overall inverse $Z$ only composes together a constant
number of such operators, the same bounds also hold for the
sizes of the numbers, up to a constant factor in the exponent.

We now use these word size bounds to bound the overall
running time.
These operators are multiplied against sets of $O(s)$ vectors,
each with magnitude at most $1$ and $O(\log{m} \log( \alpha_{T} / \epsilon))$
words after the decimal point
in the implicit low rank approximation procedure described
in Lemma~\ref{lem:LowRankApprox}.
Applying Lemma~\ref{lem:DisplacementRepresentation} to each of the
operators involved gives that this cost is
\[
O\left( 
m \log^5{m} s^{\omega} \log\left( \alpha_{T}^{-1} \epsilon^{-1} \right) \right)
=
\Otil\left( 
m s^{\omega} \log\left( \alpha_{T}^{-1} \epsilon^{-1} \right) \right)
,
\]
while the overhead cost (from matrix factorizations/
orthogonalizations) is
\[
\Otil\left(m s^{\omega}
\cdot \log{m} \log\left(\alpha_{T}^{-1} \epsilon^{-1}\right)
\right).
\]
Combining these,
substituting in $\log(1 / \alpha_T) = O(\log(1 / \alpha_H))$,
and incorporating the additional factor of
$O(\log{m})$ corresponding to the number of layers of recursion
then gives the total construction cost.

The invocation/solve cost then follows from
the implicit multiplication guarantee
of Lemma~\ref{lem:DisplacementRepresentation}
with the $O(\log^2{m} \log( \alpha_H^{-1} \epsilon^{-1}))$
word lengths shown above.
\end{proof}

%% file: PadAndSolve.tex
\section{Pad and Solve Step}
\label{sec:Pad}

We now analyze the padding algorithm that takes us from a solver
for almost the full space, to the full algorithm for solving
systems in $A$.
Our padding step combined with the singular value bound
from Theorem~\ref{thm:RandKrylov} and the stable block-Hankel
matrix solver from Theorem~\ref{thm:Solver} gives the solver
for symmetric, eigenvalue separated matrices, as stated in
Lemma~\ref{lem:PadAndSolve}.

\PadAndSolve*

The algorithm is by taking the $K$ generated
from Theorem~\ref{thm:RandKrylov}, and padding
a number of random columns to it until it becomes full rank.
We first check that we can obtain a bound on the condition
number of the overall matrix.

\begin{lemma}
	\label{lem:Extend}
	If $M \in \Re^{n \times d}$ is a $n$-by-$d$ matrix with $d < n$,
	max entry magnitude at most $1/n$,
	and minimum singular value at least $\alpha_{M}$,
	then the $n$-by-$(d + 1)$ matrix formed by appending a dense
	length $n$ Gaussian vector scaled down by $\frac{1}{n^2}$:
	\[
	\left[ M, g \right]
	\qquad
	\text{with}~
	g \sim \frac{1}{n^2} \cdot \normal\left( 0, 1 \right)^{n}
	\]
	has maximum magnitude at most $1/n$,
	and minimum singular value at least $n^{-10} \alpha_{M}$
	with probability at least $1 - n^{-2}$.
\end{lemma}

\begin{proof}
    By Claim~\ref{claim:GaussianMax}, we may assume that all
    entries of $g$ have magnitude at most $1/n$.
    
    We now bound the minimum singular value.
    
	Because $d < n$, there is some unit vector $v$ that's normal
	to all $d$ columns of $A$.
	Consider $d^{T}g$: by anti-concentration of Gaussians,
	with probability at least $1 - n^{-3}$ we have
	\[
	\abs{d^{T} g} \geq n^{-5}.
	\]
	
	We claim in this case, the matrix $[M, g]$ has minimum singular
	value at least $n^{-10} \alpha_{M}$.
	Consider any test vector $x \in \Re^{d + 1}$:
	if the last entry
	$x_{d + 1}$ has absolute value less than $n^{-2} \alpha_{M} / 10$, then
	invoking the min-singular value bound on $A x_{1:d}$ gives
	\[
	\norm{M x_{1:d}}_2
	\geq
	\alpha_{M}\norm{x_{1:d}}_2
	\geq
	\alpha_{M} \left(1 - \abs{x_{d + 1}} \right)
	\geq
	\frac{\alpha_{M}}{2}.
	\]
	Then by triangle inequality we get:
	\[
	\norm{\left[M, g\right] x}_2
	\geq
	\norm{M x_{1:d}}_2
	-
	\norm{g x_{d + 1}}_2
	\geq
	\frac{\alpha_{M}}{2}
	-
	n^2 \cdot n^{-2} \alpha_{M} / 10
	\geq
	\frac{\alpha_{M}}{10}.
	\]
	
	So it remains to consider the case where $|x_{d + 1}| > n^{-2} \alpha_{M} / 10$.
	Here because the last column has dot at least $n^{-2}$
	against the normal of the previous $d$ columns of $M$, we get
	\[
	\norm{\left[M, g\right] x}_2
	\geq
	\min_{x_{1:d}}
	\norm{g x_{d + 1} - M x_{1:d}}_2
	=
	\abs{x_{d + 1}}
	\cdot
	\norm{g - M x_{1:d}}_2
	=
	\abs{x_{d + 1}}
	\left<g, v \right>
	\geq
	n^{-7} \alpha_{M} / 10.
	\]
\end{proof}

This allows us to bound the condition number of the
overall operator.

We also need to account for the costs of computing
the blocks of $(AK)^TAK$, as well as performing matrix
multiplications in $Q$ and $Q^T$.

\begin{lemma}
\label{lem:KrylovCosts}
Given implicit matrix-vector product access to an
$n$-by-$n$ matrix $A$
whose entries have max magnitude at most $\alpha_{A}^{-1}$,
and at most $O(\log(1 / \alpha_{A}))$ words after the decimal place
(for some $0 < \alpha_{A} \leq n^{-10}$),
as well as an $n$-by-$s$ matrix $B$ with at most $nnz(B)$
non-zero entries, each with magnitude at most
$\normi{B}_{\infty}$ and at most $L_B$ digits
after the decimal place, as well as $m$ such that $ms \leq O(n)$,
we can compute the matrix
\[
K
=
\left[
\begin{array}{c|c|c|c|c}
B & A^{1} B & A^{2} B & \ldots
& A^{m - 1} B
\end{array}
\right]
\]
at with the cost of multiplying $A$ against $n$ vectors,
each with at most $O(m \log(1 / \alpha_A) + \log\normi{B}_{\infty} + L_B)$ digits,
as well as all blocks in its Gram matrix
\[
B^T A^{i} B
\qquad
0 \leq i \leq 2m
\]
with additional cost
$O( n \cdot nnz( B ) \cdot (O(m \log(1 / \alpha_A) + \log \normi{B}_{\infty} + L_B) \log{n})$.
\end{lemma}

\begin{proof}
By repeated powering,
we can compute $A^{i}B$ via
\[
A^{i}B = A \cdot A^{i - 1}B.
\]
Each such powering increases the number of digits before
and after the decimal place by at most $O(\log(1 / \alpha_{A}))$,
so combining this over the $m$ steps gives the bound.

Note that the same also works for computing all the way
up to $A^{2m}B$.
Then we need to compute $B^T$ times each of these matrices.
This cost is $s$ multiplications of length
$O( m \log(1/\alpha_{A}) + \log(\normi{B}) + L_B)$.
Each such multiplication can be done using fast multiplication
with another overhead of $O(\log{n})$,
which gives a cost of
\[
O\left( s \cdot nnz\left(B \right)
\cdot
\left(\log\left( \normi{B}_{\infty} / \alpha_{A}\right) + L_B\right)\log{n} \right)
\]
per value of $i$, which times the $O(m)$ such values
and incorporating $sm \leq n$ gives the total.
\end{proof}

This gives the initialization cost.
Calling the block Hankel matrix solver,
and incorporating the guarantees into an overall operator
for the padded matrix then gives the overall running time.

\begin{proof}(Of Lemma~\ref{lem:PadAndSolve})

By Theorem~\ref{thm:RandKrylov},
the block Krylov space
$K$ has max singular value at most $n^2$,
and min singular value at least $\alpha_{K} = \alpha_{A}^{5m}$.
Consider the full padded matrix
\[
Q
=
\left[
\begin{array}{c|c}
K & G
\end{array}
\right]
=
\left[
\begin{array}{c|c|c|c|c|c}
G^{S} & A^{1} G^{S} & A^{2} G^{S} & \ldots
& A^{m - 1} G^{S} & AG
\end{array}
\right]
\]
where $G \in \Re^{n \times (n - ms)}$ is a dense
Gaussian matrix, and $G^{S}$ is a sparse Gaussian matrix with
entries set i.i.d to Gaussians with probability
$O(\frac{m^2 \log( 1 / \alpha_{A} )}{n})$.

Since $G$ has $O(m)$ columns,
applying Lemma~\ref{lem:Extend} inductively to
the extra columns gives that $Q$ has max singular value
at most $n^2$, and min singular value at least
\[
\alpha_{Q}
\geq
\alpha_{K} \cdot n^{-10 \cdot 10m}
\geq
\alpha_{A}^{10m}.
\]

This in turn means that errors in $B$ by at most
$\alpha_{A}^{-20m}$ will produce error at most
$\alpha_{A}/2$ in the inverse.
So we can round all digits in $B$ to such error,
obtaining $L_B \leq O(m \log(1 / \alpha_{A}))$ as well.
Combining this with $\normi{B}_{\infty} \leq 1$ gives
that the cost of the initialization steps given by
Lemma~\ref{lem:KrylovCosts} is
$n$ matrix-vector multiplications of $A$ against 
length $n$ vectors with $O(m \log(1 / \alpha_A))$ entries,
plus an additional cost of
\[
O\left( 
n \cdot nnz\left( B \right) \cdot m \log\left( 1 / \alpha_{A} \right) \log{n} \right)
\leq
\Otil\left( 
n^2 m^{3} \log^{2}\left( 1 / \alpha_{A} \right) \right)
\]
where the inequality follows from $B$ having $s \leq n /m$ columns,
each of which have at most $\Otil(m^{3} \log(1 / \alpha_A))$ non-zeros
with high probability.

Theorem~\ref{thm:RandKrylov} gives that with probability
at least $1 - n^{-2}$,
the min eigenvalue of $(AK)^TAK$ is at least
\[
\left(\alpha_{A} \cdot \alpha_{K}\right)^2
\geq \alpha_{A}^{20m}.
\]
So the block-Hankel solver from Theorem~\ref{thm:Solver}
with error 
\[
\epsilon \leftarrow \alpha_{A}^{1000m}
\]
requires construction time:
\[
\Otil\left( m s^{\omega}
    \cdot \log\left( \alpha_{H}^{-1} \epsilon^{-1} \right) \right)\\
\leq
\Otil\left( m \left( \frac{n}{m} \right)^{\omega}
\cdot m \log\left( 1 / \alpha_{A} \right) \right)
\leq
\Otil\left( n^{\omega} m^{2 - \omega}
    \log\left( 1 / \alpha_{A} \right) \right).
\]
It gives access to a solver operator 
gives an operator $Z_{H}$ such that
\[
\norm{Z_{H} - \left( \left(AK\right)^T AK\right)^{-1} }_{2}
\leq
\alpha_{A}^{1000m}.
\]

Now consider the matrix 
\[
\left(AQ\right)^TAQ
=
\left[
\begin{array}{c|c}
AK & AG
\end{array}
\right]^T
\left[
\begin{array}{c|c}
AK & AG
\end{array}
\right]
=
\left[
\begin{array}{c|c}
\left(AK\right)^TAK & \left( AK\right)^T AG\\
\hline
\left(AG\right)^TAK & \left(AG\right)^T AG
\end{array}
\right]
\]
The bounds that we have on $Q$ gives that the max and
min singular values of this matrix is at most
$\alpha_{A}^{-100m}$ and $\alpha_{A}^{100m}$ respectively.
So we can apply Lemma~\ref{lem:ErrorInvert} repeatedly
to replace the top-left block by $Z_{H}$:
\begin{enumerate}
    \item First, by the eigenvalue bounds on
    $(AK)^TAK$, we have
    \[
    \norm{Z_{H}^{-1} - \left(AK\right)^TAK}_{F}
    \leq
    \alpha_{A}^{800m},
    \]
    which when block-substituted into the overall
    formula implies
    \[
    \norm{
    \left[
        \begin{array}{c|c}
            Z_{H}^{-1} & \left( AK\right)^T AG\\
            \hline
            \left(AG\right)^TAK & \left(AG\right)^T \left(AG\right)
        \end{array}
    \right]
    - \left(AQ\right)^TAQ
    }_{F}
    \leq
    \alpha_{A}^{800m}.
    \]
    \item Inverting this again using the eigenvalue
    bound
    \[
    \norm{
    \left[
        \begin{array}{c|c}
            Z_{H}^{-1} & \left( AK\right)^T AG\\
            \hline
            \left(AG\right)^TAK & \left(AG\right)^T \left(AG\right)
        \end{array}
    \right]^{-1}
    - \left(\left(AQ\right)^TAQ\right)^{-1}
    }_{F}
    \leq
    \alpha_{A}^{800m}.
    \]
\end{enumerate}

This new block matrix can also be further factorized as:
\begin{multline*}
\left[
\begin{array}{c|c}
Z_{H}^{-1} & \left(AK\right)^TAG\\
\hline
\left(AG\right)^TAK & \left(AG\right)^T \left(AG\right)
\end{array}
\right]
=
\left[
\begin{array}{c|c}
I & \\
\hline
\left(AG\right)^TAK Z_{H}& I
\end{array}
\right]\\
\left[
\begin{array}{c|c}
Z_{H}^{-1} & 0\\
\hline
0 &
\left(AG\right)^T \left(AG\right)
- \left(AG\right)^TAK Z_{H} \left(AK\right)^{T} AG
\end{array}
\right]
\left[
\begin{array}{c|c}
I & Z_{H} \left(AK\right)^T AG\\
\hline
0 & I
\end{array}
\right]
\end{multline*}
which upon inverting becomes
\begin{multline*}
\left[
\begin{array}{c|c}
Z^{-1} & \left(AK\right)^TAG\\
\hline
\left(AG\right)^TAK & \left(AG\right)^T \left(AG\right)
\end{array}
\right]
=
\left[
\begin{array}{c|c}
I & -Z_{H} \left(AK\right)^T AG\\
\hline
0 & I
\end{array}
\right]\\
\left[
\begin{array}{c|c}
Z_{H} & 0\\
\hline
0 &
\left( \left(AG\right)^T AG
- \left(AG\right)^T AK Z_{H} \left(AK\right)^{T} AG \right)^{-1}
\end{array}
\right]
\left[
\begin{array}{c|c}
I & \\
\hline
-\left(AG\right)^T AK Z_{H}& I
\end{array}
\right]
\end{multline*}
Recall from the start of Section~\ref{sec:Overview}
that our definition of linear algorithms are that they
exactly evaluate their corresponding operators.
As we have access to $Z_{H}$, as well as a multiplications by
$A$, $K$, and $G$ (and their transposes), we are able to
evaluate the first and third term exactly.
So the only place where additional errors arise are in the
the inverse of the bottom-right block of the middle term.
Any error in approximating it will get multiplied by the magnitude
of the previous and subsequent matrices.
The magnitude of the upper/lower triangluar matrices at the start/end
are bounded by
\[
1 + \norm{Z_{H} K^T A^2 G}_2
\leq
1 + \norm{Z_{H}}_2 \norm{K}_2 \norm{A}_2^2 \norm{G}_2
\leq
\alpha_{A}^{-100m},
\]
So it suffices to invert the Schur complement term
in the middle, specifically $(AG)^T(AG) - (AG)^T AK Z_{H} (AK)^T AG$,
to an additive accuracy $\alpha_{A}^{500m}$.

We thus obtain an operator $Z_{Q}$ such that
\[
\norm{Z_{Q}
 - \left( \left(AQ\right)^T AQ\right)^{-1}
 }_{F}
 \leq
 \alpha_{A}^{300m}.
\]
Finally, note that both $Q$ and $(AQ)^{T}$
are matrices with magnitude at most $\alpha_{A}^{-100m}$,
and
\[
Q
\left(\left(AQ\right)^TAQ\right)^{-1}
\left(AQ\right)^{T}
=
Q Q^{-1} A
\left(AQ\right)^{-T}
\left(AQ\right)^{T}
=
A^{-1}.
\]
Substituting this in then gives:
\begin{multline*}
\norm{A^{-1} - 
Q
Z_{Q}
\left(AQ\right)^T
}_2
=
\norm{
Q
\left( \left( \left(AQ\right)^TAQ \right)^{-1} - Z_{Q} \right)
\left(AQ\right)^T
}_2\\
\leq
\norm{Q}_2
\norm{\left( \left(AQ\right)^T AQ \right)^{-1} - Z_{Q} }_{F}
\norm{AQ}_2
\leq
\alpha_{A}^{100m}
\leq \alpha_{A} / n^{2},
\end{multline*}
which means the final post-processing step gives
the desired error guarantees.

For the bit-complexity of the operator $Z_{A}$,
Theorem~\ref{thm:Solver} gives that the numbers of words
after decimal place in $Z_{H}$ is at most
\[
O\left( \log^{2}{n} \log\left( \alpha_{H}^{-1} \epsilon^{-1} \right) \right)
=
O\left( m \log^{2}n \log\left( 1 / \alpha_{A} \right) \right).
\]
The matrices $A$, $K$, $Q$ have at most
$O(m \log(1 / \alpha_{A}))$ words after the decimal place.
So because the multiplications only involve a constant
number of such operators, the maximum number of digits
after the decimal place we'll have is also
$O(m \log^2{n} \log(1 / \alpha_{A}))$, and the max magnitude of
an entry we'll encounter is $\alpha_A^{-O(m)}$.

Thus, the solve costs from Theorem~\ref{thm:Solver}
needs to be invoked on vectors whose entries have magnitude at most
$\alpha_A^{-O(m)} \normi{b}_{\infty}$, and at most
$O(L_B + m \log^{2}{n} \log( 1 / \alpha_{A} ))$
after the decimal point, giving:
\begin{multline*}
 \Otil\left(m \cdot s^2
\cdot
\left(\log\left(
\frac{\left( 1 + \alpha_A^{-O(m)} \normi{b}_{\infty}\right)
n}{\alpha_{H} \epsilon}\right)
+  m \log^{2}{n} \log\left( 1 / \alpha_{A}\right)
+ L_B \right)
\right)\\
\leq
\Otil\left( m s^2 
\cdot
\left( m \log^2{n} \log\left( 1 / \alpha_A \right)
+ \log\normi{B}_{\infty} + L_B \right)
\right)\\
\leq
\Otil\left(n^2 \left(\log\left( 1 / \alpha_A \right)
+ \log\normi{B}_{\infty} + L_B \right) \right).
\end{multline*}
On the other hand, this vector
with $O(m \log^{2}{n} \log( 1 / \alpha_{A} )
+ \log\normi{B}_{\infty} + L_B)$ words per entry
needs to be multiplied against $Q$ and $K$, which are dense
matrices with $n^2$ entries of word-size at most
$O(m \log(1 / \alpha_{A}))$.
So the total cost is:
\[
\Otil\left(n^2 m
\left( \log\left( 1 / \alpha_{A} \right)
  + \log\normi{B}_{\infty} + L_B \right) \right),
\]
which is more than the above term from the Hankel matrix
solver due to $ms \leq n$ and $m > n^{0.01}$.
\end{proof}

%% file: Vandermonde.tex
\section{Properties of Vandermonde Matrices}
\label{sec:Vandermonde}

We prove the large entries property of Vandermonde matrices.

\Vandermonde*

We first prove the square case.

\begin{lemma}
\label{lem:VandermondeSquare}
Let $0 < \alpha < 1$ be a bound and
$\alpha < \lambda_{1} \leq \lambda_{2} \leq \ldots \leq \lambda_{m}
\leq \alpha^{-1}$
be positive real numbers such that
$\lambda_{i + 1} - \lambda_{i} \geq \alpha$,
then the Vandermonde matrix
\[
V
=
\left[
	\begin{array}{ccccc}
	1 & \lambda_{1} & \lambda_{1}^{2} & \ldots & \lambda_{1}^{m - 1}\\
	1 & \lambda_{2} & \lambda_{2}^{2} & \ldots & \lambda_{2}^{m - 1}\\
	\ldots & \ldots & \ldots & \ldots & \ldots\\
	1 & \lambda_{m} & \lambda_{m}^{2} & \ldots & \lambda_{m}^{m - 1}
\end{array}
\right]
\]
is full rank, and has minimum singular value at
least $m^{-1} 2^{-m} \alpha^{m}$.
\end{lemma}
	
\begin{proof}
We will bound the minimum singular value by bounding the
maximum magnitude of an entry in $V(\lambda, m)^{-1}$.
We will in turn bound this by explicitly writing down the inverse
using Lagrange interpolation of polynomials.
	
For a vector
$x = [ x_1 ; x_2; \ldots ; x_{m} ]$,
consider the polynomial in $t$ with coefficients
given by $x$, specifically
\[
p_{x}\left(t\right)
= 
\sum_{i = 1}^{m} x_{i} t^{i - 1}.
\]
The matrix-vector product $V x$
can be viewed as evaluating $p_{x}(\cdot)$
at $\lambda_{1}, \lambda_{2}, \ldots \lambda_{m}$:
\[
V x
= \left[ p_{x} \left( \lambda_{1} \right) ; p_{x} \left( \lambda_{2} \right); \ldots ; p_{x} \left( \lambda_{m} \right) \right].
\]
This means that if we are given these polynomial values
as a vector $b$,
we can solve for the value of $x$ using
the Lagrange interpolation formula
\[
p_{x} \left( t \right)
=
\sum_{i = 1}^{m}
	b_{i} \cdot \prod_{j \neq i}
	\frac{\left( t - \lambda_{j} \right)}
	{\left( \lambda_{i} - \lambda_{j} \right)}.
\]

Thus the inverse is formed by collecting coefficients
on each of the monomial $t^{i}$s.
To do this, consider expanding each of the product terms.
Since $\lambda_{j} \leq \alpha^{-1}$,
each of the numerator's contribution is at most
$\alpha^{-m}$.
Also, since $|\lambda_{i} - \lambda_{j}| \geq \alpha$,
the denominator contributes at most $\alpha^{-1}$ as well.
		
Coupled with the fact that there are at most $2^{m}$
terms in expansions, each entry of inverse has magnitude
at most $2^{m} \alpha^{-2m}$.
Plugging in any unit vector $x$ then gives
\[
\norm{x}_2
\leq
\norm{V^{-1}}_2
\norm{Vx}_2,
\]
which plugging in $\norm{V^{-1}}_2 \leq m \cdot 2^{m} \alpha^{-2m}$
above gives $\norm{Vx}_2 \geq m^{-1} 2^{-m} \alpha^{2m}$.
\end{proof}

For the rectangular case, we simply apply
Lemma~\ref{lem:VandermondeSquare} to every subset
of $m$ rows.

\begin{proof} (of Lemma~\ref{lem:Vandermonde})
If there is some unit $x$ such that $V x$ has $m$
entries with magnitude less than $\alpha^{3m}$,
then let this subset be $S$.
We have
\[
\norm{V_{S, :} x}_2
\leq
m^{-1} 2^{-m} \alpha^{2m}.
\]
On the other hand, because $\lambda_{S}$ is a subset of $\lambda$,
it is also in the range $[\alpha, \alpha^{-1}]$,
is separated by at least $\alpha$.
So we get a contradiction with Lemma~\ref{lem:VandermondeSquare}.
\end{proof}

%% file: DisplacementRankDetails.tex
\section{Properties and Algorithms for
Low Displacement Rank Matrices}
\label{sec:DisplacementRankDetails}

We first check the preservation of displacement rank
(up to a sign/direction flip) under inversion.

\DisplacementRankInverse*

The proof of this lemma relies on the following generalization
of `left inverse is right inverse'.
\begin{fact}
	\label{fact:FlipRank}
	For any square (but not necessarily invertible)
	matrices $M^{(1)}$ and $M^{(2)}$, we have
	\[
	\textsc{Rank}\left(I - M^{\left( 1\right)} M^{\left( 2\right)} \right)
	=
	\textsc{Rank}\left(I - M^{\left( 2\right)} M^{\left( 1\right)} \right).
	\]
\end{fact}

\begin{proof}
We will show $\textsc{Rank}(I - M^{(1)}M^{(2)})
\geq \textsc{Rank}(I - M^{(2)} M^{(1)})$, or
$\textsc{Null}(I - M^{(1)}M^{(2)})
\leq \textsc{Null}(I - M^{(2)} M^{(1)})$.
The other direction follows from flipping the role of 
$M^{(1)}$ and $M^{(2)}$ and applying the argument again.

Let dimension of the null space of $I - M^{(1)}M^{(2)}$ be $r$.	
Let $x^{(1)}, x^{(2)}, \ldots x^{(r)}$ be a basis for the
null space of $I - M^{(1)}M^{(2)}$.
The condition
\[
\left( I - M^{\left( 1\right)} M^{\left( 2\right)} \right)
x^{\left( \rhat\right)}
=
0
\qquad
\forall 1 \leq \rhat \leq r
\]
rearranges to
\[
x^{\left( \rhat \right)}
=
M^{\left( 1\right)} M^{\left( 2\right)} x^{\left( \rhat \right)}
\qquad
\forall 1 \leq \rhat \leq r
\]
then for each $\rhat$, $M^{(2)} x^{(\rhat)}$ is also
in the null space of
$I - M^{(2)} M^{(1)}$:
\[
\left( I - M^{\left( 2\right)} M^{\left( 1\right)} \right)
M^{\left( 2\right)} x^{\left( \rhat \right)}
= M^{\left( 2\right)} x^{\left( \rhat \right)}
- M^{\left( 2\right)} M^{\left( 1\right)}
    M^{\left( 2\right)} x^{\left( \rhat \right)}
= M^{\left( 2\right)} x^{\left( \rhat \right)}
- M^{\left( 2\right)} x^{\left( \rhat \right)}
= 0
\]

So it remains to show that $M^{(2)} x^{(\rhat)}$s are
linearly independent.
Note that because $M^{(2)}$ may not be invertible,
we cannot just conclude that $x^{(\rhat)}$ are themselves
linearly independent.
Instead, we need to use the fact that $x^{(\rhat)}$
is invariant under $M^{(1)} M^{(2)}$ to `locally invert'
$M^{(2)} x^{(\rhat)}$ using $M^{(1)}$.
Formally,
suppose $c_1 \ldots c_r$ not all $0$ are coefficients such that
\[
\sum_{1 \leq \rhat \leq r} c_{\rhat}
    M^{\left(2\right)} x^{\left( \rhat\right)}
=
0
\]
then
\[
0
=
M^{\left(1\right)}
\left(
\sum_{1 \leq \rhat \leq r}
c_{\rhat}  M^{\left(2\right)}
  x^{\left( \rhat \right)}
\right)
=
\sum_{1 \leq \rhat \leq r}
c_{\rhat} M^{\left(1\right)} M^{\left(2\right)}
  x^{\left( \rhat \right)}
=
\sum_{1 \leq \rhat \leq r} c_{\rhat} x^{\left( \rhat \right)},
\]
a contradiction with the assumption that $x^{(\rhat)}$s
are linearly independent.
\end{proof}
	
\begin{proof}(of Lemma~\ref{lem:DisplacementRankInverse})
Because $M$ is square and full rank,
the rank of a matrix is preserved when multiplying
by $M$ and $M^{-1}$.
\[
\alpha_{-}\left( M^{-1} \right)
=
\textsc{Rank} \left( M^{-1} - \Delta\left( s \right)^{T} M^{-1}  \Delta\left( s \right) \right)
=
\textsc{Rank} \left( I - \Delta\left( s \right)^{T} M^{-1}  \Delta\left( s \right) M\right).
\]
Invoking Fact~\ref{fact:FlipRank}
on the matrices $\Delta(s)^T M^{-1}$ and $\Delta(s) M$ then gives
\[
\alpha_{-}\left( M^{-1} \right)
=
\textsc{Rank} \left( I -  \Delta\left( s \right) M  \Delta\left( s \right)^{T} M^{-1}\right)
\]
which when multiplied by $M$ gives
\[
=
\textsc{Rank} \left( M - \Delta\left( s \right) M  \Delta\left( s \right)^{T} \right)
=
\alpha_{+} \left( M \right).
\]
\end{proof}

Next we check that representations in the displaced
version can be used to efficiently perform matrix-vector
multiplications in the original matrix.

\DisplacementRepresentation*

Our goal is to use matrix-vector convolutions.
To do this, we first represent $M$ as the product of
upper and lower (block) triangular matrices that are also
Toeplitz: every diagonal consists of the same blocks.
This property is in turn useful because multiplying by these
special types of matrices can in turn be carried using
fast convolutions.

For notational simplicity (specifically to avoid using $U$ for upper-triangular matrices),
we only use lower triangular forms of these, which we denote $T_L$.

\begin{definition}
For an $(ms) \times s$ matrix $X$,
the corresponding block lower-triangular
Toeplitz matrix $T_{L}(X)$ is given by placing $X$ on the
leftmost column,
\[
T_L\left( X \right)
= 
\left[
\begin{array}{ccccc}
X_{\left\{1, 1\right\}} &  0 & 0 & \ldots  & 0 \\
X_{\left\{2, 1\right\}} & X_{\left\{1, 1\right\}} &  0 & \ldots & 0\\ 
X_{\left\{3, 1\right\}} & X_{\left\{2, 1\right\}} & X_{\left\{1, 1\right\}} & \ldots & 0\\ 
\ldots & \ldots & \ldots & \ldots & \ldots \\
X_{\left\{m, 1\right\}} & X_{\left\{m-1, 1\right\}} & X_{\left\{m - 2, 
		1\right\}} & \ldots & X_{\left\{1, 1\right\}}
\end{array}
\right]
\]
or formally
\[
T_L\left(X\right)_{\left\{i, j \right\}}
=
\begin{cases}
X_{\left\{i - j + 1, 1 \right\}} & \qquad \text{if $j \leq i$},\\
0 & \qquad \text{otherwise}.
\end{cases}
\]
\end{definition}

It was shown in~\cite{KailathKM79} that the product applying
$T_L(\cdot)$ to the two factors gives a representation of the
original displacement matrix.
Note that both the displacement operation, and the computation
of $XY^T$ is linear in the matrices.
Therefore, we can view $X$ and $Y$ as $t$ copies of $n \times s$
matrices, with possibly padding by $0$s to make their second
dimension a multiple of $s$.
\begin{lemma}
\label{lem:RepresentationTranslation}
For any number $t$ and any sequence of $(ms) \times s$ matrices
$X^{(1 \ldots t)}$ and $Y^{(1 \ldots t)}$,
the matrix $M^{(+)}$ satisfying the equation
\[
+s\left( M^{\left( + \right)} \right)
= \sum_{\that = 1}^{t} X^{\left( \that\right)}
    \left( Y^{\left( \that\right)} \right)^T
\]
has the unique solution
\[
M^{\left( + \right)}
= \sum_{\that = 1}^{t}
    T_{L}\left( X^{\left( \that\right)} \right)
    T_{L}\left( Y^{\left( \that\right)} \right)^T
\]
and similarly the unique solution to
$s-(M^{(-)}) = XY^T$ is
\[
M^{\left( - \right)}
= \sum_{\that = 1}^{t}
    T_{L}\left( X^{\left( \that\right)} \right)^T
    T_{L}\left( Y^{\left( \that\right)} \right)
\]
\end{lemma}

\begin{proof}
By symmetry (in reversing both the order of rows
and columns), we focus on the case with $M^{(+)}$.
We first show that the solution is unique.

Since $\Delta(s)$ shifts all entries down by $s$,
we have
\[
\left(M - \Delta\left( s \right) M \Delta\left( s \right)^{T}\right)_{ij}
=
\begin{cases}
    M_{ij} - M_{\left(i - s\right),\left( j - s \right)} & \qquad \text{if $i > 1$ and $j > 1$}\\
	M_{ij} & \qquad \text{otherwise}
\end{cases}
\]
Thus we can construct $M$ uniquely from $+s(M)$ by first filling
in the first $s$ rows and columns of $M$ with the same values as
in $+s(M)$, and iteratively constructing the rest via
\[
M_{ij}
=
+s\left( M \right)_{ij}
+ M_{\left( i - s\right), \left( j - s \right)}.
\]
This means that the solution to $+s(M^{(+)})
    = \sum_{\that} X^{( \that)} (Y^{( \that)})^{T}$
is unique, and all we have to do is to show that the formula
we provided gives equality.
	
First, we express the $\{i, j\}$ block of
the product of a lower triangular Toeplitz matrix
and an upper triangular Toeplitz matrix can be
in terms of the original blocks:
\[
    \left(T_L\left( X \right) T_L\left( Y \right)^{T} \right)_{\left\{i, j \right\}}
    =
	\sum_{k} T_{L}\left( X \right)_{\left\{i, k \right\}}
    			T_L\left( Y \right)_{\left\{j, k \right\}}^T
    =
	\sum_{1 \leq k \leq \min\left\{i, j \right\} }
	    X_{\left\{i + 1 - k, 1\right\}} Y_{\left\{j + 1 - k, 1\right\}}^{T}
\]
Now consider the matrix with each row shifted down by $1$
that results from multiplying by $\Delta(s)$.
This leads to
\[
\left( \Delta\left( s \right) T_L\left( X \right)
    T_L\left( X \right)^T \Delta\left( s \right)^T \right)
=
\left( \Delta\left( s \right) T_L\left( X \right) \right)
\left( \Delta\left( s \right) T_L\left( Y \right) \right)^T.
\]
Which when substituted into the above formula gives:
\[
\left( \Delta\left( s \right) T_L\left( X \right)
    T_L\left( X \right)^T \Delta\left( s \right)^T \right)_{\left\{i, j\right\}}
=
\sum_{1 \leq k \leq \min\left\{i, j \right\} - 1}
    X_{\left\{i - k, 1\right\}} Y_{\left\{j - k, 1\right\}}^{T}.
\]
Upon comparison, the only different term is
$X_{\{i, 1\}} Y_{\{j, 1\}}^{T}$,
which is precisely the corresponding block in $XY^{T}$.
Note that the column indices are $1$ in both
of these because $X$ and $Y$ are both $ms$-by-$s$,
so exactly one column block.
\end{proof}
	
Efficient multiplications against $T_L(X)$ and $T_L(X)^T$
can in turn be realized (in an operator manner) via fast
Fourier transforms.

First, observe that for an $ms$-by-$k$ matrix $B$,
$T_L(X)B$ is an $n$-by-$s$ matrix.
If we interpret it as $m$ $s$-by-$s$ blocks, these
blocks the result of computing the convolution 
of $X_{\{1, :\}} \ldots X_{\{m, :\}}$ with
$B_{\{m, :\}} \ldots B_{\{1, :\}}$.

Raising $m$ to a power of $2$, and filling the rest
of the blocks with $0$s means it suffices to show that
we can generate a covolution operator with good
bit complexity.
Due to the connection with Fourier transforms, it is
more convenient for us to use $0$-indexing of the blocks.

\begin{restatable}{lemma}{StableConvolution}
\label{lem:StableConvolution}
    Given $t$ that is a power of $2$, along with
    a length $t$ sequence of $s$-by-$s$ matrices
    $X^{(0)} \ldots X^{(t - 1)}$, and any error $\delta$,
    there is an algorithm corresponding to
    a linear operator $\Ztil_{Conv(X, \delta)}$ with at most
    $O( \log{t} \log(t s ( 1 + \normi{X}_{\infty} ) / \delta))$
    words after the decimal point such that for the
    exact convolution matrix $\Zbar_{Conv(X)}$ given by:
    \[
    \Zbar_{Conv\left(X\right)\left\{i, j\right\}}
    =
    X^{\left(\left(i - j\right) \mod m\right)}
    \]
    we have
    \[
    \norm{\Ztil_{Conv\left(X, \delta\right)}
    - \Zbar_{Conv\left(X\right)}}_{F} \leq \delta
    \]
    and for any length $t$ sequence of $s$-by-$k$ matrices
    $B^{(0)} \ldots B^{(t - 1)} \in \Re^{s \times k}$ with
    at most $L_{B}$ words after decimal point,
    corresponding vertically concatenated matrix
    $B \in \Re^{ts \times k}$,
    evaluating $\Ztil_{Conv(X, \delta)}B$ takes time
    \[
    O\left( t \log^2{t}
    \cdot
    \max\left\{s^2k^{\omega - 2}, s^{\omega - 1}k \right\}
    \cdot \left( \log\left(
    \left( 1 + \normi{X}_{\infty} \right)
    \left( 1 + \normi{B}_{\infty} \right)
    ts / \delta \right) + L_B \right)
    \right)
    \]
    in the unit-cost RAM model.
\end{restatable}

Applying this convolution twice to the lower/upper
block triangular matrices then gives the overall algorithm.

\begin{proof} (of Lemma~\ref{lem:DisplacementRepresentation})
Applying Lemma~\ref{lem:StableConvolution} to the blocks
of $X$ and $Y$ for some error $\deltahat$ that we will set
at the end of this proof
gives that there are routines for multiplying by
$T_L(X)$ and $T_L(Y)^T$ that correspond to linear operators
such that:
\begin{enumerate}
    \item The error in Fronbenius norm is at most $\deltahat$.
    \item The number of digits after decimal point is at most
    $O(\log{m} \log(ms (1 + \normi{X}_{\infty}) (1+ \normi{Y}_{\infty} ) / \deltahat ))$
    \item For am $ms$-by-$k$ matrix $B$ with $L_B$ digits
    after the decimal point, the evaluation cost for these
    operators is
    $O(m \log^2 m \max\{ s^2 k^{\omega} 
        (\log( (1 + \normi{X}_{\infty}) (1 + \normi{Y}_{\infty})  (1 + \normi{B}_{\infty}) ms / \deltahat) + L_B)) $.
\end{enumerate}

We then composing these operators
for $L_T(X^{(\that)})$ and $L_T(Y^{(\that)})^T$
via the composition statement from Lemma~\ref{lem:ErrorCompose}.
Note that we have
\[
\normi{T_L\left(X\right)}_{\infty}
\leq \normi{X}_{\infty}
\]
as $T_L()$ simply duplicates the blocks in $X$.
So it suffices to choose
\[
\deltahat
\leftarrow
\frac{\delta}{ms
\left( 1 + \normi{X}_{\infty} \right)
\left( 1 + \normi{Y}_{\infty} \right)}.
\]

Substituting this value into the bit-length of the operators
gives that the number of bits in them is still
$O(\log{m} \log(ms (1 + \normi{X}_{\infty}) (1+ \normi{Y}_{\infty} ) / \delta ))$,
which in turn gives that the maximum number of words after
decimal point passed as input when invoking them,
for a particular $B$ is at most
\[
O\left(\log{t} \log\left(ms
\left(1 + \normi{X}_{\infty}\right)
\left(1 + \normi{Y}_{\infty} \right) / \delta
\right) + L_B\right).
\]
Substituting this into the algorithmic costs of
Lemma~\ref{lem:StableConvolution} along with the
$\lceil r / s \rceil$ involved gives a total cost of
\begin{align*}
& \left\lceil \frac{r}{s} \right\rceil
\cdot
O\left( m \log^3{m}
\cdot
\max\left\{s^2k^{\omega - 2}, s^{\omega - 1}k \right\}
\right. \\ & \qquad \left.
\cdot \left( \log\left(
\left( 1 + \normi{X}_{\infty} \right)
\left(1 + \normi{Y}_{\infty} \right)
\left( 1 + \normi{B}_{\infty} \right)
ms / \delta \right) + L_B \right)
\right)\\
&=
O\left( m \log^3{m}
\cdot
\max\left\{r, s \right\}
\max\left\{sk^{\omega - 2}, s^{\omega - 2}k \right\}
\right. \\ & \qquad \left.
\cdot \left( \log\left(
\left( 1 + \normi{X}_{\infty} \right)
\left(1 + \normi{Y}_{\infty} \right)
\left( 1 + \normi{B}_{\infty} \right)
ms / \delta \right) + L_B \right)
\right)
\end{align*}

\end{proof}

%% file: Convolution.tex
\section{Convolution Operator}
\label{sec:Convolution}

In this section we bound the word complexity of the
operator version of fast convolution.

\StableConvolution*

Note that all numbers in $X$ can be truncated to
$O(\log(t / \epsilon))$ words after the decimal point.
After this, the easiest way of doing this is to treat
the numbers as integers, and use modulus against primes
(of size around $t$) to reduce this to computing
convolutions over finite fields via the exact same
formulas below (with $\omega$ replaced by a generator
over said finite fields).
Due to the involvement of different number systems,
we omit details on such an implementation in favor
of more thorough expositions,
e.g. Chapters 30 and 31 of the Third Edition of
Introduction to Algorithms
by Cormen-Leiserson-Rivest-Stein~\cite{CormenLRS09:book}.

Below we also verify how to do this using decimal expansions.
We compute this convolution via real-valued Fast Fourier
Transform (FFT) with all numbers rounded to $O(1 / \epsilon)$
words of accuracy.
We first briefly describe how the exact FFT algorithm
works: a detailed explanation can be found in Chapter 30~\cite{CormenLRS09:book}.

Let
\[
\omega = \cos\left( \frac{2 \pi }{t} \right)
    + i \cdot \sin\left( \frac{2 \pi }{t}\right)
\]
be a $t\textsuperscript{th}$ roof of unity.
Then the $i\textsuperscript{th}$ block of the
Fourier transform of $X$ is given by
\[
\sum_{j} \omega^{i \cdot j} X^{\left(j\right)}
\]
the corresponding block for $B$ is given by:
\[
\sum_{k} \omega^{i \cdot k} B^{\left(k\right)}.
\]
Taking the product of these two blocks
(with the same index $i$) then leads to:
\[
\sum_{j, k} \omega^{i \cdot \left( j + k \right)}
    X^{\left(j\right)} B^{\left(k\right)}.
\]
Note that the corresponding blocks have dimensions
$s$-by-$s$ and $s$-by-$k$ respectively, so their
product gives an $s$-by-$k$ block.

Applying inverse FFT to these then gives blocks of the form (for index $l$):
\[
\sum_{i, j, k} \omega^{-il} \cdot \omega^{i \cdot \left( j + k \right)}
    X^{\left(j\right)} B^{\left(k\right)}
=
\left[
\sum_{i} \omega^{i \cdot \left( j + k - l \right)}
\right]
\cdot
\sum_{j, k} X^{\left(j\right)} B^{\left(k\right)}.
\]
The properties of roots of unity then gives that this
leading coefficient is $t$ if $j + k \equiv l \pmod{t}$,
and $0$ otherwise.
So the above equations, when implemented exactly, gives
$t$ times the convolution between $X$ and $B$.

We first sketch how the Fast Fourier Transform
can be efficiently implemented as an linear operator
in the presence of roundoff errors.
Note that notions of Frobenius/$\ell_2$ norms extend to
the complex setting as well.
\begin{lemma}
\label{lem:StableFFT}
Given any value $t$ and block size $s$, and any $0 < \delta < 0.1$,
there is an algorithm that
correspond to a linear operator $\Ztil_{FFT(t, s, \delta)}$ with at most
$O( \log( t /  \delta) \log{t})$ digits after the decimal
place such that
\[
\norm{\Ztil_{FFT(t, s, \delta)} - \Zbar_{FFT(t, s)}}_{F}
\leq \delta,
\]
where $\Zbar_{FFT(t, s)}$ is the exact FFT operator
on $t$ blocks of size $s$:
\[
\left(\Zbar_{FFT(t, s)}\right)_{\left\{i, j\right\}}
=
\omega^{i \cdot j} \cdot I\left( s \right)
\qquad \forall 1 \leq i, j \leq t,
\]
where $I(s)$ is the $s$-by-$s$ identity matrix.
Given any matrix-sequence $X \in \Re^{ts \times k}$
with at most $L_X$ words after the
decimal place, $\Ztil_{FFT(t, s, \delta)} X$ can be computed in time
$O(t \log{t} 
\cdot sk \cdot (1 + \log(\normi{X}_{\infty}) + L_X + \log{t}\log(t / \delta)))$.

A similar bound also holds for the inverse FFT operator,
$\Ztil_{InvFFT(t, s)}$, specifically
$\norm{\Ztil_{InvFFT(t, s, \delta)} - Z_{InvFFT(t, s)}}_{F} \leq \delta$ with the same running time / word length bounds.
\end{lemma}

\begin{proof}
We introduce round-off errors into the Fast Fourier transform algorithm,
but only in one place: the generation of each complex coefficient
used to evaluate the butterfly diagram.

Because all coefficients in the FFT matrix have magnitude at most
$1$, errors to the corresponding coefficients
accumulate additively.
So it suffices to have round-off error at most
$\delta / poly(t)$ these coefficients.

This implies at most $O( \log(t / \delta))$ words after decimal
point for each of the coefficients.
As the depth of the FFT recursion is $O(\log{t})$, the total length
of the words are bounded by $O( \log{t} \log( t / \delta))$.
This in turn implies that any linear combination of these coefficients
and $X$ can have at most $O(L_X +  \log{t} \log( t / \delta)))$
words after the decimal place.
Combining this with none of the intermediate numbers
exceeding $t \normi{X}_{\infty}$, and the $O(sk)$ cost of adding matrices gives the overall cost.
\end{proof}

\begin{proof} (Of Lemma~\ref{lem:StableConvolution})
We first round all entries in $X$ to an additive
accuracy of $\delta / (ts)$, or $O(\log(ts / \delta))$
words after the decimal point.
This perturbs the convolution operator in $X$
by at most $\delta / ts$.

We apply the FFT operators to $X$ and $B$
separately, namely we first invoke the FFT algorithm
from Lemma~\ref{lem:StableFFT} to compute
$\Ztil_{FFT(t, s, \deltahat)} X$
and
$\Ztil_{FFT(t, s, \deltahat)} B$
for some $\deltahat$ that we will choose later as a function of $\delta$, $\normi{X}_{\infty}$, $\normi{B}_{\infty}$, and $ms$.

The cost of these two (forward) FFT transforms is
then at most
\begin{multline*}
O\left( t \log{t} \cdot s^2 \cdot
\left( \log \left( 1 + \normi{X}_{\infty} \right)
    + \log\left(ts /  \delta \right)  \right)
+ \log{t} \log\left( t / \deltahat \right) \right)\\
+
O\left( t \log{t} \cdot sk \cdot
\left( \log \left( 1 + \normi{B}_{\infty} \right)
    + L_{B} + \log{t} \log\left( t / \deltahat \right) \right) \right)\\
\leq
O\left( t \log{t} \cdot \left( s^2 + sk \right) \cdot
\left( \log \left( \left( 1 + \normi{X}_{\infty} \right)
\left(1 + \normi{B}_{\infty} \right)
ms/ \delta \right) + L_{B} + \log{t} \log\left( t / \deltahat \right) \right) \right)
\end{multline*}

The products of the corresponding $s$-by-$s$
and $s$-by-$k$ matrices involve numbers with at most
$L_{X} + O(\log(t / \deltahat) \log{t}) 
=  O(\log(t / \deltahat) \log{t})$ words and
$L_{B} + O(\log(t / \deltahat) \log{t})$ words
after the decimal point respectively.
So the computation of the product involves numbers of length at most
\[
O\left(\log\left(
 \left( 1 + \normi{X}_{\infty} \right)
\left(1 + \normi{B}_{\infty} \right)
\right)
+ L_{B} + \log(t / \deltahat) \log{t}\right).
\]
When $k \leq s$, it can be reduced to multiplying
$\lceil s / k \rceil^2$ $k$-by-$k$ matrices,
for a total of $O((s/k)^2 k^{\omega}) = O(s^2 k^{\omega - 2})$.
When $k \geq s$, it reduces to $\lceil k / s\rceil$
multiplications of $s$-by-$s$ matrices, for a total
of $O((k / s) s^{\omega}) = O(ks^{\omega - 1})$.
As the larger exponent is only the larger term,
the overall running time is bounded by the max of these two.

The $i\textsuperscript{th}$ block of this matrix product
\[
\left[\Ztil_{FFT\left(t, s, \deltahat\right)} X\right]
_{\left\{i\right\}}
\cdot
\left[\Ztil_{FFT\left(t, s, \deltahat\right)} B\right]
_{\left\{i\right\}}
\cdot
\]
can also be viewed as a linear transformation of $B$,
namely the matrix
\[
\Ztil_{FFT\left(t, s, \deltahat\right), \left\{i\right\}}
X
\Ztil_{FFT\left(t, s, \deltahat\right), \left\{i\right\}}.
\]
Applying the error composition bound given
in Lemma~\ref{lem:ErrorCompose} gives that its error
to the exact operator is bounded by
\begin{multline*}
\norm{
\Ztil_{FFT\left(t, s, \deltahat\right), \left\{i\right\}}
X
\Ztil_{FFT\left(t, s, \deltahat\right), \left\{i\right\}}
-
\Zbar_{FFT\left(t, s, \deltahat\right), \left\{i\right\}}
X
\Zbar_{FFT\left(t, s, \deltahat\right), \left\{i\right\}}
}_{F}\\
\leq
100
\deltahat
\max\left\{1, \norm{X}_{2}\right\}
\leq
100
\deltahat \left( ts \right)^2
\left( 1 + \normi{X}_{\infty} \right).
\end{multline*}
Summing over all $0 \leq i < t$ gives a total error
that's larger by a factor of $(ts)$.

Composing this matrix once more against the inverse
FFT given by Lemma~\ref{lem:StableFFT} then gives the
overall error bound.
Specifically, for an overall operator error of $\delta$, it suffices to set
\[
\deltahat
\leftarrow
\frac{\delta}{\left( 1 + \normi{X}_{\infty}\right) t^2 s^2}.
\]
This means each FFT/inverse FFT increases the number
of digits behind decimal place by at most
\[
O\left( \log{t} \log\left(t  / \deltahat \right) \right)
=
O\left( \log t \log\left(
\left( 1 + \normi{X}_{\infty} \right)
ts / \delta \right) \right).
\]
additively, which in turn gives the bound on the number of
digits after decimal place for $\Ztil_{FFT(t, s, \delta)}$.

Incorporating the size of $B$, we get that
the total number of words after decimal point that
we need to track is
$O(L_B + \log{t} \log( (1 + \normi{X}_{\infty} ) ts /
    \delta ))$
On the other hand,
none of the intermediate matrices during the convolution
have magnitude more
than $(ts)^2 \normi{X}_{\infty} \normi{B}_{\infty}$.
This means the total number of words tracked in order
to represent all intermediate results is at most
\[
O\left(\log\left( \left( 1 + \normi{X}_{\infty} \right)
\left( 1 + \normi{B}_{\infty} \right)
ts \right)
+L_B
+ \log{t} \log\left(
\left( 1 + \normi{X}_{\infty} \right)ts / \delta \right)\right).
\]
So combining this with the cost of matrix multiplication,
and the $O(t\log{t})$ block rearrangements of the FFT steps gives a total cost upper bound of
\[
O\left( t \log^2{t}
\cdot \max\left\{s^2k^{\omega - 2}, s^{\omega - 1}k \right\}
\cdot \left( \log\left(
  \left( 1 + \normi{X}_{\infty} \right)
  \left( 1 + \normi{B}_{\infty} \right)
  ts / \delta \right) + L_B \right) \right)
\]
word operations.
\end{proof}

%% file: Notations.tex
\section*{Notation}

\begin{table}[H]
\begin{center}
\begin{tabular}{c|l}
$A$ & matrix to be solved\\
$b$, $B$ & right-hand side vector(s) of linear equations\\
$n$ & size of matrix\\
$m$ & number of Krylov space steps\\
$s$ & size of Krylov space blocks\\
$r$ & number of remaining columns `padded' to block-Krylov space,\\
& also number of columns of null space basis we work with.\\
$K$ & block-Krylov space\\
$\textsc{Solve}_{A}(\cdot)$
& solve procedure for matrix $A$ that takes as input $b$, and outputs
$x \approx A^{-1}b$\\
$Z_{\textsc{Alg}}$ & linear operator corresponding to $\textsc{Alg}$\\
$Q$ & full $n$-by-$n$ matrix formed by appending columns a Krylov space matrix\\
$g$ & Gaussian vector\\
$G$ & Gaussian matrix\\
$g^{S}$, $G^{S}$ & sparse vector/matrix with non-zeros set to Gaussians\\
$[\cdot]$ & subset of $m$ columns in Krylov space corresponding to a single initial vector.\\
$\{ \cdot \}$ & subset of $s$ columns in Krylov space corresponding to a particular power.\\
I & identity matrix\\
$U$, $W$ & orthonormal basis\\
$d$ & dimension of a basis\\
$u$, $w$ & unit vectors\\
$\lambda$ & eigenvalues\\
$v$ & eigenvectors\\
$\Lambda$ & diagonal matrix containing eigenvalues\\
$\sigma$ & singular values\\
$\Sigma$ & diagonal matrix containing singular values\\
$\kappa$ & condition number\\
$\epsilon$ & granularity of $\epsilon$-nets\\
$\alpha$ & non-degeneracy size\\
$C, \Cbar$ & remaining/eliminated split of coordinates in block Gaussian elimination\\
$X, Y$ & matrices that form low rank approximations, and tall-and-thin matrices\\
$Z$ & matrices/linear operators that correspond to linear algorithms\\
$\cdot_{S}$, $\cdot_{J}$ & subset of columns of matrices\\
$\norm{\cdot}_2$ & $2$-norm of vectors\\
$\norm{\cdot}_{F}$ & Frobenius norms of matrices\\
$\normi{\cdot}_{\infty}$ & entry-wise max magnitude of matrices\\
$A-$ & $A$ with last row/column removed (only in Lemma~\ref{lem:Interlacing})\\
$\theta$ & rescaling/renormalizing coefficients
\end{tabular}
\end{center}

\caption{Variable Names and Meaning}
\label{table:Notations}
\end{table}